\documentclass[aps, prx, reprint, superscriptaddress, notitlepage]{revtex4-2}
\usepackage{wrapfig}

\pdfoutput=1

\usepackage{header}
\begin{document}

\title{Dissipative continuation for ground-state preparation at chemical transition states}

\author{Thomas W. Watts}
\email{thomas.watts@student.uts.edu.au}
\affiliation{HRL Laboratories, LLC, Malibu, CA, USA}
\affiliation{Centre for Quantum Software and Information, %
School of Computer Science, Faculty of Engineering \& Information Technology, University of Technology Sydney, NSW 2007, Australia}

\author{Soumya Sarkar}
\email{soumya.sarkar@student.uts.edu.au}
\affiliation{Centre for Quantum Software and Information, %
School of Computer Science, Faculty of Engineering \& Information Technology, University of Technology Sydney, NSW 2007, Australia}

\author{Daniel Collins}
\affiliation{Centre for Quantum Software and Information, %
School of Computer Science, Faculty of Engineering \& Information Technology, University of Technology Sydney, NSW 2007, Australia}

\author{Nam Nguyen}
\affiliation{Applied Mathematics, Boeing Research \& Technology, Huntington Beach, USA}

\author{Luke Quezada}
\affiliation{Media Arts + Practice Division, University of Southern California, Los Angeles, CA 90089, USA}

\author{Michael J. Bremner}
\affiliation{Centre for Quantum Software and Information, %
School of Computer Science, Faculty of Engineering \& Information Technology, University of Technology Sydney, NSW 2007, Australia}

\author{Samuel J. Elman}
\email{samuel.elman@uts.edu.au}
\affiliation{Centre for Quantum Software and Information, %
School of Computer Science, Faculty of Engineering \& Information Technology, University of Technology Sydney, NSW 2007, Australia}

\begin{abstract}
    Simulating chemical reactions exhibits a pronounced unevenness in computational difficulty: while equilibrium reactant and product geometries are often tractable, transition-state (TS) geometries frequently display strong multi-reference character that challenges both classical solvers and coherent quantum state-preparation methods. We introduce a dissipative continuation protocol for preparing electronic ground states near TS geometries within a hybrid classical--quantum workflow. In the intended setting, classical electronic-structure methods supply an approximate TS geometry, a computationally motivated continuation path, and a locally compatible active-space representation along that path. Starting from a warm start at a tractable geometry on the same aligned path, the quantum routine transports the state toward the TS using orbital-gauge-aligned Hamiltonians and engineered dissipative cooling primitives that repeatedly contract population into the instantaneous low-energy sector.
    We prove that, for continuation paths satisfying a Lipschitz smoothness condition and a localized Eigenstate Thermalization Hypothesis (ETH)-motivated downward-drift condition within the relevant energy window, the ground state at the target geometry can be prepared to total energy error $\epsilon_E$ with total ideal cooling-step complexity $\widetilde{O}(C_{\mathrm{DK}}^2 N_o^2 / \epsilon_E)$. Equivalently, for an intensive target error $\epsilon_{\rm int}\coloneqq \epsilon_E/N_o$, the ideal cooling-step complexity is $\widetilde{O}(C_{\mathrm{DK}}^2 N_o / \epsilon_{\rm int})$. Here $C_{\mathrm{DK}}$ quantifies ground-state rotation along the aligned path. The corresponding logical gate count is obtained by multiplying this primitive count by the cost of implementing one dissipative step, which is polynomial in the block-encoding size of the Lindbladian under standard Lindblad-simulation algorithms. This identifies a structured regime in which dissipative continuation provides a conditional route to ground-state preparation at strongly correlated TS geometries.
\end{abstract}
\maketitle

\section{Introduction}

\begin{figure*}[ht]
    \centering
    \includegraphics[width=.99\textwidth]{figures/OS_final_feb_2_tryptich.pdf}
    \caption{\textbf{(a)} The transition digraph associated with the  Lindbladian-induced Markov process $P$ depicting the longest canonical path between an excited state and the ground state; this path length provides a diagnostic of downhill connectivity for the displayed kernel, rather than by itself proving rapid mixing. The underlying transition graph is a numerical example of H$_4$ molecule at its multi-reference square geometry (see \cref{fig:h4_pes,fig:h4_numerical}); the transition graph has the same number of nodes as the FCI dimension $D = \binom{2N_o}{N_e} = 70$ for a $\text{CAS}(4,4)$ active space in the STO-3G orbital basis. The values above each arrow are the transition probabilities from $|E_j\rangle \to |E_k\rangle$, $P_{j \to k}$ along the longest path; in this case the longest path is $|E_{D-1}\rangle \to |E_{0}\rangle$ with length $\ell = 4$ whose transition probabilities are bounded below by $P_{j \to k} \geq 0.01$. \textbf{(b)} An example of a Ru-based carbon-capture catalyst reaction VIII $\to$ VIII–IX (transition state) $\to$ IX of the multi-step reaction mechanism for CO$_2$ reduction~\cite{vonBurg2021}; the 2D potential energy surface (PES) is representative of a more complex, multi-dimensional mechanism underlying the reaction. \textbf{(c)} The logical quantum circuit for implementing the \textit{purely dissipative} map $e^{t \mathcal{L}}$ that sends the approximate ground state $\rho_i$ and geometry $\mathbf{R}(s_i)$ to the next ground state along the reaction path $\rho_{i+1}$ where circuit $U_H$ denotes the block encoding of a chemical Hamiltonian which is used to implement the logical circuit described in~\cref{sec:logical_circuit}.}
    \label{fig:main_fig}
\end{figure*}

Accurate electronic ground states along a reaction path are essential for predicting chemical reactivity and thermodynamic observables. On a typical potential energy surface (PES), reactants and products correspond to metastable basins of attraction: collections of nuclear geometries that relax under energy minimization to distinct local minima representing chemically identifiable species, and are often well approximated by tensor-network solvers and other classical methods. In contrast, transition-state (TS) geometries lie near the dividing surface between these basins, where bonds are formed and broken and electron--electron correlations become strong. In this regime, classical solvers are often unreliable or prohibitively expensive \cite{vonBurg2021,thomas2015, kanungo2023}. As a result, reaction paths typically exhibit a highly uneven computational profile, with most geometries being classically tractable while the TS region forms a computational bottleneck.

This bottleneck is consequential. Chemical accuracy ($\sim 1~\mathrm{kcal/mol} \approx 1.6~\mathrm{mHa}$) sets the bar for the performance---and thereby the quantum resources---required to determine ground-state energies along a continuation path. In chemical kinetics, the \textit{rate constant} $k(T)$ is the temperature-dependent proportionality factor that relates reactant concentrations to the reaction rate; it provides a direct, quantitative link between microscopic energetics and macroscopic observables such as reaction timescales, selectivity, and temperature dependence. Because $k(T)$ is exponentially sensitive to the activation barrier, quantitative kinetics requires barrier heights at approximately chemical accuracy. In particular, within an Arrhenius or transition-state approximation,
\begin{equation}
    k(T) \propto \exp\left(-\frac{\Delta E^\ddagger}{k_B T}\right),
\end{equation}
where $\Delta E^\ddagger = E_{\mathrm{TS}}-E_{\mathrm{R}}$ is the activation barrier between the reactant minimum and the transition state. At room temperature, $k_B T \approx 0.6~\mathrm{kcal/mol}$, so an error of $1~\mathrm{kcal/mol}$ in $\Delta E^\ddagger$ changes the predicted rate by a factor of $\exp(1/0.6)\approx 5$. This motivates the usual ``chemical accuracy'' target: relative energies---especially barrier heights---should be accurate to $\sim 1~\mathrm{kcal/mol}\approx 1.6~\mathrm{mHa}$.

Quantum computers offer the potential to provide high-accuracy approximations to ground states in strongly correlated regimes through algorithms such as quantum phase estimation (QPE)~\cite{Kitaev1995AbelianStabilizer, ge2019faster, mande2023tight, AspuruGuzik2005MolecularEnergies, Whitfield2011ElectronicStructure, OMalley2016PRXQPE, Kassal2009PropertiesGeometry, Kassal2008ChemicalDynamics, Lanyon2010NatChemH2} and quantum imaginary time evolution (QITE)~\cite{Motta2020-qite}. Recent progress on engineered open-system protocols has also enabled thermal and ground-state preparation for chemical Hamiltonians~\cite{ding2024single,ding2025End}. In particular, \citet{ding2024single} provides detailed analysis of cooling subroutines that we build on here by applying them to transition-state chemistry and modifying the family of filter functions used. Across these approaches, however, performance is constrained by the availability of a useful initial state, or ``warm start,'' with nontrivial overlap with the target ground state~\cite{Cade2023, Gharibian2023, Waite2025Physical, Chen2025}.

A standard fault-tolerant quantum chemistry workflow determines approximate geometries and active spaces using classical electronic-structure methods, and then applies QPE or related routines to a classically computed initial state to refine electronic energies at selected geometries~\cite{Reiher_2017}. Near strongly correlated TS geometries, however, a classically prepared ansatz optimized at a fixed geometry may have very small overlap with the true ground state. Direct QPE-based preparation can then incur an expected overhead scaling as $1/\sqrt{p_0}$ when amplitude amplification is used, or $1/p_0$ under plain repetition, where $p_0=|\langle\phi|E_0\rangle|^2$ is the squared overlap between the input state $|\phi\rangle$ and the target ground state $|E_0\rangle$. This motivates path-following strategies that maintain overlap by transporting a warm start through locally smooth Hamiltonian changes, rather than preparing the TS ground state independently from scratch.

In many organic reaction workflows, DFT and related approximate methods are effective tools for locating stationary points and obtaining useful reaction mechanisms, and suitably chosen functionals can also give accurate geometries for many transition-metal complexes~\cite{Gruden2017DFTBOrganicBarriers,Liu2025DFTKinetics,Vermeeren2022PericyclicDFT,Buhl2006TMGeometries,ZhaoTruhlar2008M06Suite}. However, this classical preprocessing step is not assumed to be infallible. Transition-metal systems, spin-state crossings, and strongly correlated active spaces remain sensitive to functional choice, orbital optimization, and electron correlation~\cite{Nipanutiyan2026SpinCrossingDFT,Tzeli2024Fe2S2Correlation}. The theorem below therefore assumes a chemically validated continuation path and active-space representation as input, rather than claiming to solve transition-state search or active-space selection.

In this work, we present \textit{dissipative continuation}, a quantum state-preparation primitive for strongly correlated TS geometries within a hybrid classical--quantum reaction workflow. The classical part of the workflow supplies an approximate TS geometry, a chemically meaningful continuation path, and a locally compatible active-space representation along that path; the quantum part uses this information to transport a warm start from a tractable geometry toward the TS by applying dissipative cooling primitives at successive geometries. The continuation path need not be unique: it may be chosen as a segment of an intrinsic reaction coordinate (IRC) or another minimum-energy path (MEP), a converged string-method or nudged-elastic-band (NEB) approximation to an MEP, a one-sided downhill path from a validated transition state toward the reactant or product basin, or, more generally, any chemically justified smooth path through the relevant region of configuration space. The flexibility lies in choosing a continuation path and orbital gauge that make the state-preparation problem well conditioned.

The intended workflow is therefore:
(i) use classical electronic-structure methods to identify a TS candidate or chemically meaningful continuation path;
(ii) construct a TS-compatible active-space envelope or sequence of corresponding active spaces along that path;
(iii) align the residual orbital gauge, for example by Procrustes rotations, while preserving the selected active-space content;
(iv) choose a tractable geometry on the aligned path where a warm start can be prepared; and
(v) run dissipative continuation forward along the same aligned representation toward the TS.
Step (v) is the quantum state-preparation algorithm analyzed in this work; the remaining steps constitute classical preprocessing and validation inherited from standard computational-chemistry workflows. In particular, we do not assume that an active space chosen at the reactant geometry remains valid and unchanged through the TS: the compatible representation is part of the classically supplied input to the continuation procedure.

Our main result quantifies the resulting advantage. When the Hamiltonians along the selected continuation path admit a locally compatible active-space representation and the ground-state projector varies smoothly in that representation, TS ground states can be prepared with polynomial overhead relative to the cost of preparing a warm start at a tractable geometry on the same path. The relevant smoothness is quantified by a Davis--Kahan path constant $C_\text{DK}$, which measures how rapidly the ground-state subspace rotates along the aligned Hamiltonian path. For total energy error $\epsilon_E$, the ideal dissipative cooling-step count scales as $\widetilde{O}(C_\text{DK}^2 N_o^2/\epsilon_E)$. Equivalently, for intensive energy error $\epsilon_{\rm int}\coloneqq \epsilon_E/N_o$, this is $\widetilde{O}(C_\text{DK}^2 N_o/\epsilon_{\rm int})$, and under the coarse estimate $C_{\mathrm{DK}}=O(N_o/\Delta_{\min})$ it becomes $\widetilde{O}(N_o^3/(\Delta_{\min}^2\epsilon_{\rm int}))$ ideal cooling steps, before multiplying by the per-step Lindbladian implementation cost. This avoids the pointwise-preparation bottleneck in regimes where $p_0$ becomes exponentially small in the multireference TS window, with a cost polynomial in system size, yielding a conditional advantage whenever a smooth compatible path exists.

Although our protocol transports a state along a continuation path, it should not be viewed as a dissipative analogue of adiabatic continuation. Adiabatic approaches rely on coherent evolution and require suppressing diabatic transitions globally along the path, leading to worst-case runtime bounds governed by small spectral gaps. In contrast, our method performs a sequence of local dissipative stabilization steps that repeatedly re-concentrate population into the instantaneous low-energy manifold. The role of the continuation path is therefore not to define a slow physical evolution, but to ensure geometric smoothness so that each cooling stage begins from a warm start with bounded excess energy. Algorithmic efficiency arises from mixing-time properties of the engineered Lindbladian rather than from adiabatic gap protection.

There are two distinct sources of error in this workflow: error in the electronic energy at a specified geometry, and error in the location of the TS geometry itself. The present algorithm addresses the first; we do not claim to replace classical transition-state search, reaction-path validation, or active-space selection. A complete quantum-assisted TS refinement algorithm---which would combine our state preparation with quantum-computed nuclear forces for iterative geometry updates~\cite{OBrien2019EnergyDerivatives,OBrien2022MolecularForces}---is beyond the present scope; our theorem concerns state preparation and energy refinement for a specified reaction channel.

Within this scope, the approach targets the uneven computational profile of chemistry, and is most beneficial when (a) some region of the continuation path admits reliable classical or hybrid preparation, (b) the TS window is strongly multireference so that pointwise overlaps can become extremely small, (c) the selected path admits a locally compatible active-space representation with gradually rotating low-energy projectors, and (d) the dissipative primitive mixes rapidly in the relevant energy window. In this overlap-limited regime, dissipative continuation replaces independent pointwise preparation attempts with a sequence of locally stabilized warm-start updates, yielding polynomial overhead beyond warm-start preparation under the assumptions stated below~\cite{Lee_2023}. Conversely, if a globally good guiding state exists at all geometries, if no compatible active-space representation can be constructed along the selected continuation path, or if the path crosses an abrupt electronic rearrangement whose sharp wavefunction changes violate the smoothness and rapid-mixing conditions, then pointwise methods, coherent transport, a larger active-space envelope, or a different continuation path may be preferable.

Our approach is most directly related to the recent warm-start strategy of \citet{puig2026warmstartscoldstates}, which also tracks a discretized Hamiltonian deformation along a path through configuration space, using the solution at each Hamiltonian to initialize the next. The key distinction is mechanistic: their path tracking applies \emph{variational} ground-state preparation at each step and argues that such tracking improves trainability when the path avoids gap closings, whereas we apply an engineered \emph{dissipative} cooling primitive whose efficiency derives from Lindbladian mixing times rather than from variational trainability or adiabatic gap protection. More broadly, our work complements recent progress on dissipative, thermal, open-system, and hybrid strategies for chemical state preparation along continuation coordinates, including constructions that improve warm starts, enable controlled access to excited manifolds, and implement Lindbladian or beyond-unitary dynamics for quantum chemistry~\cite{ding2025End,AminLoveTruncik2008ThermallyAssisted,ZhouWu2026EngineeredDissipation,Zhan2026RapidDissipative,Lin2025DissipativeReview,LiLin2025ExcitedStates,Marthaler2026BeyondUnitary}. In this emerging picture, equilibrium-quality classical or hybrid preparation supplies high-fidelity states at tractable geometries, while path-based dissipative continuation provides a mechanism to carry those warm starts through the TS bottleneck where classical guiding states are least reliable.

The remainder of the paper is organized as follows. In~\Cref{sec:reaction-path-theorem}, we identify the regime in which dissipative continuation requires $O(\mathrm{poly}(N_o))$ quantum resources, and show how continuation-path conditioning can reduce the Davis--Kahan smoothness cost without replacing classical transition-state search. In~\cref{sec:resource-estimates} we provide an implementation of our algorithm, {we present a set of numerical models to validate the algorithm and provide evidence for the ubiquity of the ETH-like regime}, and present resource estimates associated with running our algorithm on real-world chemical problems for fault-tolerant quantum computers. In~\cref{sec:discussion}, we discuss how our algorithm fits into the ongoing pursuit of developing practical fault tolerant quantum algorithms for studying chemical reactions. In Appendices~\ref{sec:section1} and~\ref{sec:thm2}, we provide proofs of theorems introduced herein, in addition to the methodology behind our resource estimation for practical quantum circuit implementations of our algorithm in Appendices~\ref{sec:resource_estimation_method} and~\ref{sec:logical_circuit}. Additionally, we establish our algorithm within the broader landscape of ground-state preparation methods through a detailed comparative analysis in Sections~\ref{sec:das_vs_diss} and~\ref{sec:pr_dc_vs_dissip}.

\section{Dissipative Continuation}
\label{sec:reaction-path-theorem}

In this section, we present \textit{dissipative continuation}, and the conditions under which it may be used to prepare the ground state of a chemical Hamiltonian with a target geometry.
The efficiency proof has two parts: a smoothness bound that controls how the ground state changes as the geometry is varied along the path, and a Markov-process bound that controls how many cooling steps are needed at each discretized geometry.

\subsection{Definitions and Assumptions}
We begin by introducing the key definitions we require to analyze the complexity of the dissipative continuation algorithm.
Our analysis is formulated for a one-dimensional continuation path $\mathbf{R}:[0,1]\to \mathbb{R}^{3N_{\mathrm{nuc}}}$ through nuclear-configuration space. The continuation path is supplied by the surrounding classical workflow and may be a segment of an IRC or another MEP, including a discrete MEP approximation obtained using the string method or NEB. We do not require that $s=0$ and $s=1$ correspond to reactant and product minima. Instead, we choose the parameterization so that $s=s^*$ denotes the target TS geometry, while $s=0$ denotes a tractable geometry on the same aligned continuation path from which a high-quality warm start is available. The transition-state geometry $\mathbf{R}_{\mathrm{TS}}\coloneqq \mathbf{R}(s^*)$ is supplied or refined by classical electronic-structure methods and corresponds, in the idealized setting, to a first-order saddle point of the PES (\cref{fig:main_fig}b).

\begin{definition}[Continuation-path Hamiltonian]
    Let ${\mathbf{R}:[0,1]\to \mathbb{R}^{3N_{\mathrm{nuc}}}}$ be a continuation path and define ${H(s) := H_{\mathrm{el}}(\mathbf R(s))}$.
    Let $\{\ket{E_j(s)},E_j(s)\}_j$ denote the energy-ordered eigenstates and eigenvalues of $H(s)$, with unique ground state $\ket{E_0(s)}$, and spectral gap at $s$ ${\Delta(s) =  E_1(s)-E_0(s)>0}$.
    We denote by $\Delta_{\min}$ the minimum spectral gap over the selected continuation path.
\end{definition}

Although the map $\mathbf{R}\mapsto H_{\rm el}(\mathbf{R})$ is smooth in the nuclear coordinates, its second-quantized representation can exhibit apparent discontinuities if orbitals are chosen independently at each geometry \cite{BensbergReiher2023CorrespondingActiveSpaces}. We separate two issues. First, a chemically appropriate active-space envelope, or a sequence of corresponding active spaces, must be selected along the continuation path. Automated and corresponding active-space selection methods along continuation coordinates provide a practical means of constructing such compatible active spaces~\cite{SteinReiher2016AutoCAS,BensbergReiher2023CorrespondingActiveSpaces}. Second, after the relevant active spaces have been embedded in a common representation, arbitrary orbital-gauge discontinuities are removed by the Procrustes alignment procedure described by~\cref{prop:gs-overlap-principal-angles}. The smoothness assumptions below are imposed on the aligned family of embedded Hamiltonians, rather than on an active space chosen at the reactant geometry and held fixed along the continuation path. We therefore treat active-space compatibility as an explicit input condition, not as an automatic consequence of orbital alignment. If this condition fails---for example, because the relevant active subspaces cannot be matched continuously along the path or because a compatible active-space envelope becomes prohibitively large---our efficient preparation guarantee need not hold.

For this aligned, locally compatible family of embedded Hamiltonians, we impose the following Lipschitz continuity condition:
\begin{definition}[Lipschitz continuation path]
\label{def:regular-path}
    Let
    \[
H(s)=\sum_{\ell=1}^{M} h_\ell(s),
\]
be an aligned, locally compatible family of embedded Hamiltonians with a fixed and consistently indexed decomposition into $M$ fragments. A continuation path is said to be \textit{Lipschitz continuous} if, in the aligned and locally compatible representation, for all $\ell\in [M]$ and all $s,s'\in[0,1]$,
    \begin{equation}
        \norm{h_{\ell}(s)-h_{\ell}(s')}\le L_h|s-s'|,
        \label{eq:lipschitz-fragments-main}
    \end{equation}
    where $L_h>0$ is the associated Lipschitz constant.
\end{definition}

In particular, the smoothness condition enforces the following bound ${\| H(s) - H(s') \| \le M L_h |s-s'|}$. To relate this operator smoothness to the motion of the ground state along the continuation path, we additionally assume the Hamiltonian fragments are piecewise differentiable in the chosen local representation, i.e., each fragment $h_\ell(s)$ is piecewise $C^1$. Then 
\begin{equation}
    \label{eq:fragment_bound}
    \| \partial_s H(s)   \| \le \sum_{\ell=1}^{M} \| \partial_s h_\ell (s) \|.
\end{equation}
To quantify the smoothness of a continuation-path Hamiltonian in the aligned representation, we introduce the Davis--Kahan constant:
\begin{definition}[Davis--Kahan constant]
    \label{def:dk_constant}
       Let $\{H(s)\}_{s\in[0,1]}$ be the aligned, locally compatible Hamiltonians associated with a Lipschitz continuation path as in~\cref{def:regular-path}. The associated geometric Davis--Kahan constant is defined as
        \begin{equation}
        C_{\mathrm{DK}}\coloneqq\int_0^1\frac{\|\partial_s H(s)\|}{\Delta(s)} \, \mathrm{d}s.
            \label{eq:C_DK-def-appendix}
        \end{equation}
\end{definition}

The Davis--Kahan constant $C_\text{DK}$ is analogous to the ``adiabatic condition'' appearing in adiabatic ground-state preparation~\cite{Boixo2009Eigenpath,Berry2025Large}, although here it is used as a geometric rather than dynamical quantity. The local ratio $\|\partial_sH(s)\|/\Delta(s)$ controls how rapidly the instantaneous ground-state projector $P_0(s) = \ketbra{E_0(s)}$ can vary with $s$, and $C_{\mathrm{DK}}$ bounds its cumulative rotation along the continuation path. Consequently, if $C_{\mathrm{DK}}$ is small, the ground state does not rotate too rapidly along the path. In particular,
\begin{equation}
C_{\mathrm{DK}}
\leq
\frac{1}{\Delta_{\min}}
\sum_{\ell=1}^{M}
\int_0^1
\|\partial_s h_\ell(s)\|
\,\mathrm{d}s
\leq
\frac{ML_h}{\Delta_{\min}}.
\end{equation}
Thus, when $M=O(N_o)$ and $L_h=O(1)$, one obtains the coarse scaling $C_{\mathrm{DK}}=O(N_o/\Delta_{\min})$ used below.

Given the freedom to choose the continuation path, it can be optimized to reduce the algorithmic runtime. To make this optimization precise, we require a geometric measure of how sensitively the electronic Hamiltonian responds to infinitesimal changes in the nuclear coordinates. This sensitivity is captured by the operator-Lipschitz constant.

\begin{definition}[Operator-Lipschitz constant]
\label{def:op_lipschitz}
    Let $\mathbf{R}:[0,1]\to\mathbb{R}^{3N}$ be a $C^1$ continuation path with Fr\'echet derivative $D_{\mathbf{R}}H_{\mathrm{el}}(\mathbf{R}) : \mathbb{R}^{3N}\to \mathsf{Herm}$. Define the pointwise operator norm of the Fr\'echet derivative
    \[ L_H(\mathbf{R}) := \sup_{\|v\|=1}\bigl\|D_{\mathbf{R}}H_{\mathrm{el}}(\mathbf{R})[v]\bigr\|,\]
    which is bounded by an operator-Lipschitz constant $L_{\max}$,
    \[
        L_{\max} := \sup_{s\in[0,1]} L_H(\mathbf{R}(s)).
    \]
\end{definition}

Let $\rho$ be a density operator. In the instantaneous energy eigenbasis ${\ket{ E_i(s)}_i}$, it induces a population distribution
\[
    \mu(i)=\bra{E_i(s)}\rho\ket{E_i(s)}.
\]
Fix an energy-layer width $\Delta E_{\text{layer}}>0$, and define the layer index of eigenstate $j$ by
\begin{equation}
    L(j)\coloneqq \left\lceil\frac{E_j-E_0}{\Delta E_{\text{layer}}}\right\rceil,
    \label{eq:layer_index}
\end{equation}
which partitions the spectrum into layers of equal energy thickness. We say that $\mu$ is $\Delta E_{\text{layer}}$-concentrated if there exist $L_\mu\in\mathbb{N}$ and $\eta\in(0,1)$ such that
\begin{equation}
    \Pr_{j\sim\mu}(L(j)>L_\mu)\le\eta.
\end{equation}
With these definitions, the population dynamics induced by a dissipative primitive can be reduced to a classical Markov chain $(J_t)_{t\ge 0}$ on the eigenstate labels.
Following~\cite{ding2024single, LevinPeresWilmer2009}, the transition kernel is a row-stochastic matrix $P$ mapping the pre-primitive distribution $\mu$ to the post-primitive distribution $\mu'$, i.e., $\mu'=\mu P$. For a formal definition, see~\cref{app:preliminaries}. The matrix $P$ can also be viewed as an edge-weighted directed graph (see~\cref{fig:main_fig}a).

We assume the ground state corresponds to the target state $0$, and define the hitting time of state $0$ (ground state) by 
\begin{equation}
    \label{eq:hitting_time}
    \tau_0 = \inf \{ t \in \mathbb{N}: J_t = 0 \}.
\end{equation}
This hitting time controls the number of dissipative steps required to reach the ground state up to a target failure probability $\varepsilon$. 
\begin{definition}[Dissipative Continuation Mixing time]
Formally, for $ 0 < \varepsilon < 1$, the (warm-start) mixing time is 
\begin{equation}
\label{eq:tmix}
    t_{\mathrm{mix}}(\varepsilon,\mu)
    \coloneqq
    \inf\Big\{t\in\mathbb{N}_0:\ \Pr_{\mu}(\tau_0>t)\le \varepsilon\Big\},
\end{equation}
\end{definition}

Equivalently, $t_{\mathrm{mix}}$ is the smallest number of dissipative steps sufficient to reach the ground state with probability at least $1 - \varepsilon$ when initialized from the warm-start distribution $\mu$ induced by $\rho$. Therefore, once we can obtain a bound on $t_{\mathrm{mix}}(\varepsilon,\mu)$ under the assumption above, we can provide a bound on the total number of dissipative steps required by the algorithm. Combined with the cost of implementing each step, this yields an overall runtime bound, which we will show in~\cref{subsec:efficientGSP} is polynomial in the problem size (i.e. $N_o$). In the next section we compare our approach to popular alternatives and provide logical resource estimates for implementing our dissipative state-preparation primitive.

\subsection{Efficient Ground State Preparation at Transition State Geometries}
\label{subsec:efficientGSP}

In this section, we demonstrate how dissipative continuation enables efficient quantum state preparation within the hybrid workflow described above. The classical part of the workflow supplies a target TS geometry, a continuation path, and a locally compatible active-space representation along that path. The quantum part begins at a tractable geometry on the same aligned path, where a high-quality warm start can be prepared classically or by a hybrid method, and transports this state toward the TS. At the transition-state geometry, $H(s^*)$ may be too strongly correlated for classical methods to provide a high-overlap state directly. Thus, our goal is to design a quantum algorithm that prepares a state $\rho$ such that
\[
    \bigl|\mathrm{tr}(H(s^*)\rho) - E_0(s^*)\bigr| \le \epsilon_E
\]
where $E_0(s^*)$ is the ground-state energy of $H(s^*)$ in the chosen locally compatible representation.

\begin{theorem}[Continuation-path ground-state preparation; informal]
    \label{thm:reaction-path-runtime_informal}
    Given a Lipschitz-continuous sequence of locally compatible, orbital-gauge-aligned Hamiltonians $H_{\rm el}(\mathbf R(s))$ with $N_o$ active spatial orbitals, a target TS geometry $s^*$, and a warm start at a tractable geometry on the same continuation path, the ground state at $s^*$ can be prepared with total energy error $\le \epsilon_E$ using
    \[
       \widetilde{O}\left( \frac{C_\text{DK}^2 N_o^2}{\epsilon_E} \right)
    \]
    ideal dissipative cooling steps. Equivalently, for intensive energy error $\epsilon_{\rm int}\coloneqq \epsilon_E/N_o$, the ideal step count is
    \[
       \widetilde{O}\left( \frac{C_\text{DK}^2 N_o}{\epsilon_{\rm int}} \right).
    \]
    Under the coarse estimate $C_{\mathrm{DK}}=O(N_o/\Delta_{\min})$, the intensive-error version becomes $\widetilde{O}(N_o^3/(\Delta_{\min}^2\epsilon_{\rm int}))$ ideal cooling steps, before multiplying by the per-step Lindbladian implementation cost.
\end{theorem}

\cref{thm:reaction-path-runtime_informal} is reached by combining the number of continuation points required by the Davis--Kahan smoothness bound with the number of cooling rounds required by the local downward-drift bound. The quantity $\widetilde{O}(C_{\mathrm{DK}}^2N_o^2/\epsilon_E)$ should be read as an ideal cooling-step complexity for total energy error, or equivalently $\widetilde{O}(C_{\mathrm{DK}}^2N_o/\epsilon_{\rm int})$ for intensive energy error. The corresponding logical gate count is obtained by multiplying this primitive count by the cost of implementing each dissipative step; this cost is analyzed in~\cref{sec:resource-estimates} using standard Lindblad-simulation subroutines. The key structural feature we exploit is that along a physically meaningful continuation coordinate the Hamiltonian varies smoothly and the ground-state projector does not rotate arbitrarily fast.

At the algorithmic level, the classical workflow supplies a continuation path, an approximate TS geometry, and a locally compatible active-space representation along the path. This representation can be constructed around the TS region and aligned backward toward a tractable geometry. The quantum routine then proceeds forward along the same aligned path: we discretize $s\in[0,s^*]$ into a sequence of geometries and, at each step, apply a dissipative cooling primitive that amplifies the ground-state population at the new geometry, using the ground-state approximation from the previous step as input (see~\cref{fig:main_fig}b). The overall complexity is governed primarily by how rapidly the ground-state subspace varies along the aligned path, by the dissipative mixing time, and by the target energy accuracy $\epsilon_E$.

Concretely, we first assume that a chemical system of interest is accompanied by a classically generated reaction channel or TS candidate, obtained using established tools such as nudged elastic band, string methods, synchronous-transit-guided quasi-Newton methods, constrained optimization, or related transition-state search procedures~\cite{Williams2016,VaucherReiher2018MEP}. The continuation path used by our algorithm need not itself be a newly discovered physical MEP; it is an input path, or a smooth deformation of an input path, that must remain chemically validated and pass through the target TS region. The design freedom relevant to the quantum algorithm is to choose a path and representation that reduce abrupt projector rotations and improve the Davis--Kahan conditioning of state transport.

Once a continuation path and compatible active-space representation are selected, we discretize the path into $N_H$ points $\{s_i\}_{i=0}^{N_H-1}$, where $s_{N_H-1}=s^*$ is the target TS geometry and $s_0$ is a tractable geometry on the same aligned path. At each point, the classical preprocessing step builds the electronic Hamiltonian $H_\text{el}(\mathbf{R}(s_i))$ in the aligned active-space representation. A classical or hybrid method, such as DMRG in a tractable region, is then used to generate an effective warm start at $s_0$.

Geometrically, the instantaneous ground state $\ket{E_0(s)}$ traces out a smooth curve on the state manifold as $s$ varies along the aligned continuation path toward the target TS geometry. For a Lipschitz-continuous continuation path, adjacent Hamiltonians differ only slightly, so the associated low-energy projectors do not rotate abruptly between $s_i$ and $s_{i+1}$. As a result, the ground state at $s_i$ has nontrivial overlap with the low-energy subspace at $s_{i+1}$ and therefore serves as an effective warm start. Consequently, the dissipative dynamics at each step need only remove a small residual excited-state population rather than cool from a generic high-energy state. A complete description of our algorithm is provided in \cref{sec:pseudocode}.

\subsubsection*{Comparison to other Ground State Preparation Techniques}
Ground-state preparation at transition-state geometries is known to be a computationally challenging problem for which a myriad of quantum algorithms have been designed. Notable examples include quantum phase estimation~\cite{Kitaev1995AbelianStabilizer} (QPE), digital adiabatic simulation~\cite{Berry2025Large} (DAS), phase randomization~\cite{Boixo2009Eigenpath} (PR) and dynamic cooling~\cite{motlagh2024_cooling} (DC).~\cref{tab:comp_tab} presents a comparison between the dominant asymptotic state-preparation overheads of these algorithms and dissipative continuation in terms of their shared parameters: the number of orbitals $N_o$, the maximum Hamiltonian norm along the continuation path  $\| H\| \coloneqq \max_s \|H(s)\|$, and the minimum spectral gap $\Delta_{\min} \coloneqq \min_s \Delta(s)$. In terms of these parameters, dissipative continuation requires the ideal dissipative cooling-step count
\begin{equation}
\label{equ:gate_complexity}
    N_{\mathrm{cool}}^\mathrm{diss}
    = \widetilde{O}\left( \frac{C_{\mathrm{DK}}^2 N_o^2}{\epsilon_E} \right)
    = \widetilde{O}\left( \frac{C_{\mathrm{DK}}^2 N_o}{\epsilon_{\rm int}} \right),
    \qquad
    \epsilon_{\rm int}\coloneqq\frac{\epsilon_E}{N_o}.
\end{equation}
Using the coarse estimate $C_{\mathrm{DK}}=O(N_o/\Delta_{\min})$, the intensive-error version gives $N_{\mathrm{cool}}^\mathrm{diss}=\widetilde{O}(N_o^3/(\Delta_{\min}^2\epsilon_{\rm int}))$. We compare this ideal primitive-count scaling to each of the other ground state preparation algorithms below. The logical gate count is obtained by multiplying by the per-step Lindbladian implementation cost, as discussed in~\cref{sec:resource-estimates}. A more detailed discussion of the comparisons made here between these ground state preparation algorithms can be found in~\cref{sec:das_vs_diss} and \cref{sec:pr_dc_vs_dissip}.

\begin{table}[t]
    \centering
    \renewcommand{\arraystretch}{1.15}
    \setlength{\tabcolsep}{8pt}
    \begin{tabular}{ll}
        \toprule
        \textbf{Algorithm} & \textbf{Dominant scaling} \\
        \midrule
        \makecell[l]{QPE~\cite{Kitaev1995AbelianStabilizer}} &
        \makecell[l]{$
        \widetilde{O}\!\left(
        \frac{\|H\|}{\min\{\Delta_{\min},\epsilon_E\}} \cdot \frac{1}{\sqrt{p_0}} \cdot\mathrm{poly}(N_o)
        \right)$} \\ [6pt]
        DAS~\cite{Berry2025Large} &
        $\widetilde{O}\!\left(
        \frac{\|H\|^{3}}{\Delta_{\min}^{5}} \cdot\frac{1}{\epsilon_E} \cdot\mathrm{poly}(N_o)
        \right)$ \\ [6pt]
        PR~\cite{Boixo2009Eigenpath} &
        $\widetilde{O}\!\left(
        \frac{\|H\|^2}{\Delta_{\min}^{3}} \cdot\frac{1}{\epsilon_E} \cdot\mathrm{poly}(N_o)
        \right)$ \\ [6pt]
        DC~\cite{motlagh2024_cooling} &
        $\widetilde{O}\!\left(
        \frac{\|H\|}{\Delta_{\min}^{3}} \cdot\frac{1}{\epsilon_E} \cdot\mathrm{poly}(N_o)
        \right)$ \\ [6pt]
        This work &
        \makecell[l]{$
        \widetilde{O}\!\left(
        \dfrac{N_o^{4}}{\Delta_{\min}^{2}\epsilon_E}
        \right)$ ideal steps; \\ $\times$ per-step Lindbladian cost} \\
        \bottomrule
    \end{tabular}
    \caption{Dominant state-preparation scaling comparison, suppressing polylogarithmic factors via $\widetilde{O}$. Here $\|H\|$ denotes the maximum operator norm over the continuation path, and $\Delta_{\min}$ is the minimum spectral gap along the path. Rows for QPE, DAS, PR, and DC are stated for target total energy precision $\epsilon_E$. Our row displays the ideal dissipative cooling-step count in the same total-error convention; equivalently, it is obtained from the intensive bound $\widetilde{O}(C_{\mathrm{DK}}^2N_o/\epsilon_{\rm int})$ by setting $\epsilon_{\rm int}\coloneqq\epsilon_E/N_o$. $N_o$ is a problem-size parameter, e.g. the number of spatial orbitals, and $\mathrm{poly}(N_o)$ denotes method-dependent polynomial overheads such as basis changes, state-preparation subroutines, term-selection costs, Hamiltonian simulation overheads, or Lindbladian block-encoding costs. For QPE, $p_0 \coloneqq |\langle E_0(s^*)|\phi\rangle|^2$ is the initial squared overlap of the prepared state $|\phi\rangle$ with the target ground state $|E_0(s^*)\rangle$ at the hardest point $s^*$; the displayed $1/\sqrt{p_0}$ assumes amplitude amplification, while plain repetition gives $1/p_0$. For our method, the row reports the ideal dissipative cooling-step count after converting to total precision and using $C_{\mathrm{DK}}=O(N_o/\Delta_{\min})$. The logical gate count additionally includes the per-step Lindbladian implementation cost discussed in~\cref{sec:resource-estimates}.}
    \label{tab:comp_tab}
\end{table}

\paragraph*{Quantum Phase Estimation}
QPE projects an input state onto approximate energy eigenstates of a Hamiltonian by coherently correlating each eigencomponent with its eigenphase and collapsing onto the measured outcome. To present a meaningful comparison to QPE, we present the algorithmic cost for applying QPE directly at the target transition-state geometry to prepare the ground state. Typically $\Delta_{\min} \sim \Delta(s^*)$, therefore we may use $\Delta_{\min}$ to directly compare QPE with dissipative continuation. The practical limitation of QPE, however, is the need for an efficiently preparable, semi-classical guiding state $\ket{\phi}$~\cite{Waite2025Physical,Cade2023,Gharibian2023}. The fidelity of the guiding state with the true ground state,
\begin{equation}
    p_0 \coloneqq|\langle E_0(s^*) |\phi\rangle|^2 \in (0,1],
    \label{eq:qpe_overlap_pi}
\end{equation}
informs the gate complexity, therefore the comparable QPE gate complexity is
\begin{equation}
    n_G^{\mathrm{QPE}} = \widetilde{O}\left(
    \frac{\|H\|}{\min\{\Delta_{\min}, \epsilon_E\}}\cdot \frac{1}{\sqrt{p_0}} \cdot \mathrm{poly}(N_o) \right),
    \label{eq:qpe_pointwise_complexity}
\end{equation}
where the factor of $\min\{\Delta_{\min}, \epsilon_E\}$ sets the required energy resolution.

Dissipative continuation is intended for settings where constructing the guiding state $|\phi\rangle$ with large overlap with the ground-state is unavailable or unreliable. This situation commonly arises in strongly multi-reference regions~\cite{Wardzala2024,Szalay2012,Jiang2012,Sharma2021}. In these regimes the ground state is no longer well described by a single determinant or low-entanglement ansatz: several near-degenerate configurations contribute with comparable weight and natural-orbital occupations deviate significantly from idempotency as the nuclear geometry approaches a symmetry-enhanced or bond-rearrangement configuration. As a result, even chemically motivated ansätze prepared at TS geometries may have extremely poor overlap, making it challenging in practice to guarantee provably non-negligible overlap without extensive active-space optimization that may be exponentially expensive.
In such cases, QPE's cost may be dominated by the overlap-amplification factor in~\cref{eq:qpe_pointwise_complexity}.

Instead, we assume access to a warm start at an easy geometry (Hamiltonian) and transport it along a continuous evolution path by repeatedly applying a cooling primitive at intermediate geometries. This shifts the state-preparation requirement from finding a good guiding state $| \phi \rangle$ for $H(s^*)$ to starting from a good state and following the path. 
In practice, an approach combining our dissipative continuation approach with QPE could be employed to boost the overlap with $\ket{E_0(s^*)}$ at the target geometry before applying QPE for high-precision energy estimation. 

\paragraph*{Digital Adiabatic Simulation}
DAS prepares the ground state by varying a discretized Hamiltonian slowly enough to ensure that the state of system remains within the instantaneous low-energy subspace throughout the evolution.
To make a comparison to DAS, we consider the algorithmic complexity of using DAS to coherently evolve the state $\ket{E_0(s_0)}$ to $\ket{E_0(s^*)}$ by digital adiabatic simulation along the same continuation path from $H(s_i)$ to $H(s_{i+1})$ until we reach $\ket{E_0(s^*)}$.
Using the recent improvements on large-timestep DAS~\cite{Berry2025Large}, we can write the complexity of the improved algorithm as
\begin{equation}
    n_G^{\mathrm{DAS}} = \widetilde{O}\left( \frac{\|H\|^3}{\Delta_{\min}^5} \cdot \frac{1}{\epsilon_E} \cdot \text{poly}(N_o) \right).
\end{equation}
This comparison should be read as a conservative path-following estimate obtained by applying the large-step DAS scaling uniformly along the continuation mesh. A sharper comparison may be possible by using a total-evolution-time bound that tracks the full smooth interpolation directly, rather than multiplying a segment-level worst-case estimate by the number of continuation points.

Under this conservative segment-wise accounting, dissipative continuation has a milder apparent dependence on the Hamiltonian norm and minimum gap than DAS. A key distinction between DAS and dissipative continuation is how errors scale and accumulate along the path.
Digital adiabatic simulation is a coherent transport method: to keep the state within the instantaneous low-energy subspace, the schedule must suppress diabatic transitions, and existing worst-case bounds therefore exhibit a strong dependence on the minimum spectral gap along the path (e.g., the $\Delta_{\min}^{-5}$ dependence in~\cite{Berry2025Large}).
By contrast, dissipative continuation implements a \textit{contractive} channel that relaxes toward the instantaneous low-temperature state (and, in the zero-temperature limit, toward the ground-state manifold).
In the regime we target---where a warm start is available at an easy geometry and the path is sufficiently smooth---state-preparation reduces to repeatedly re-cooling at intermediate geometries.
Under the uniform downward drift / mixing assumptions used in our analysis, the number of cooling steps is controlled by a mixing-time parameter (captured in our constant $C_{\mathrm{DK}}$ and related hitting-time bounds) rather than by an adiabatic gap condition.
Accordingly, our advantage over DAS should be understood as replacing worst-case adiabatic gap dependence with a mixing-time dependence that can be substantially milder in ETH-like, locally thermalizing regions; we do not claim an unconditional improvement in $\|H\|/\Delta_{\min}$ in the worst case.

\paragraph*{Phase Randomization}
In PR, at each point $H(s)$, we apply the instantaneous Hamiltonian for a \textit{random} evolution time. Averaging over random times induces dephasing in the eigenbasis of $H(s_j)$, which approximates a projective measurement onto the instantaneous ground state (a quantum-Zeno-type stabilization), thereby allowing ground-state transport along the discretized path.
Under the standard Hamiltonian-simulation cost model (evolution for time $t$ costs $\widetilde{O}(\mathrm{poly}(N_o))$ gates), the phase-randomization method has total gate complexity,
\begin{equation}
    n_G^{\mathrm{PR}} =\widetilde{O}\left(
        \frac{\|H\|^2}{\Delta_{\min}^3} \cdot \frac{1}{\epsilon_E} \cdot \mathrm{poly}(N_o)
    \right),
\end{equation}
The complexity of the dissipative continuation matches that of the PR algorithm in minimum gap size $\Delta_{\min}$. The complexities differ in terms of the degree of the polynomial scaling with $N_o$ as well as an additional factor of $\|H\|$ arising from PR requiring a finer path discretization than dissipative continuation.

\paragraph*{Dynamic Cooling} DC achieves ground-state energy reduction via a purely unitary, closed-system procedure rather than explicit coupling to an external bath or ancilla bath degrees of freedom. In this approach, one alternates between a spectral transformation of the Hamiltonian implemented using quantum signal processing and evolution under a weak perturbation chosen to induce transitions that preferentially move populations toward lower energies in order to prepare the ground-state energy to within a target resolution $\epsilon_E$.
In this case, the asymptotic complexity of using dynamic cooling for transporting a ground state along a continuation path is
\begin{equation}
    \label{eq:dc_runtime}
    n_G^\text{DC} = \widetilde{O}\!\left(
    \frac{\|H\|}{\Delta_{\min}^{3}} \cdot\frac{1}{\epsilon_E} \cdot\mathrm{poly}(N_o)
\right).
\end{equation}

DC matches the displayed scaling of dissipative continuation with respect to $\|H\|$ and $\Delta_{\min}$ after suppressing method-dependent polynomial factors in $N_o$. The key difference is hidden in this polynomial overhead. In the dynamic-cooling accounting, the number of calls to the primitive implementing $e^{\pm iH}$ is governed by an effective dimension parameter $d$; for the active-space electronic-structure continuation regime considered here, the relevant extensive energy window gives $d=O(N_o)$. This yields $\widetilde{O}(N_o^{3/2})$ Hamiltonian-evolution calls per geometry, before the additional polynomial cost of implementing each Hamiltonian simulation primitive. Thus DC remains a natural comparator, but its system-size overhead is not the same as the ideal dissipative-step count reported for our method.

\subsection{Proof Outline and Complexity Analysis}
The proof of~\Cref{thm:reaction-path-runtime_informal} rests on three lemmas. First,~\cref{lem:dk_constant} shows that if we discretize the continuation path finely enough, then successive ground states along the grid have non-trivial overlap. Equivalently, the cumulative rotation of the ground state along the path determines how many grid points $N_H$ are needed.  Thus, the smoothness of the continuation path implies that a polynomial small number of reaction geometries are required.

Next,~\Cref{lem:eth-informal} gives a tail bound on the ground-state hitting time $\tau_0$ for the Markov chain with kernel $P$, assuming the ETH-motivated uniform downward drift condition (see~\Cref{sec:eth-shell-mixing} for a formal definition). The uniform downward drift condition reflects the ETH expectation that a local jump operator couples a typical excited eigenstate to a dense set of lower-energy eigenstates with broadly distributed matrix elements.  In such settings, convergence does not rely on the absence of local minima but rather on local smoothness and conditioning that induce a monotonic descent direction~\cite{jacot2018neural}. Although the number of accessible lower-energy states grows rapidly with energy, the typical matrix elements to each state decrease accordingly, so that the summed probability of downward energy transitions remains essentially constant.

\Cref{lem:linear_cooling_informal} then shows that for a sufficiently warm start and Lipschitz-smooth path, the required number of cooling steps scales linearly with $N_o$. Combined with the discretization bound~\Cref{lem:dk_constant} and the implementation scaling from~\cref{sec:resource-estimates}, these results yield the overall runtime scaling in~\cref{equ:gate_complexity}.

\begin{lemma}[Bounding Successive Ground-State Overlap; informal]
    \label{lem:dk_constant}
    Let $H(s)$ be a continuation path Hamiltonian with a non-degenerate ground state $|E_0(s)\rangle$. Then, for any ${0 \le s_a < s_b \le 1}$,
    \begin{equation}
    \label{eq:cdk}
        1 - \bigl|\braket{E_0(s_a)}{E_0(s_b)}\bigr|^2
        \leq
        \left(
        \int_{s_a}^{s_b}
        \frac{
        \|\partial_s H(s)\|
        }{
        \Delta(s)
        } \mathrm{d}s
        \right)^2.
    \end{equation}
\end{lemma}

\Cref{lem:dk_constant} formalizes the idea that the ground state cannot rotate arbitrarily fast along the continuation coordinates. In particular, the quantity $\|\partial_s H(s) \| / \Delta(s)$ controls the local rate at which the ground-state can rotate. To see where~\cref{eq:cdk} comes from, fix a smooth phase choice  (parallel transport) so that $\langle E_0(s)|\partial_s E_0(s)\rangle=0$. Differentiate the eigenvalue equation $H(s)\ket{E_0(s)} = E_0(s)\ket{E_0(s)}$ with respect to $s$ and project onto the excited subspace. Using the gap $\Delta(s) = E_1(s) - E_0(s)$ gives the standard bound 
\begin{equation}
    \label{eq:local_rotation_bound}
    \| \partial_s |E_0(s) \rangle \|  \leq  \frac{\| \partial_s H(s) \|}{\Delta(s)}.
\end{equation}
Geometrically, the Fubini--Study distance between the eigenstates at $s_a$ and $s_b$ is at most the path length traced by $\ket{E_0(s)}$ on $[s_a, s_b]$. Hence, 
\begin{equation}
    1-|\langle E_0(s_a)|E_0(s_b)\rangle|^2
    \leq \left( \int_{s_a}^{s_b}\|\partial_s |E_0(s) \rangle \|ds \right)^2.
    \label{eq:infidelity}
\end{equation}
Combining~\cref{eq:infidelity} with~\cref{eq:local_rotation_bound} yields~\cref{eq:cdk}. In~\Cref{sec:dk_appendix}, we provide a proof to a more general statement, of which~\cref{eq:cdk} is a special case.

\Cref{lem:dk_constant} can then be used to choose a discretization of the reaction coordinate. Let $\{s_i\}_{i=0}^{N_H-1}$ be an ordered grid with $s_0=0$ and $s_{N_H-1}=s^*$. If this grid is fine enough that the integral in~\cref{eq:cdk} over each interval $[s_{i-1}, s_i]$ is $O(\epsilon_{\rm int}/C_{\mathrm{DK}})$, then the accumulated overlap error over the mesh is at most of order $\epsilon_{\rm int}$. Thus the number of grid points is $N_H = \widetilde{O}(C_{\mathrm{DK}}^2/\epsilon_{\rm int})$, equivalently $N_H = \widetilde{O}(C_{\mathrm{DK}}^2N_o/\epsilon_E)$ for total energy error (see~\cref{sec:proof_of_thm1} for derivation).

Once a grid is fixed~\footnote{A grid is fixed by calculating $N_H$ and distributing according to the magnitude of gradient.}, it remains to bound the number of cooling steps required at each point $s_i$. Let $\rho$ denote an initial input state with an eigenbasis distribution $\mu$.
One application of the cooling primitive at $s_i$ updates populations according to a Markov kernel $P_i$, in the sense that the population distribution is updated to
\[
    \mu(k)\gets\sum_{j=0}^{D-1} \mu(j) (P_i)_{j \to k},
\]
where we view $\mu$ as a row vector. The number of $P_i$ applications required at configuration $s_i$ is determined by the hitting time $\tau_0$ defined in~\cref{eq:hitting_time}, equivalently the warm-start mixing time $t_{\mathrm{mix}}(\varepsilon,\mu)$ in~\cref{eq:tmix}. To bound this cost, we need an assumption on the cooling kernel $P_i$ that prevents the population from getting stuck in the excited states.

In the transition-state regime---where the spectral structure is most crowded and naive cooling can stall---we assume that the row-normalized transition kernel $P_i$ satisfies a \textit{uniform downward drift} condition. Concretely, fix a layer decrement $\Delta E_{\mathrm{layer}}>0$ and define, for each eigenstate index $j$ in the relevant microcanonical window, the ``successful downhill'' set
\[
    S_j \coloneqq\{k:\ E_k \le E_j-\Delta E_{\mathrm{layer}}\}.
\]
The drift condition asserts the existence of a constant $p_{\min}>0$ such that, uniformly over the indices $j$ in this window,
\begin{equation}
    \label{eq:uniform-drift-text}
    \sum_{k\in S_j} P_{j\to k} \ge p_{\min}.
\end{equation}
Equivalently, conditioned on the current eigenstate label $j$, each cooling step has probability at least $p_{\min}$ to decrease the energy by $\ge \Delta E_{\mathrm{layer}}$; this rules out exponentially rare ``useful'' downhill moves (strong bottlenecks) within the window.

At first glance,~\cref{eq:uniform-drift-text} seems like a very strong condition, since it is a \textit{uniform} (worst-case-over-$j$) lower bound rather than a typical-case statement. Two features of our setting make it plausible and, importantly, sufficient for our analysis. First, we do not invoke~\cref{eq:uniform-drift-text} from an arbitrary initialization: the previous lemma provides a warm start at step $i$ with state mass already concentrated in the low-energy sector relevant to the subsequent cooling stage. Thus the drift assumption is only required \textit{locally}---throughout the narrow energy window actually visited during the transition-state portion of the path---rather than globally across the full spectrum.

Second, in over-parameterized optimization and related dynamics, a provably close warm start often places the iterate in a regime where the effective dynamics becomes well-conditioned and admits monotone progress guarantees; see~\cite{jacot2018neural,nguyen2025Guarantees,Chen2025} and references therein for a full discussion of how closeness to a target region upgrades qualitative mixing/escape behavior into quantitative, uniform progress bounds. The downward-drift assumption mirrors the warm-start gradient-descent regime of neural tangent kernel (NTK) theory, where linearized dynamics induce a contractive flow toward lower loss without requiring the absence of local minima~\cite{jacot2018neural,Chen2025}. We explore this connection between our setting and NTK rigorously in~\cref{app:ntk_drift}.

This assumption should be interpreted as a local sufficient condition rather than a generic claim about all molecular Hamiltonians. It is motivated by ETH-type behavior in chaotic interacting many-body systems and by recent evidence that low-symmetry molecular electronic spectra can display Wigner--Dyson/GOE statistics once symmetry constraints are resolved~\cite{TaoGalitski2026AbInitioRMT}. In the present algorithm, the condition is required only in the energy window reached from the transported warm start near a selected geometry. If numerical diagnostics reveal bottlenecks, symmetry-protected sectors, or poor downhill connectivity in this window, the cooling-time theorem should not be expected to predict efficient preparation without modifying the jump set, representation, or continuation path.
For this reason, we use ETH only as physical motivation for why such transition graphs may arise in chaotic molecular regimes; the formal assumption used in the theorem is the explicit downward-drift condition in~\cref{eq:uniform-drift-text}.

\begin{lemma}[Local downward-drift cooling bound; informal]
\label{lem:eth-informal}
Let $(J_t)_{t\geq 0}$ be a Markov chain with transition kernel $P$ and initial distribution
\[
\mu(j)\coloneqq \Pr(J_0=j).
\]
Suppose that, for some $L_\mu\in\mathbb{N}_0$ and $\eta\in(0,1)$,
\[
\Pr_\mu\bigl(L(J_0)\leq L_\mu\bigr)\geq 1-\eta.
\]
Assume that $P$ satisfies the uniform monotone downward-drift condition: while $1\leq L(J_t)\leq L_\mu$, the layer index cannot increase, and each step decreases it by at least one with conditional probability at least $p_{\min}>0$. Then, for every $\varepsilon\in(0,1)$,
\begin{equation}
\Pr_\mu\left(
\tau_0>
\frac{2L_\mu+8\ln(1/\varepsilon)}
{p_{\min}}
\right)
\leq \eta+\varepsilon.
\end{equation}
Equivalently, with probability at least $1-\eta-\varepsilon$,
\begin{equation}
\tau_0
\leq
\frac{2L_\mu+8\ln(1/\varepsilon)}
{p_{\min}}.
\end{equation}
\end{lemma}

At a fixed grid point $s_i$, we apply~\cref{lem:eth-informal} to the kernel $P = P_i$ governing one application of the cooling stage at $s_i$. In this setting, the warm-start mixing time $t_{\mathrm{mix}}(\varepsilon, \mu)$ in~\cref{eq:tmix} is the number of kernel applications required to make $\Pr_\mu(\tau_0 > t) \le \varepsilon$. 
\cref{lem:eth-informal} bounds this cost by a term linear in the warm-start layer cutoff $L_\mu$ plus a term logarithmic in $\varepsilon^{-1}$. To connect $L_\mu$ to the chemistry along the path, we use the fact that a warm start at $s_i$ comes from the previous point $s_{i-1}$. If the continuation path is Lipschitz continuous, then one can show that
\[
    \langle E_0(s_{i-1}) | H(s_i) |E_0(s_{i-1})\rangle - E_0(s_i) \le c_E N_o 
\]
for a constant $c_E>0$. This bound shows that the warm start has only $O(N_o)$ excess energy with respect to $H(s_i)$. For constant $\eta$, Markov's inequality gives
\[
    \Pr_{\mu} \left( E_{J_0}(s_i) - E_0(s_i) > \frac{c_E}{\eta}N_o \right) \le \eta,
\]
and therefore concentration in the first $L_\mu=O(N_o/\eta)$ layers. In the formal theorem below, however, the accumulated failure budget takes $\eta$ to be inverse-polynomial in the number of continuation points. Markov's inequality alone would then give a larger layer cutoff. We therefore state the theorem conditionally in terms of the required low-energy-window concentration, or equivalently one may replace this Markov step by a stronger tail/variance assumption for the transported warm starts. The next lemma summarizes the resulting linear cooling-time bound when the extensive-window concentration condition holds.

\begin{lemma}[Linear Cooling Time; informal]
    \label{lem:linear_cooling_informal}
    Adopt the setting of~\cref{lem:eth-informal}. Suppose the warm start is concentrated over an extensive energy window above the ground state, i.e., there exists $c_E>0$ such that  \[\Pr_{\mu}\!\big(E_{J_0}-E_0 \le c_E N_o\big) \geq  1-\eta,\] then for any $0 < \varepsilon < 1$, the cooling time required at each fixed $s_i$ is 
    \begin{equation*}
        t = O\left(N_o + \ln\frac{1}{\varepsilon}\right).
    \end{equation*}
\end{lemma}

Combining \cref{lem:dk_constant}, which sets the required grid size, with \cref{lem:linear_cooling_informal}, we obtain an upper bound on the total number of cooling steps required along the path:
\[
    \sum_{i=1}^{N_H} t_{i} = \widetilde{O} \left( \frac{C_{\mathrm{DK}}^2 N_o^2}{\epsilon_E} \right)
    = \widetilde{O} \left( \frac{C_{\mathrm{DK}}^2 N_o}{\epsilon_{\rm int}} \right),
\]
where $t_i$ denotes the number of ideal kernel applications used at grid point $s_i$ and $\epsilon_{\rm int}=\epsilon_E/N_o$. This is the central mixing-time quantity for dissipative continuation. In~\cref{sec:resource-estimates}, we provide a concrete quantum procedure that implements a cooling step realizing such a Markov kernel. If $C_{\mathcal L}$ denotes the cost of constructing or block-encoding the filtered jump operators and $C_S(t,\epsilon)$ denotes the cost of simulating the Lindbladian for time $t$ to precision $\epsilon$, then the end-to-end logical cost has the schematic form
\[
    G_{\rm total}
    =
    \sum_{i=0}^{N_H-1} t_i\, C_{\mathcal L}(s_i)\, C_S(t_i,\epsilon_i).
\]
Using nearly linear-time Lindblad simulation, $C_S(t,\epsilon)=\widetilde{O}\!\left(t\|\mathcal L\|_{\rm be}\ln(t\|\mathcal L\|_{\rm be}/\epsilon)\right)$~\cite{li2023_lineartime}. Thus the displayed result is a primitive-count or mixing-time bound; the physical gate complexity acquires the additional block-encoding and Lindbladian-simulation factors, which are polynomial in $N_o$ for the implementations considered below.

\subsection{Continuation Path Optimization}
\label{sec:path_optimization}

\begin{figure*}[th]
\includegraphics[width=0.7\linewidth]{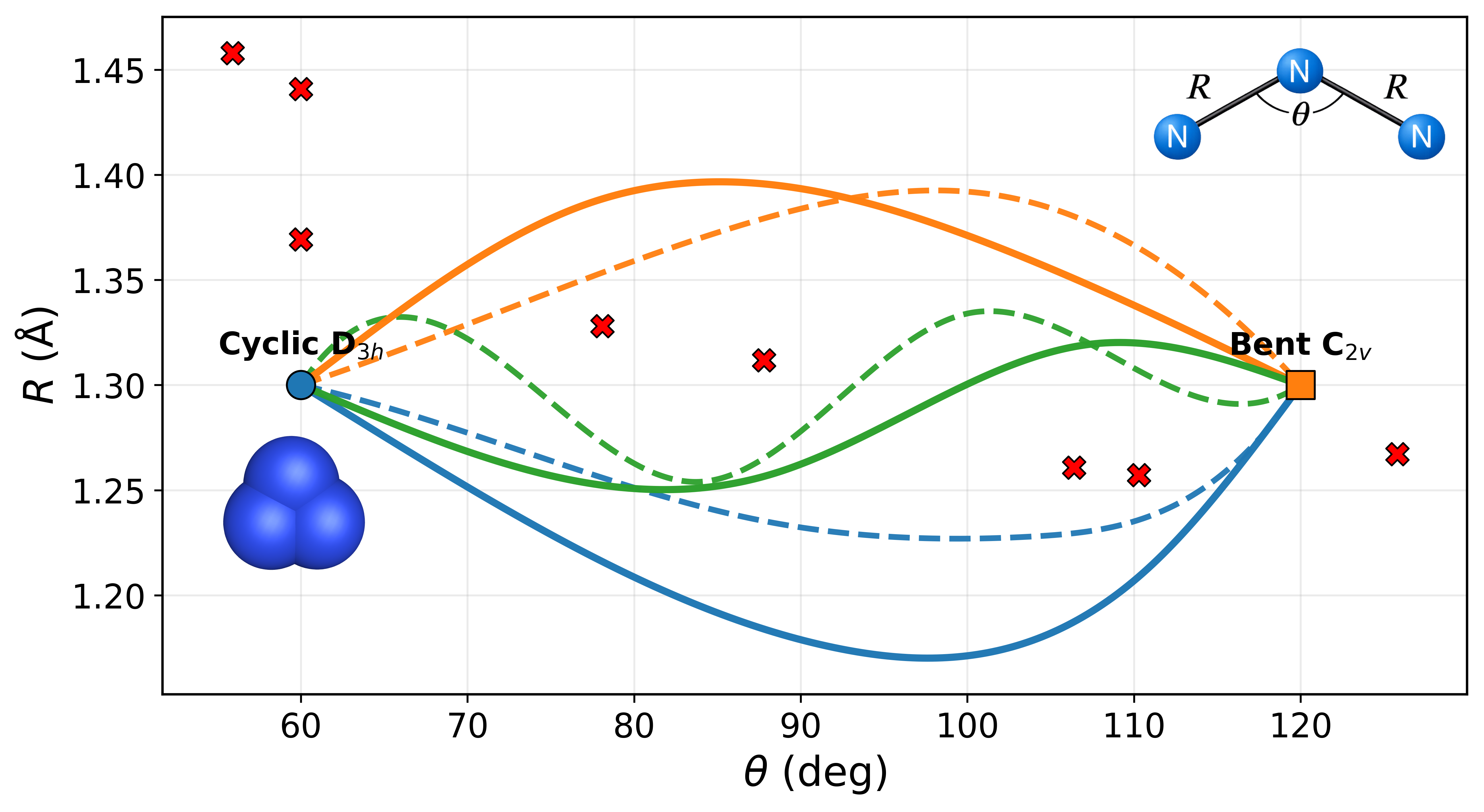}
   \caption{ Continuation-path optimization for N$_3$ in the \((\theta,R)\) internal-coordinate plane, demonstrating that our dissipative ground-state transport need not follow a minimum-energy path. Dashed curves show initial guess paths and solid curves show the optimized trajectories obtained by minimizing the path functional \(J(\mathbf R)\), which trades off geometric length, smoothness, and avoidance of conical-intersection (CI) neighborhoods. All paths connect the cyclic \(D_{3h}\) endpoint (\(\theta \approx 60^\circ\), \(R \approx 1.30~\text{\AA}\)) to the bent \(C_{2v}\) endpoint (\(\theta \approx 120^\circ\), \(R \approx 1.30~\text{\AA}\)). Red \(\times\) markers indicate surface-crossing / CI-type structures reported by \citet{Kerkines2009}, which the conditioned continuation paths detour around to improve clearance from small-gap regions.
}
    \label{fig:n3_landscape}
\end{figure*}

In quantum chemistry, the term \textit{reaction path} most commonly refers to a minimum-energy path (MEP) on the potential-energy surface, typically computed by standard path-finding routines such as the nudged elastic band method, string methods, synchronous-transit methods, or constrained optimization~\cite{Williams2016,VaucherReiher2018MEP}. Our use of a \textit{continuation path} is different. The continuation path in dissipative continuation is not proposed as a standalone method for discovering the chemically correct TS or MEP. Rather, the TS candidate and reaction channel are inputs supplied by the surrounding classical workflow, and the role of path optimization is to condition the quantum state-preparation problem.

More precisely, once a chemically validated reaction channel or TS-containing path segment has been specified, one may locally deform, smooth, or reparameterize the continuation path to improve the conditioning of the quantum state-preparation problem. This optimization need not minimize the physical potential energy at every point along the route. Instead, it targets quantities that influence continuation: geometric length, mesh smoothness, and clearance from known or diagnosed conical-intersection or small-gap neighborhoods.

This motivates the smoothness-regularized continuation-path functional
\begin{equation}
    \begin{split}
        J(\mathbf{R})
        \coloneqq&\,
        \lambda_{\mathrm{len}}\int_0^1 \|\mathbf{R}'(s)\|ds \\
        +& \,
        \lambda_{\mathrm{curv}}\int_0^1 \|\mathbf{R}''(s)\|^2ds \\
        +& \,
        \lambda_{\mathrm{CI}}\int_0^1 \Phi_\mathrm{CI}(\mathbf{R}(s))\,ds, 
    \end{split}
\end{equation}
Here the length term controls the coarse Davis--Kahan surrogate in~\cref{cor:cdk_uniform_surrogate}, the curvature term suppresses kinks and mesh artifacts, and the CI/small-gap term is an optional diagnostic penalty based on available chemical information. Importantly, minimizing the unweighted geometric length is only guaranteed to reduce a uniform upper bound on $C_{\rm DK}$, obtained by replacing the local weight $L_H(\mathbf R(s))/\Delta(s)$ by the global factor $L_{\max}/\Delta_{\min}$. It is not, by itself, guaranteed to reduce the true Davis--Kahan constant. A more targeted objective would replace the first term by a weighted length $\int_0^1 [L_H(\mathbf R(s))/\Delta(s)]\|\mathbf R'(s)\|ds$, but evaluating that weight accurately may require expensive electronic-structure calculations. We therefore use the length and curvature terms as inexpensive geometric surrogates, supplemented when possible by diagnostics for small-gap regions. We do not assume an oracle for locating conical intersections. In practice, $\Phi_{\rm CI}$ may be constructed from prior multireference studies, small-active-space diagnostics, approximate excited-state calculations, or spin-flip TDDFT-type methods~\cite{Matsika2021ConicalIntersections,WinslowCrossRobinson2020SpinFlipCI}. If such diagnostics are unavailable or unreliable, the penalty should be omitted or treated as a heuristic, and the theorem applies only to the resulting validated continuation path.

The N$_3$ molecule provides a concrete example where this continuation path freedom matters. As shown by the high-level multireference study of \citet{Kerkines2009}, the low-lying doublet manifold of N$_3$ contains many surface crossings and conical-intersection-type features as the NNN angle bends from the linear region toward cyclic geometries. This implies that a naive path (e.g., a minimum-energy path or a simple linear interpolation in internal coordinates) can easily encounter geometries where the electronic gap to a nearby excited state becomes small, precisely the regime that worsens adiabatic sensitivity and inflates the Davis--Kahan constant \cite{Kerkines2009}.

Operationally, we use the catalog of crossing structures and the qualitative location of intersection regions reported for N$_3$ as targets to avoid during optimization. This yields a simple, reproducible workflow: compute or import a coarse map of ``problematic'' regions (small-gap or intersection neighborhoods), set $\lambda_{\mathrm{CI}}$ to enforce clearance from those regions, and tune $\lambda_{\mathrm{len}},\lambda_{\mathrm{curv}}$ to maintain a short, smooth path.
The cost and reliability of constructing such a diagnostic map are not included in the state-preparation theorem; they belong to the surrounding classical preprocessing workflow, just like TS localization and active-space selection.

In \cref{fig:n3_landscape}, we visualize the effect of optimizing the path functional $J(\mathbf R)$ on N$_3$.
All three candidate paths connect the cyclic $D_{3h}$ region ($\theta \approx 60^\circ$, $R \approx 1.30$~\AA) to the bent $C_{2v}$ region ($\theta \approx 120^\circ$, $R \approx 1.30$~\AA), but the optimized trajectories (solid curves) systematically deform away from the intersection ``hot spots.''
Compared to their initial guesses (dashed curves), the optimized paths remain smooth and avoid sharp turns while making deliberate detours that increase clearance from the labeled crossing regions (notably the cluster near $\theta \sim 106$--$112^\circ$ and the interior crossings around $\theta \sim 80$--$90^\circ$). This behavior is consistent with the intended role of the CI penalty term in $J(\mathbf R)$, trading a modest increase in geometric distance for a more favorable gap profile along the route, thereby targeting a more favorable gap profile and a smaller surrogate for the Davis--Kahan cost of subsequent ground-state transport.
The N$_3$ example is used to illustrate the geometric role of continuation-path conditioning; it is not presented as a full resource estimate or end-to-end quantum simulation of the N$_3$ reaction dynamics.

\section{Implementation and Resource analysis}
\label{sec:resource-estimates}

In this section, we present a practical quantum-circuit that implements the cooling primitive from the dissipative continuation algorithm. We begin by analyzing the complexity associated with a single application of the cooling primitive for a Hamiltonian at a fixed point along the continuation path $H(s_i)$. We then show a numerical demonstration of dissipative continuation looking at a chemical reaction in H$_4$, for increasing applications of the cooling primitive, providing some evidence that the algorithm is effective (in terms of preparing a state with high overlap with the ground state), even for a low number of applications or the cooling primitive. Finally, we perform a cost analysis of the full algorithm, separating the ideal cooling-step count from the logical cost of implementing each dissipative primitive for a number of chemical systems of interest.

\begin{figure*}[t]
    \centering
    \includegraphics[width=.9\linewidth]{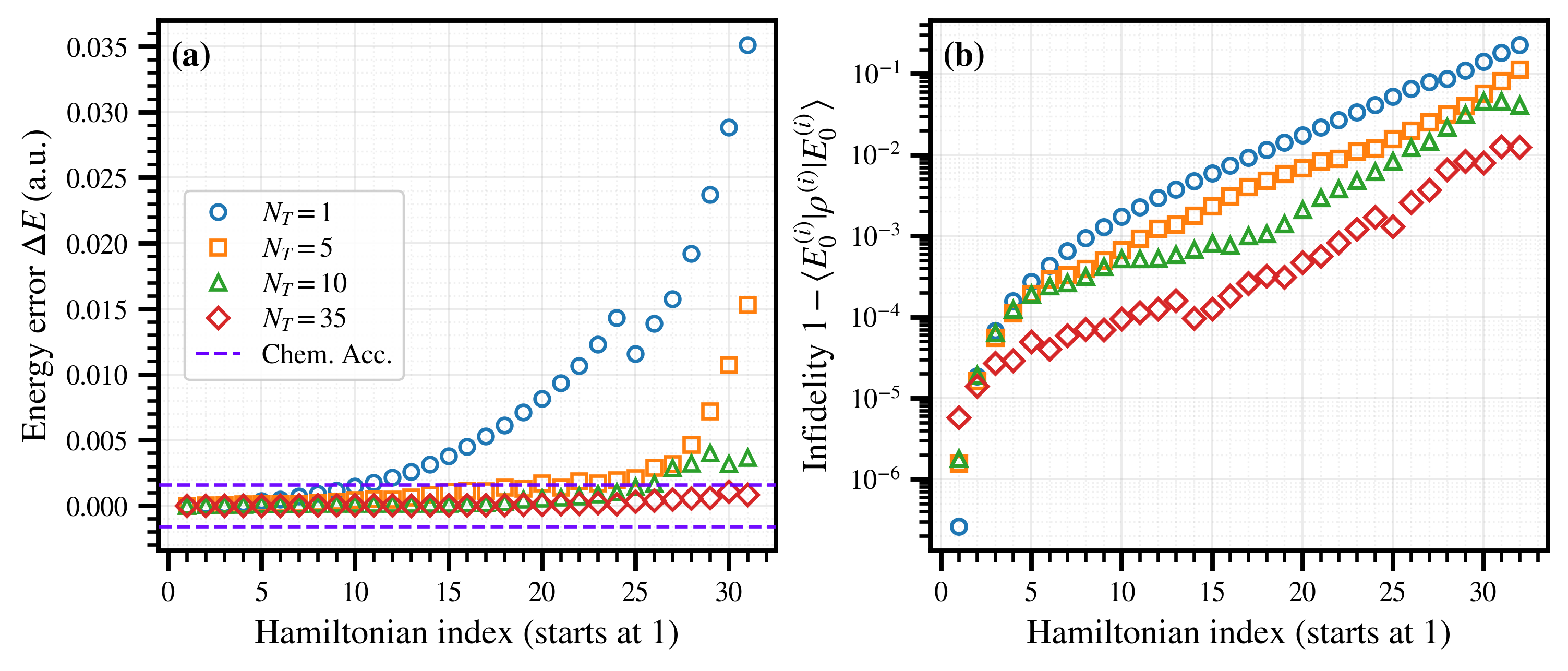}
    \caption{\textbf{(a)} Energy error $\Delta E(N_T) = \left|\mathrm{tr}\left(H (e^{\tau \mathcal{L}})^{N_T}[\rho^{(0)}]\right)-E_0\right|$ (a.u.) after $N_T$ applications of the dissipative time step at each successive Hamiltonian along a discretized path, beginning from the ground state of $H(0)$, $\rho^{(0)} = |E_0^{(0)}\rangle\langle E_0^{(0)}|$, on the H$_4$ potential energy surface (STO-3G) where we take $\tau=0.01$ in these numerical simulations; the final chemical geometry is the transition state (square H$_4$ molecule) corresponding to $s=0.5$ and the 32nd Hamiltonian along the continuation path. Each marker shows the final energy estimate after $N_T$ dissipative time steps per Hamiltonian along the continuation path, with $N_T\in\{1,5,10,35\}$ and Hamiltonian indices starting at $H(s_1)$.  Chem. Acc. indicates where the threshold for \textit{chemical accuracy} is: $1.6$ mHa. \textbf{(b)} Corresponding infidelity $1-\langle E^{(i)}_0|\rho^{(i)}|E^{(i)}_0\rangle$ on a base-10 logarithmic scale for the same runs, showing that larger $N_T$ improves preparation fidelity across the path where $|E_0^{(i)}\rangle$ is the ground state of the $i$-th continuation-path Hamiltonian $H(s_i)$ and $\rho^{(i)}$ is the approximate ground state prepared via dissipative continuation.}

   \label{fig:h4_numerical}
\end{figure*}

\begin{figure}
    \centering
    \includegraphics[width=\linewidth]{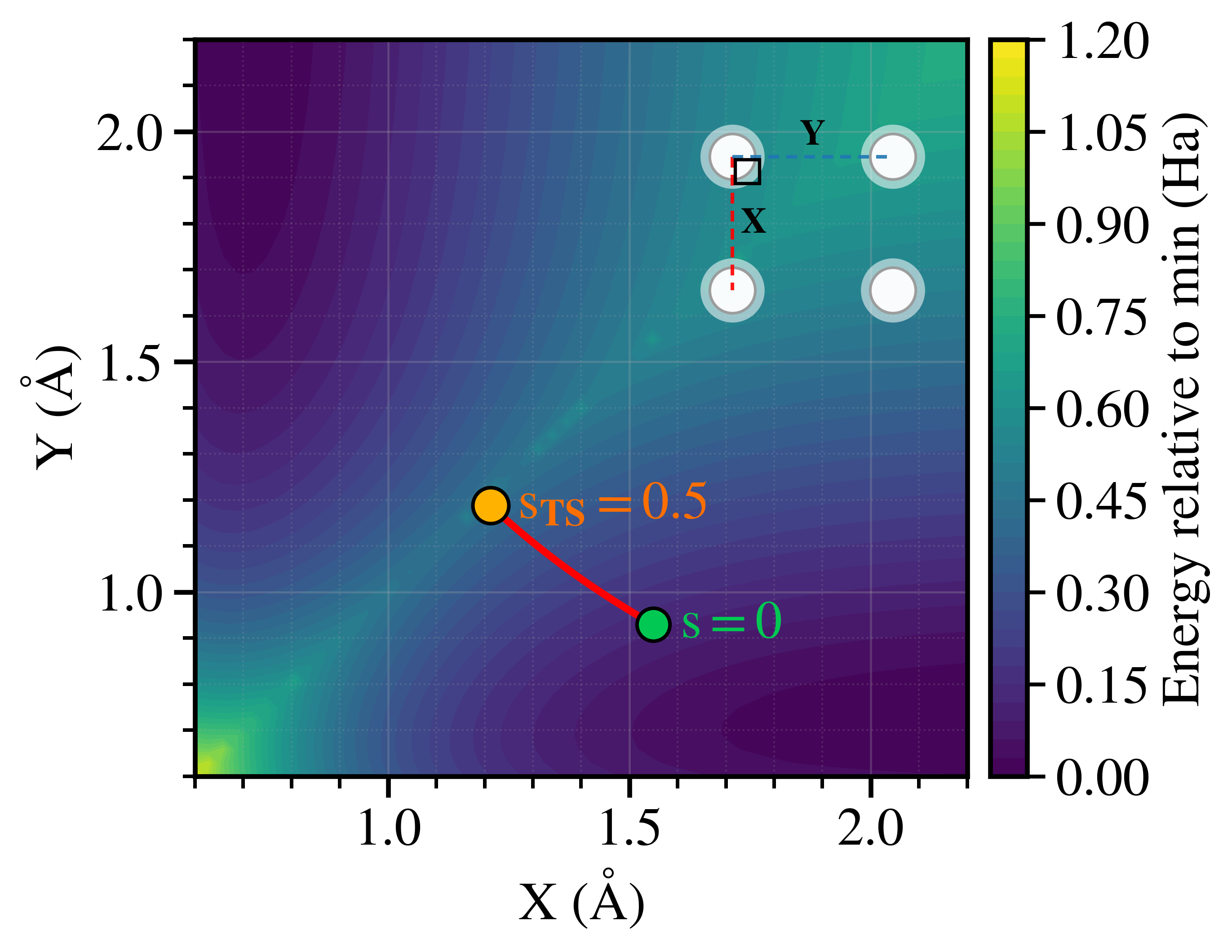}
    \caption{ H$_4$ potential energy surface (STO-3G) over rectangular geometries parameterized by $(X,Y)$, with the continuation path segment connecting $s=0$ weakly correlated starting point to the target transition state (TS) point $s=0.5$. The path follows a rectangle distortion of H$_4$ at fixed area $XY=a^2$, with atoms at $(\pm X/2,\pm Y/2,0)$ and aspect ratio swept as $\lambda(s)=\lambda_{\min}^{1-s}\lambda_{\max}^{s}$ with $\lambda_{\min}=0.6$, $\lambda_{\max}=1.6$, and $X=a/\sqrt{\lambda}$, $Y=a\sqrt{\lambda}$ for $s\in[0,0.5]$; the square geometry with side length $a=1.2$~\AA\ occurs at $s=0.5$ and exhibits strong multi-reference character.}

   \label{fig:h4_pes}
\end{figure}

A standard mathematical starting point for \textit{engineering} Lindbladian dynamics that drive a system into low-energy states is the Davies weak-coupling construction~\cite{Davies1974Markovian,Davies1976MarkovianII}. In this framework one obtains a GKLS generator whose stationary state is the Gibbs state of the target Hamiltonian $H$. Modern dissipative ground-state preparation generalizes this approach to tailor jump operators such that the resulting Lindbladian has the ground state as its steady state. A particularly efficient method for doing so was introduced in Ref.~\cite{ding2024single} which uses \textit{filtered} jump operators to construct such a Lindbladian.

Formally, we define the purely dissipative Lindbladian
 \begin{equation}
        \mathcal{L}[\rho]
        =
        \sum_{a \in \mathcal{A}}
        \left(
        K_a \rho K_a^\dagger
        - \frac{1}{2} \{ K_a^\dagger K_a, \rho \}
        \right),
        \label{Eq:Lindbladian}
    \end{equation}
over a finite set of filtered jump operators indexed by $a \in \mathcal{A}$. In the chemical setting, there are several choices for the unfiltered operators $A_a$ that lead to efficient mixing~\cite{Li_2025}, e.g., ${\mathcal{A} = \{ a_i^\dagger a_j \}_{i,j=1}^{N_o}}$, where $a_i^\dagger$ ($a_j$) are the second quantized fermion creation (annihilation) operators.

Once a suitable set of jump operators is selected,Once a suitable set of jump operators is selected, we obtain the filtered jump operators $K_a$ by integrating out the Bohr frequencies ($\omega\in \{E_i - E_j : E_i, E_j \in \mathrm{Spec}(H)\}$) that correspond to heating the system, leaving only cooling transitions.
\begin{definition}
\label{def:filtered_jumps}
    For each \textit{unfiltered} jump operator $A_a \in \mathcal{A}$, define the corresponding filtered jump operator $K_a$ as
    \begin{equation}
        K_a \coloneqq \int_{\mathbb{R}} f(t) e^{iHt} A_a e^{-iHt} dt,
    \end{equation}
    where the associated $f(t)$ is the time-domain kernel of a filter function $\hat{f}(\omega)$ that filters out positive Bohr frequencies $\hat{f}(\omega)=0$, $\forall \omega >0$.
\end{definition}

In \cref{def:filtered_jumps}, the filtered jump operator $K_a$ is written as a continuous superposition of Heisenberg-evolved operators $e^{iHt}A_a e^{-iHt}$, weighted by a kernel $f(t)$. In practice, however, this integral must be \textit{truncated} to a finite time window and \textit{discretized} to a finite set of times:
\begin{equation}
    K_a \approx \sum_{m=-M}^{M} c_m e^{iHt_m}A_a e^{-iHt_m}, \qquad t_m = m\Delta t,
\end{equation}
where the coefficients $c_m$ arise from the choice of the filter function and the associated numerical integration scheme, e.g., quadrature. The largest queried evolution time $T\coloneqq M\Delta t$ controls how sharply the frequency-domain filter $\hat f(\omega)$ can approximate the ideal projector onto negative Bohr frequencies: increasing $T$ narrows the effective transition band of the filter and reduces residual ``heating'' leakage, while smaller $T$ yields a broader filter and larger leakage. Motivated by this tradeoff, we define the \textit{time support} $S(\varepsilon_{\mathrm{leak}})$ as the maximum real time magnitude that must be queried by the (truncated and discretized) implementation in order to guarantee the target leakage level $\varepsilon_{\mathrm{leak}}$. Operationally, when sampling on a uniform grid, the number of distinct real-time evolution primitives $e^{iHt}(\cdot)e^{-iHt}$ required is proportional to the number of grid points in $[-S(\varepsilon_{\mathrm{leak}}),S(\varepsilon_{\mathrm{leak}})]$, i.e. $\#\mathrm{calls}=\Theta(S(\varepsilon_{\mathrm{leak}})/\Delta t)$.

\subsection{Time support scaling}
\cref{thm:filter-scaling} identifies and compares two classes of filter functions that each carry their own distinct asymptotic scaling. In both constructions, the filtered jump operator is realized by a time--domain representation (an integral, approximated by a finite sum) in which copies of the bare jump operator are conjugated by real-time evolution under $H$ and weighted by the chosen filter kernel. The decisive implementation parameter is therefore the kernel's \textit{time support} $S(\varepsilon_{\mathrm{leak}})$, i.e., the effective time window over which the kernel has non-negligible weight. This window determines the largest evolution time that must be synthesized and, after choosing a discretization step, the number of distinct time samples appearing in the approximation of $e^{\tau \mathcal{L}}$ for a small time step $\tau > 0$~\cite{ding2024single}.

\begin{theorem}[Time support scaling for Gevrey vs. periodic square-wave filters; informal]
    \label{thm:filter-scaling}
    Let $H$ be a Hamiltonian with nondegenerate ground state $\ket{E_0}$, spectral gap $\Delta > 0$, and let $A$ be a bounded jump operator with $\|A\|\le 1$. For any target leakage $0<\varepsilon_{\mathrm{leak}}<1$ there exist two families of filtered jump operators $K^{\mathrm{G}}_{\varepsilon_{\mathrm{leak}}}$ and $K^{\mathrm{SW}}_{\varepsilon_{\mathrm{leak}}}$ with the following time-support:
    \begin{enumerate}
        \item There exists a Gevrey-class filter with time-domain kernel
        $f^{\mathrm{G}}_{\varepsilon_{\mathrm{leak}}}(t)$ essentially supported on an interval
        of length
        \begin{equation}
            \label{equ:gevrey}
            S_{\mathrm{G}}(\varepsilon_{\mathrm{leak}})
            =
            O\!\left(
            \frac{\|H\|}{\Delta}
            \bigl[\ln(1/\varepsilon_{\mathrm{leak}})\bigr]^{\alpha}
            \right),
        \end{equation}
        for some $\alpha > 1$ depending only on the Gevrey class~\cite{ding2024single}.
    
        \item There exists a truncated periodic square-wave filter with discrete-time kernel $f^{\mathrm{SW}}_{\varepsilon_{\mathrm{leak}}}(t)$ supported on
        \begin{equation}
            S_{\mathrm{SW}}(\varepsilon_{\mathrm{leak}})
            =
            O\!\left(
            \frac{\|H\|}{\Delta}
            \frac{1}{\varepsilon_{\mathrm{leak}}}
            \right),
        \end{equation}
        up to bounded oscillatory factors in the leakage.
        \end{enumerate}
        Up to constants from the discretization step, the number of distinct real-time evolution calls (the number of sample points in the time integral) is proportional to the time support length $\Theta(S(\varepsilon_{\mathrm{leak}}))$.
\end{theorem}
A more detailed statement of \cref{thm:filter-scaling} and its proof are provided in the \Cref{sec:thm2} as~\cref{thm:filter-scaling_formal}.

For Gevrey filters, the dominant accuracy parameter is $M_s$ the number of quadrature points~\cite{ding2024single}. For Fourier filters, this parameter is the number of Fourier modes $N_\omega$; see \Cref{sec:thm2}. The values of these parameters are set by the target ground-state leakage $\varepsilon_\mathrm{leak}$. In practice, the conversion between target leakage $\varepsilon_\mathrm{leak}$ and the accuracy parameter value is straightforward in the case of the Fourier filter e.g., $N_\omega \sim 1000$ Fourier modes can tackle spectral gaps in the milliHartrees. In contrast, a Gevrey filter's optimal parameters are more difficult to determine a priori~\cite{ding2024single}.

\subsection{Numerical Demonstrations}



We consider a one-dimensional path on the Born–Oppenheimer potential energy surface (PES) of planar H$_4$. To reduce complexity, we consider a path along which the nuclei remain in a rectangular configuration (i.e., the angle between the axes of the pairs is fixed), reducing the dimension of the energy manifold to a two-dimensional surface, with the aspect ratio between the separation of the pairs varied smoothly as a function of the reaction coordinate $s \in [0,  1]$; see~\cref{fig:h4_pes}. The path is chosen to pass through the configuration in which the nuclei sit in a perfect square (that is, the rectangular configuration has equal length sides). 

This square geometry of the H$_4$ molecule is non-trivial from both chemical and algorithmic perspectives: in the STO-3G basis near a side length of $\sim 1.2$~\AA, the increased symmetry induces near-degeneracies among competing bonding patterns, producing a strongly multi-reference ground state and a compressed low-energy spectrum~\cite{Liu2024-h4}. Thus, the rectangular-to-square deformation provides a minimal proxy for a reaction coordinate that crosses a strongly correlated ``bottleneck'' region, while remaining small enough in FCI dimension to allow us to classically simulate the performance of our algorithm. Additionally, we ran the Procrustes alignment procedure described by~\cref{prop:gs-overlap-principal-angles} to establish a consistent orbital gauge along the path. In this benchmark the alignment is initialized at one endpoint for convenience; in TS applications, the same procedure can instead be initialized from the TS-centered active-space frame and propagated outward toward a tractable geometry. The quantum continuation then uses the tractable endpoint as the warm-start location and transports the state back toward the TS in the same aligned representation.

We employ quadratic, number-conserving jump operators $A_{ij}^{(1)}=c_i^\dagger c_j+c_j^\dagger c_i$ and $A_{ij}^{(2)}=i(c_i^\dagger c_j-c_j^\dagger c_i)$ for $i<j$, allowing for spin mixing between, for example, $\alpha$ and $\beta$ orbitals.
With this selected dissipative ansatz, we observe behavior consistent with the local downward-drift condition in the H$_4$ molecule by visualizing the graph $G=(V,E)$ associated with the Markov process induced by dissipative cooling. This provides a diagnostic example of the regime assumed in the cooling-time theorem, rather than evidence that the condition holds generically for all transition-state Hamiltonians. The vertices of $G$ are the eigenstates in the FCI subspace $V = \{|E_i\rangle\}_{i=0}^{D-1}$ and two eigenstates are connected by an edge $\left( |E_i\rangle, |E_j\rangle\right)\in E$ if the probability of a transition between the two is non-zero $P_{i \to j} >0$. \cref{fig:main_fig}a depicts a visualization of the graph $G$ with the energy-ordered vertices arranged on a circle, where edges are only drawn if their transition probability is bound below by $P_{j \to k} \geq 0.01$. We can observe that H$_4$ is consistent with the uniform drift condition of \cref{eq:uniform-drift-text} at the square multi-reference geometry.

\cref{fig:h4_numerical}a demonstrates that increasing the number of dissipative steps per geometry leads to convergence to the true ground state at all points along the continuation path for this example. Indeed, as shown, $N_T=35$ dissipative steps are sufficient to achieve chemical accuracy, $\epsilon_E \leq 1.6$ mHa, across all $N_H$ geometries. \cref{fig:h4_numerical}b shows the infidelity of the prepared state with the true ground state. As we can see, even one application of the dissipative step per geometry allows for a reasonable overlap with the true ground state at the transition-state geometry. Increasing the number of applications only improves the simulation.
These numerical simulations show rapid convergence, corresponding to short mixing times, at each geometry along the path. One could imagine a scenario where the number of dissipative steps is tuned along the path (increasing with the proximity to the TS) to minimize cost, while maximizing accuracy.

\begin{figure*}[t]
    \centering
    \includegraphics[width=.8\linewidth]{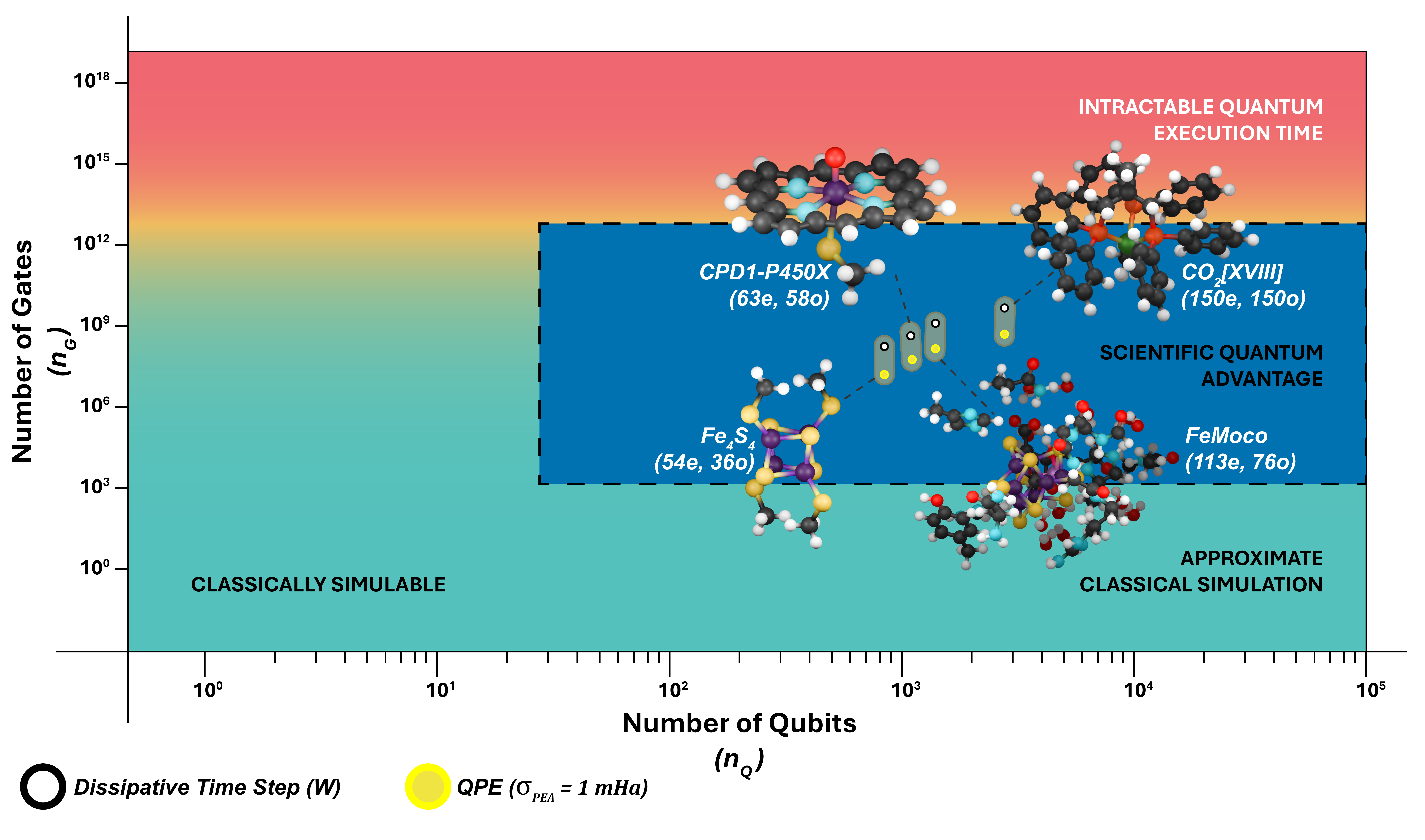}
    \caption{A log–log plot (base 10) plot of the quantum logical resource estimates for a single dissipative time step operator $W(\sqrt{\tau}) \approx e^{\tau \mathcal{L}}$ along with the cost of doing QPE across benchmark chemical systems~\cite{Low_2025}. The resource estimates are calculated from resource estimates provided in~\cite{Low_2025} for QPE and the block encoding $U_H$ using the ``DFTHC-BLISS-SA'' SOTA electronic structure compression method to minimize $T$/Toffoli gate counts. Each point corresponds to a row in \Cref{tab:block-enc-lindblad-qpe}, with Toffoli count on the vertical axis and logical qubits on the horizontal axis. Regions of current and projected hardware capabilities are shaded to illustrate which systems and algorithms may be accessible on near- or medium-term fault-tolerant devices~\cite{camps2025quantumcomputingtechnologyroadmaps}.} 
    \label{fig:joe_plot}
\end{figure*}



\subsection{Cost Analysis}

To assess the practical prospects for efficient dissipative ground-state preparation along continuation paths, we estimate the logical quantum resources required to prepare ground states for strongly correlated chemical systems along such paths, using optimized quantum circuit encodings~\cite{Low_2025}.

Several quantum algorithms can be used to implement time steps that approximate the purely cooling dynamics generated by $\{K_a\}_{a\in\mathcal{A}}$. Broadly, these approaches fall into two classes: collision-based schemes, in which interactions with one or more ancilla qubits induce dissipative kicks that drive the system’s energy downward~\cite{ding2024single, ding2025End}; and quantum operator Fourier transform–based circuits, which realize the desired filter function using coherent ancilla states of the form $\sum_i \sqrt{f(t_i)}\lvert i\rangle$ together with controlled time evolution~\cite{chen2023quantum, Chen2025}.

Dissipative kicks are considered to be the most viable near-term approach for implementing dissipative dynamics on quantum computers~\cite{ding2025End}. For this class of methods, the circuit depth required to implement a short-time dissipative continuation is determined by the leakage properties of the chosen filter function.

For the purpose of resource estimation, we select a periodic square-wave filter consisting of $N_\omega$ Fourier modes and perform the real-time evolutions $\exp(\pm iHt)$ using QSP/QSVT~\cite{gilyen2019quantum,low2017optimal,Low2019hamiltonian,low2016resonant,Haah2019product}. For a single exponential of duration $t$, the QSVT degree scales as
\begin{equation}
    d_{\rm QSVT}(t,\epsilon_{\rm sim})
    =
    \Theta\!\left(\alpha |t|+\ln(1/\epsilon_{\rm sim})\right),
\end{equation}
where $\alpha$ is the block-encoding normalization of $H$ and $\epsilon_{\rm sim}$ is the target polynomial-approximation error. For the square-wave filter used here, the largest evolution time is $t=\pi/2$ after rescaling, giving the estimate
\begin{equation}
    N_{U_H}
    =
    \Theta\!\left(
    2N_\omega\left(\pi\alpha/2+\ln(1/\epsilon_{\rm sim})\right)
    \right)
\end{equation}
block-encoding calls per dissipative time-step operator $W(\sqrt{\tau})$. Thus, ignoring lower-order Clifford control and state-preparation costs, the logical cost of one dissipative step at geometry $s_i$ is
\begin{equation}
    G_W(s_i)
    \approx
    N_{U_H}(s_i)\,\mathrm{Cost}(U_H(s_i))
    +G_{\rm jump}(s_i),
\end{equation}
where $G_{\rm jump}$ denotes the cost of implementing the dilated jump-operator kicks and associated controls. The full continuation cost is then
\begin{equation}
    G_{\rm total}
    =
    \sum_{i=0}^{N_H-1}N_T(s_i)G_W(s_i),
\end{equation}
where $N_T(s_i)$ is the number of dissipative steps used at geometry $s_i$. Under the ideal local-mixing assumptions analyzed above, $N_H=\widetilde O(c_E C_{\rm DK}^2/\epsilon_{\rm int})=\widetilde O(c_E N_o C_{\rm DK}^2/\epsilon_E)$ and $N_T(s_i)=O(N_o+\ln(1/\varepsilon_i))$, but in the resource estimates below we report the dominant local primitive cost near the TS region, where the smallest gaps and slowest mixing are expected.

Several implementation errors are not present in the ideal Markov-chain analysis and must be accounted for in a fault-tolerant implementation. The main contributions are filter leakage, QSVT polynomial error, splitting error in the product formula used to implement the short-time dissipative step, dilation/measurement error in the one-ancilla Lindbladian simulation primitive, and sampling error associated with randomly selecting jump operators. We treat these contributions as producing an effective per-step deviation
\begin{equation}
    \left\|\widetilde{\Phi}_{i}-e^{\tau\mathcal L_{i}}\right\|_\diamond
    \le
    \delta_{\rm step}(s_i),
\end{equation}
where $\widetilde{\Phi}_{i}$ is the implemented channel at the fixed geometry $s_i$. At the population level, the most important effect is an effective leakage probability $q_{\rm eff}$: with small probability the imperfect filter or imperfect implementation permits an upward or insufficiently downward move between energy layers. In the ideal analysis this leakage is zero. In the noisy implementation, the same hitting-time proof can be interpreted as applying to a perturbed Markov chain with a reduced effective downward-drift probability, provided $q_{\rm eff}$ is smaller than the drift margin $p_{\min}$. This increases the required number of cooling rounds $N_T(s_i)$ and contributes an accumulated channel error bounded by $\sum_iN_T(s_i)\delta_{\rm step}(s_i)$.
At the level of this logical-resource estimate one may take, schematically,
\[
\delta_{\rm step}
\lesssim
\delta_{\rm filt}
+
\delta_{\rm QSVT}
+
\delta_{\rm split}
+
\delta_{\rm dil}
+
\delta_{\rm samp},
\]
where the terms denote filter leakage, QSVT polynomial error, product-formula error, dilation/measurement error, and sampling error over jump operators, respectively. We do not attempt to optimize these terms independently here; instead, they are absorbed into the target simulation precision used for the quoted logical estimates.
We treat this effective-leakage picture as implementation bookkeeping rather than as a new asymptotic theorem: the formal runtime theorem is stated for the ideal Lindbladian channel, while the estimates above indicate how finite-precision implementation errors would enter the number of cooling rounds and the accumulated final error.

Reducing the cost of implementing block encodings of electronic structure Hamiltonians $\mathrm{Cost}(U_H)$ has been a pervasive subject of recent quantum algorithms research. There are now many schemes for compressing and encoding electronic structure Hamiltonians into quantum circuits \cite{Motta_2021, Low_2025}. At present, the state-of-the-art is achieved by using a combination of several compression protocols resulting in the ``DFTHC-BLISS-SA" block encoding~\cite{Low_2025}. These optimized block encodings have established the lowest quantum resource cost for investigating chemical reactions, enabling quantum solutions to several application areas that carry significant scientific utility.

\cref{fig:joe_plot} summarizes our logical resource estimates for a single dissipative time step $W(\sqrt{\tau})$ along with the cost for quantum phase estimation (QPE) with chemical accuracy $\sigma_\text{PEA}=1$ mHa for representative systems for each of these application areas. In the strongly correlated regimes in which these systems live, we report resource estimates for $N_\omega = 1000$ Fourier modes for the filter function $f(t)$ given that spectral gaps on the order of 1 mHa are common when studying transition states. The QPE costs presented are single-shot QPE costs and do not include the repetition overhead associated with imperfect initial-state overlap. At strongly correlated TS geometries, no generally efficient classical procedure is known that guarantees inverse-polynomial overlap with the exact ground state. In such overlap-limited cases, direct pointwise QPE incurs an additional expected factor $1/p_0$, which may dominate the apparent single-shot cost advantage.

We observe that the cost of implementing the dissipative cooling operator $W$ is more than an order of magnitude greater than that of QPE. This motivates the exploration of methods to decrease the total number of $W$ operators that need to be applied at once during any single coherent run of a quantum circuit during the execution of our algorithm. We can employ techniques developed in~\cite{ceroni2023generatingapproximategroundstates}. The authors demonstrate how machine learning can be used to learn a function $\nu(\mathbf{R}; \boldsymbol{\gamma})$, where $\boldsymbol{\gamma}$ are the neural network parameters and $\mathbf{R}$ are the chemical geometries along the continuation path. This function determines the parameterization of a variational circuit $U(\nu(\mathbf{R}; \boldsymbol{\gamma}))|\psi_{0}(\mathbf{R})\rangle$, where $|\psi_0(\mathbf{R})\rangle$ is some geometry dependent initial state that gets us into the right region of Hilbert space in order to find a low-depth circuit representation $U$ of the ground state at each point along the continuation path. Using this approach to obtain ground states along a continuation path, we could manage with as few as just a single $W$. Our numerical study of the H$_4$ near its square geometry shows that in regions further from the transition state (emulated by the square geometry of H$_4$), only a single $W$ is required to maintain overlap.

\section{Discussion}
\label{sec:discussion}
A large part of the inspiration for this work comes from how tensor-network solvers such as DMRG are used in practice to map the boundary between classically tractable and classically hard regions of chemical continuation paths: in many systems, tensor-network methods provide chemically accurate reference energies and high-quality guiding states over large portions of the path, yet they can degrade near strongly correlated transition-state geometries where the wavefunction becomes highly entangled. We therefore use tensor-network performance as a practical indicator of where to deploy our dissipative state-preparation subroutine: the overall strategy is to exploit classical structure where available (to obtain warm starts and reliable energies) while reserving the quantum overhead of the engineered cooling primitive $W$ for the narrow multi-reference window near the transition state, where the mixing time can be large and classical ans\"atze lose accuracy.

In this work, we develop an efficient quantum algorithm, \textit{dissipative continuation}, that transports approximate ground states along classically supplied continuation paths from tractable geometries to strongly correlated TS geometries. The protocol can yield a conditional advantage over pointwise state-preparation methods in the regime where direct TS overlap is small, neighboring low-energy subspaces remain smoothly connected, and the engineered dissipative primitive mixes rapidly in the relevant energy window. In this regime, the number of ideal dissipative cooling steps scales as $\widetilde{O}(C_\text{DK}^2 N_o^2/\epsilon_E)$ for total energy error, equivalently $\widetilde{O}(C_\text{DK}^2 N_o/\epsilon_{\rm int})$ for intensive error, while the physical gate complexity is obtained by multiplying by the cost of implementing each Lindbladian step. A continuation path may be conditioned to reduce inexpensive surrogates for the Davis--Kahan cost, although unweighted path-length minimization alone is not guaranteed to reduce the exact $C_\text{DK}$. Within this specified setting, dissipative continuation improves over pointwise overlap-limited preparation and offers a complementary alternative to coherent transport methods such as digital adiabatic evolution and phase randomization~\cite{Berry2025Large, Boixo2009Eigenpath}. 

Our results build on recent progress in quantum algorithms for simulating dissipative dynamics, and we provide a complexity guarantee for the quantum state-preparation component under explicit path-smoothness and local rapid-mixing assumptions.However, the result should be read as a conditional state-preparation theorem for a structured hybrid workflow, not as an unconditional solution to transition-state search, active-space selection, or all classically hard transition-state instances. Strongly correlated, classically hard instances of this type are ubiquitous in nature~\cite{Reiher_2017,Wardzala2024,Rudshteyn2020AFQMC_TMComplexes, otten2024quantumresourcesrequiredbinding}. FeMoco, for example, has long served as a canonical strongly correlated target for quantum chemistry; recent work has demonstrated an equilibrium-geometry classical ansatz with very high overlap for FeMoco, rendering QPE efficient at or near equilibrium~\cite{zhai2026classicalsolutionfemoco} and strengthening the case for concentrating quantum resources on out-of-equilibrium regimes where classical tensor-network methods remain most challenged. In our setting, equilibrium-quality classical solutions provide warm starts near reactant or intermediate geometries, while the quantum dissipative step is used to propagate through the transition-state window where classical methods fail, enabling chemical-accuracy barrier heights relevant for screening and design of FeMoco-inspired synthetic analogues~\cite{Xu2025-femoco}. \footnote{FeMoco is used here as motivation for the kind of strongly correlated, chemically important systems where such a workflow could be valuable, rather than as evidence that the specific smoothness and drift assumptions have already been verified for a FeMoco transition path.}

Looking ahead, several complementary directions could further reduce costs and broaden applicability. First, trajectory-style approaches suggest that even faster mixing may be achievable by interleaving dissipative continuation with QPE-based primitives~\cite{JiangIPAM2026Trajectory}; consistent with this, our logical resource estimates (\cref{fig:joe_plot}) indicate that QPE can be substantially cheaper than a single application of $W$, suggesting that QPE interleaving could reduce the required number of $W$ applications. Second, approximate ground-state generation methods such as that of Ceroni \textit{et al.}~\cite{ceroni2023generatingapproximategroundstates} may further reduce the total $W$ budget over an entire pathway. Third, it is useful to distinguish \textit{basis-set error} from \textit{algorithmic hardness}: O'Gorman \textit{et al.}~\cite{OGorman2022Intractability} prove that approximating the ground-state energy of the electronic-structure Hamiltonian can remain QMA-complete even when restricted to a fixed single-particle basis and fixed particle number, by encoding hardness in the choice of orbitals used to discretize the problem, and they further show NP-completeness of determining the lowest-energy Slater determinant (Hartree--Fock) within such a fixed basis. They also explicitly pose as open directions whether the problem becomes easier when one \textit{promises} a bounded basis-set error (i.e., that the chosen basis contains a state sufficiently close to the true ground state) and, separately, how complexity changes when the basis is drawn from a parameterized family and must be selected as part of the instance. 

This suggests a principled interface between representation and dynamics. Active-space selection and TS validation remain classical electronic-structure tasks, but once a chemically validated active-space envelope or sequence of corresponding spaces has been chosen, the residual orbital gauge and continuation-path parameterization become controllable design degrees of freedom for the quantum state-preparation problem. Our protocol does not assume that basis selection is solved by the quantum routine, but it does exploit representation choices that improve the smoothness of the Hamiltonian path seen by the quantum computer. Finally, since coherent transport, QPE refinement, and dissipative continuation are not mutually exclusive, a unified theory of when coherent transport suffices and when it is optimal to switch to dissipative refinement would be valuable; in this direction, quantum speed-limit results for adiabatic algorithms at zero and finite temperature~\cite{chou2026universalscalingfinitetemperaturequantum} may provide principled criteria for strategy switching, including variants that transport a Gibbs state and cool only near the end, or that learn finite temperature reaction rates and extrapolate toward the near-room-temperature regime rather than strictly preparing the ground state directly.

Another promising direction for future work is to incorporate into our algorithm the \textit{Dissipative Quantum Eigensolver} (DQE) framework of Cubitt~\cite{cubitt2023dissipativegroundstatepreparation}, which formalizes ground-state preparation as an iterated \textit{quantum instrument}: the algorithm repeatedly applies weak, local generalized measurements (realizing an AGSP as one Kraus branch) and uses an explicit \textit{stopping rule} that depends only on the observed measurement record, thereby providing an intrinsic operational ``success certificate'' without requiring coherent phase estimation or deep reversible arithmetic. Notably, Cubitt emphasizes that a large portion of the DQE circuit can be implemented with Clifford operations, with only a single non-Clifford single-qubit rotation applied on a fixed ancilla line in a representative construction---an architectural feature that directly aligns with our emphasis on low-$T$ overhead primitives and mid-circuit measurement as an algorithmic resource. Building on this, Purcell, Rajput, and Cubitt~\cite{purcell2025faultresiliencedissipativeprocessesquantum} prove that for \textit{stabilizer-encoded} (geometrically local) Hamiltonians of the type that naturally arise in fermion-to-qubit encodings, a modified DQE procedure under circuit-level depolarizing noise can suppress the additive error in the final ground-space overlap \textit{exponentially in the code distance}, moving toward fault-tolerant-quality state preparation \textit{without} the full overhead of universal fault tolerance; in contrast, they show that dissipative quantum computation (as a model for general computation) does not enjoy such an intrinsic robustness advantage over the standard circuit model.

\section*{Acknowledgments}
The authors thank Ryan Mann, Gabriel Waite, Sophie Crouch, Mauro E. S. Morales, Lin Lin, Saurabh Totey and Adam Holmes for helpful discussions. S.J.E., T.W.W. and N.N. would also like to thank the attendees of the IPAM: Frontiers for Open Systems Workshop of whom there are too many to name.
M.J.B., S.S., D. C., and S.J.E. acknowledge the support of the ARC Centre of Excellence for Quantum Computation and Communication Technology (CQC2T), project number CE17010001. S. S. acknowledges support from the SQA.
Finally, we owe a debt of gratitude to Itay Hen for organizing a collaboration venue at USC.

\onecolumngrid

\appendix

\pagebreak
\section{Preliminaries}
\label{app:preliminaries}

This appendix collects the mathematical objects, norms, and conventions used throughout the proofs.
Unless stated otherwise, all operators act on a finite-dimensional Hilbert space $\mathcal{H}$ of
dimension $D$ (e.g., the FCI space in a fixed symmetry sector at a fixed geometry).

\subsection*{Quantum channels and continuous-time dynamics}

A quantum channel is a completely positive trace-preserving (CPTP) map
$\mathcal{E}:\mathcal{B}(\mathcal{H})\to\mathcal{B}(\mathcal{H})$. A (time-homogeneous) Markovian open-system
evolution is described by a one-parameter semigroup $\{e^{t\mathcal{L}}\}_{t\ge 0}$ of CPTP maps with generator
$\mathcal{L}$ (the \textit{Lindbladian}).

\begin{definition}[GKLS (Lindblad) generator]
\label{def:gkls}
    A linear map $\mathcal{L}:\mathcal{B}(\mathcal{H})\to\mathcal{B}(\mathcal{H})$ is a Lindbladian (GKLS generator)
    if it can be written as
    \begin{equation}
        \label{eq:gkls}
        \mathcal{L}(\rho)=-i[H,\rho]
        +
        \sum_{a\in\mathcal{A}}
        \Big(
        K_a \rho K_a^\dagger
        -\tfrac12\{K_a^\dagger K_a,\rho\}
        \Big),
    \end{equation}
    where $H=H^\dagger$ is Hermitian and $\{K_a\}_{a\in\mathcal{A}}$ are jump operators.
\end{definition}

In the main text and appendices we often focus on the \textit{purely dissipative} case (drop the Hamiltonian
commutator term) because the cooling analysis concerns population flow to low energy.

\subsection*{Energy basis, dephasing, and induced population dynamics}

Fix a geometry $s$ and let $H(s)=\sum_{j=0}^{D-1}E_j(s)\ket{E_j(s)}\!\bra{E_j(s)}$ be the electronic Hamiltonian (diagonal in its eigenbasis). For a state $\rho$, define the associated \textit{energy populations}
\[
    p_j \ \coloneqq\ \bra{E_j(s)}\rho\ket{E_j(s)},
    \qquad p=(p_0,\dots,p_{D-1})\in\Delta^{D-1}.
\]
We write $\mathcal{D}_s(\rho)\coloneqq\sum_j \bra{E_j(s)}\rho\ket{E_j(s)}\ket{E_j(s)}\!\bra{E_j(s)}$ for the dephasing channel in the energy basis at $s$.

Many dissipative constructions (including filtered-jump implementations used here) admit a convenient \textit{population picture}: after each short-time dissipative step, one may track the evolution of the diagonal of $\rho$ in the energy basis. This motivates the following induced Markov kernel.

\begin{definition}[Markov kernel induced by jump operators]
    \label{def:markov-kernel}
    Fix an orthonormal basis $\{\ket{j}\}_{j=0}^{D-1}$ (we will take $\ket{j}=\ket{E_j(s)}$). Given jump operators $\{K_a\}_{a\in\mathcal{A}}$, define unnormalized transition weights
    \begin{equation}
        \label{eq:gamma_jk_prelim}
        \Gamma_{j\to k}
        \ \coloneqq\
        \sum_{a\in\mathcal{A}}\big|\langle k|K_a|j\rangle\big|^2,
    \end{equation}
    and the associated row-stochastic kernel
    \begin{equation}
        \label{eq:P_prelim}
        P_{j\to k}
        \coloneqq
        \begin{cases}
        \dfrac{\Gamma_{j\to k}}{\sum_{\ell=0}^{D-1}\Gamma_{j\to \ell}},
        & \sum_{\ell=0}^{D-1}\Gamma_{j\to \ell}>0,\\[8pt]
        \delta_{jk},
        & \sum_{\ell=0}^{D-1}\Gamma_{j\to \ell}=0.
        \end{cases}
    \end{equation}
    The second line covers states with no outgoing jump weight. In particular, for an ideal strictly cooling filter this convention makes the ground state absorbing, so $P_{0\to0}=1$ and $P_{0\to k}=0$ for $k\ne0$.
    Then a population vector $p$ updates as $p'_k=\sum_j p_j P_{j\to k}$.
\end{definition}

The eigenstate thermalization hypothesis (ETH) posits that, in generic nonintegrable quantum many-body systems, individual energy eigenstates reproduce thermal expectation values for local observables; in our case, quasi-local filtered jump operators. This property motivates a \textit{uniform drift} condition that formalizes the resulting behavior of the Markov process with kernel $P$ associated with the dissipative dynamics.
\begin{definition}[Uniform monotone downward drift]
    \label{def:uniform_drift}
    The kernel $P$ has \textit{uniform monotone downward drift} on the relevant energy window if the following two conditions hold. First, the layer index is non-increasing under one step of the chain, equivalently $P_{j\to k}=0$ whenever $L(k)>L(j)$. Second, every non-ground state in the window has a uniformly positive probability of decreasing by at least one layer:
    \begin{equation}
        \sum_{k:E_k\le E_j-\Delta E_{\mathrm{layer}}} P_{j\to k}\ \ge\ p_{\min}
        \qquad \text{for some }p_{\min}>0.
        \label{eq:uniform-drift}
    \end{equation}
\end{definition}
See~\cref{sec:eth-shell-mixing} for further details behind the connection between uniform drift and the ETH.

\subsection*{Continuation-path Hamiltonians and discretization}

A continuation path is a smooth map $\mathbf{R}:[0,1]\to\mathbb{R}^{3N}$ and $H(s)\coloneqq H_\mathrm{el}(\mathbf{R}(s))$.
We discretize $[0,1]$ by $0=s_0<s_1<\cdots<s_{N_H}=1$. A central quantity controlling drift of low-energy
projectors is the Davis--Kahan-type integral
\[
    C_{\mathrm{DK}}\coloneqq\int_0^1 \frac{\|\partial_s H(s)\|}{\Delta(s)}ds,
\]
defined in the main text (\cref{def:dk_constant}).

\subsection*{Filtered jump operators and leakage}

Given an unfiltered operator $A_a$ and a (time-domain) kernel $f$, we define the filtered operator
\begin{equation}
    \label{eq:filtered_jump_prelim}
    K_a\ \coloneqq\ \int_{\mathbb{R}} f(t)e^{itH}A e^{-itH}dt,
\end{equation}
with Fourier transform $\widehat f(\omega)=\int_{\mathbb{R}} f(t)e^{it\omega}dt$. In the eigenbasis of $H$, this yields $ \bra{E_j}K_a\ket{E_k}=\widehat f(E_j-E_k)\bra{E_j}A_a\ket{E_k}$. A \textit{cooling} filter is chosen so that $\widehat f(\omega)$ suppresses or removes contributions from heating Bohr frequencies (typically $\omega>0$), leading to small \textit{ground-state leakage} $\|K_f\ket{E_0}\|\le \varepsilon_{\mathrm{leak}}$.

\subsection*{Hitting time, layers, and warm starts}

We will use a layer discretization of energy. Fix $\Delta E_{\mathrm{layer}}>0$ and define the layer index
\begin{equation}
    L(j)\ \coloneqq\ \left\lceil \frac{E_j-E_0}{\Delta E_{\mathrm{layer}}}\right\rceil \in \mathbb{N}_0,
    \label{eq:app_layer}
\end{equation}
so $L(0)=0$. For the Markov chain $(J_t)_{t\ge 0}$ with kernel $P$, define the ground-state hitting time
\[
    \tau_0\ \coloneqq\ \inf\{t\in\mathbb{N}_0:\ J_t=0\}.
\]
A \textit{warm start} is an initial law $\mu$ such that most mass lies in low layers, e.g.,
$\Pr_{J_0\sim\mu}(L(J_0)>L_\mu)\le \eta$ for some $L_\mu$ and small $\eta$.

\section{Efficient ground-state preparation along continuation paths}
\label{sec:section1}

\cref{thm:reaction-path-runtime} establishes sufficient conditions for efficient ground-state preparation along continuation paths, particularly in the strongly correlated regime near transition states. The proof relies on decomposing the problem into three lemmas that collectively demonstrate an ideal cooling-step count linear in the number of orbitals $N_o$ (up to polylogarithmic factors) once the path discretization has been fixed.  In this section, we first formalize the main complexity result in \cref{thm:reaction-path-runtime} then provide proofs for all theorems and the supporting lemmas.

\begin{theorem}[Efficient ground-state preparation along continuation paths]
    \label{thm:reaction-path-runtime}
    Consider a continuation path whose electronic Hamiltonians have been represented in a locally compatible active-space envelope and Procrustes-aligned orbital gauge, so that the resulting embedded Hamiltonian path satisfies the Lipschitz condition in~\cref{def:regular-path}. Let the path be discretized into $N_H$ steps, with the final point corresponding to the target TS geometry. Under the conditions of \cref{lem:eth-informal} and \cref{lem:linear_cooling_informal}, for any target total energy error $\epsilon_E > 0$, the dissipative-continuation algorithm prepares a final state $\rho_{N_H}$ satisfying
    \[
        \bigl|\mathrm{tr}(H(s^\star)\rho_{N_H}) - E_0(s^\star)\bigr| \le \epsilon_E
    \]
    with a total ideal dissipative cooling-step complexity of
    \begin{equation}
        \widetilde{O}\!\left( \frac{C_{\mathrm{DK}}^2 N_o^2}{\epsilon_E} \right).
    \end{equation}
    Equivalently, for the intensive target precision $\epsilon_{\rm int}\coloneqq\epsilon_E/N_o$, the complexity is
    \begin{equation}
        \widetilde{O}\!\left( \frac{C_{\mathrm{DK}}^2 N_o}{\epsilon_{\rm int}} \right).
    \end{equation}
    Here $C_{\mathrm{DK}}$ is the geometric Davis--Kahan constant defined in~\cref{def:dk_constant}. The logical gate count additionally includes the cost of implementing each filtered Lindbladian step.
\end{theorem}

\subsection{Formal rank-one Davis--Kahan overlap bound}
\label{sec:dk_appendix}

\begin{lemma}[Formal rank-one Davis--Kahan overlap bound]
    \label{lem:davis-kahan}
    Let $H(s)$ be a Lipschitz-continuous continuation-path Hamiltonian as in \cref{def:regular-path}. Assume that $H(s)$ has a non-degenerate ground state $|E_0(s)\rangle$ and gap $\Delta(s)=E_1(s)-E_0(s)>0$ for all $s$ along the interval under consideration. Then, for any $0\le s_a<s_b\le 1$,
    \begin{equation}
        1-\bigl|\langle E_0(s_a)|E_0(s_b)\rangle\bigr|^2
        \le
        \left(
        \int_{s_a}^{s_b}\frac{\|\partial_s H(s)\|}{\Delta(s)}ds
        \right)^2 .
        \label{eq:rankone_dk_overlap_formal}
    \end{equation}
    Equivalently, the Fubini--Study distance $d_{FS}(\cdot, \cdot)$ between the two ground-state rays is bounded by the same integral without squaring.
\end{lemma}

\begin{proof}
    Choose the phase of $|E_0(s)\rangle$ by parallel transport, so that
    \[
        \langle E_0(s)|\partial_s E_0(s)\rangle=0 .
    \]
    Differentiating the eigenvalue equation
    $H(s)|E_0(s)\rangle=E_0(s)|E_0(s)\rangle$ and projecting onto an excited eigenstate $|E_j(s)\rangle$, $j>0$, gives
    \[
        \langle E_j(s)|\partial_s E_0(s)\rangle
        =
        \frac{\langle E_j(s)|\partial_s H(s)|E_0(s)\rangle}{E_0(s)-E_j(s)} .
    \]
    Hence
    \begin{align}
        \|\partial_s |E_0(s)\rangle\|^2
        &=
        \sum_{j>0}
        \frac{|\langle E_j(s)|\partial_s H(s)|E_0(s)\rangle|^2}{(E_j(s)-E_0(s))^2} \\
        &\le
        \frac{\|\partial_s H(s)|E_0(s)\rangle\|^2}{\Delta(s)^2}
        \le
        \frac{\|\partial_s H(s)\|^2}{\Delta(s)^2} .
        \label{eq:parallel_transport_velocity_bound}
    \end{align}
    Therefore the Fubini--Study path length obeys
    \[
        d_{\rm FS}\bigl(E_0(s_a),E_0(s_b)\bigr)
        \le
        \int_{s_a}^{s_b}\|\partial_s |E_0(s)\rangle\|ds
        \le
        \int_{s_a}^{s_b}\frac{\|\partial_s H(s)\|}{\Delta(s)}ds .
    \]
    For normalized vectors, after optimizing the relative phase,
    \[
        1-|\langle \psi|\phi\rangle|^2
        =\sin^2 d_{\rm FS}(\psi,\phi)
        \le d_{\rm FS}(\psi,\phi)^2 .
    \]
    Combining these inequalities proves~\cref{eq:rankone_dk_overlap_formal}.
\end{proof}

\subsection{Formal rapid cooling bound via ETH-motivated drift}
\label{sec:eth-shell-mixing}

For our dissipative cooling protocol, the relevant mechanism is that the chosen few-body jump operators $\{K_a\}$ induce a transition graph $G=(V,E)$ on the eigenstates $V=\{|E_i\rangle\}$, with an edge $(i,j)\in E$ whenever $\langle E_i|K_a|E_j\rangle\neq 0$ for some $a$ (See~\cref{fig:main_fig}a for an example transition graph visualized). In the multi-reference transition-state (TS) regime, the state support spans a vast subset of this graph. Provided the jump set avoids symmetry-protected bottlenecks (e.g., approximate conserved quantities, weakly coupled fragments, or selection-rule dead ends), the resulting Lindblad dynamics is expected to be \textit{rapidly mixing} within a narrow energy window, while simultaneously opening many \textit{energy-lowering} paths between adjacent shells. This expectation is most justified when the engineered jumps are explicitly energy-biased (e.g., approximately satisfy a detailed-balance/Davies condition, or otherwise preferentially induce downward transitions in energy).

The population update induced by evolving a quantum state according to a Lindblad equation defines a Markov kernel $P$ on energy-eigenstate indices $\{0,1,\dots,D-1\}$ via $p'_k=\sum_j p_j P_{j\to k}$, and we take the ground state to be absorbing. The analysis is phrased in terms of (i) a layer size $\Delta E_{\mathrm{layer}}=O(1)$ and (ii) a one-step downhill drift probability $p_{\min}$, defined below, which lower-bounds the chance of decreasing energy by at least one layer while the chain remains inside a microcanonical window $\mathcal{J}$. Fix a microcanonical window centered at energy $E_\star$ with width $\Delta E$, i.e.,
\[
    \mathcal{J}
    \coloneqq
    \{ j:\ E_j \in [E_\star-\tfrac{\Delta E}{2},E_\star+\tfrac{\Delta E}{2}] \},
    \qquad
    D_{\mathrm{shell}}\coloneqq|\mathcal{J}|,
\]
where for each $j,k\in \mathcal{J}$, write $E\coloneqq\tfrac{1}{2}(E_j+E_k)$ and $\omega\coloneqq E_k-E_j$.

For practical applications, we consider the jump set consisting of fermionic bilinears $A_{ij}^{(1)}=c_i^\dagger c_j+c_j^\dagger c_i$ and $A_{ij}^{(2)}=i(c_i^\dagger c_j-c_j^\dagger c_i)$ for $i<j$. ETH suggests that within such a shell the matrix elements of any few-body operator $A$ (in particular $A\in\{A^{(1)}_{ij},A^{(2)}_{ij}\}$) take the form
\begin{equation}
    \langle E_k|A|E_j\rangle
    = \bar{A}(E)\delta_{jk} + e^{-S(E)/2} f_A(E,\omega) R_{jk},
\end{equation}
where  $E=(E_i+E_j)/2$, $\bar{A}(E)$ is a smooth microcanonical (diagonal) expectation, $S(E)$ is the microcanonical entropy, $f_A(E,\omega)$ is a smooth envelope that varies only on $O(1)$ energy scales, and $R_{jk}$ is an $O(1)$ zero-mean random variable encoding state-to-state fluctuations~\cite{Deutsch_2018}. Since $e^{S(E_\star)}\sim D_{\mathrm{shell}}$ (up to subextensive constants set by $\Delta E$~\cite{DAlessio_2016}), for $j\neq k$ in the shell and for $\omega$ restricted to a bounded interval one has the scaling
\[
    |\langle E_k|A|E_j\rangle|^2 \sim e^{-S(E)} \sim \frac{1}{D_{\mathrm{shell}}}.
\]
Thus, after row-normalization, the induced transition probabilities $P_{j\to k}$ have no parametrically preferred targets within the energetically allowed set: for a fixed source state $j\in\mathcal{J}$, the probability mass is spread over $\Theta(D_{\mathrm{shell}})$ accessible states $k$, up to the smooth dependence encoded by $f_A(E,\omega)$.

Now suppose the filter function $f(t)$ that we employ to construct filtered versions of local jump operators is (approximately) one-sided, so that only downhill energy differences $\omega<0$ contribute appreciably.  After row-normalization, this implies that conditioned on being in an eigenstate $j$ in the shell, the induced transition probabilities $P_{j\to k}$ are roughly comparable across energetically allowed downhill targets $k$ (no strong bottlenecks), so the probability of a ``layer-success'' event
\[
    E_k \le E_j-\Delta E_{\mathrm{layer}}
\]
is controlled primarily by how many downhill targets satisfy this inequality. Equivalently, define the success set $S_j=\{k:\ E_k\le E_j-\Delta E_{\mathrm{layer}}\}$. ETH
suggests the approximation~\cite{DAlessio_2016}
\[
    \sum_{k\in S_j}P_{j\to k}
    \ \approx\
    \frac{\#\{k \text{ accessible}: E_k\le E_j-\Delta E_{\mathrm{layer}}\}}
     {\#\{k \text{ accessible}: E_k < E_j\}}
    \ \approx\
    \frac{\int_{-\infty}^{E_j-\Delta E_{\mathrm{layer}}}\rho(E)dE}
     {\int_{-\infty}^{E_j}\rho(E)dE},
\]
where $\rho(E)$ is the many-body density of states (restricted to the relevant symmetry sector). To estimate this ratio, write $\rho(E)=e^{S(E)}$ with entropy $S(E)$ and microcanonical inverse temperature $\beta(E)=\partial_E S(E)$.  For a fixed $O(1)$ energy decrement $\Delta E\equiv \Delta E_{\mathrm{layer}}$, expand
\[
    S(E_j-\Delta E)=S(E_j)-\beta(E_j)\Delta E+O(\Delta E^2 |S''(E_j)|).
\]
If $\beta(E)$ varies slowly across the microcanonical window (and $\Delta E=O(1)$), the quadratic correction is subleading and one obtains the standard density-of-states estimate
\[
    \frac{\rho(E_j-\Delta E)}{\rho(E_j)} \approx e^{-\beta(E_j)\Delta E}.
\]
Moreover, because the integrals above are dominated by their upper limits when $\rho(E)$ grows rapidly,
the same scaling carries over (up to constant prefactors) to the cumulative ratio:
\[
    \frac{\int_{-\infty}^{E_j-\Delta E}\rho(E)dE}{\int_{-\infty}^{E_j}\rho(E)dE}
    =\Omega\!\big(e^{-\beta(E_j)\Delta E}\big).
\]
Putting these steps together yields the heuristic constant-order uniform drift bound (condition),
\begin{equation}
    \label{equ:uni_drift}
    \sum_{k:E_k\le E_j-\Delta E_{\mathrm{layer}}}P_{j\to k}
    \geq p_{\min} = \Omega\!\big(e^{-\beta\Delta E_{\mathrm{layer}}}\big),
\end{equation}
throughout a microcanonical window in which $\beta(E)$ is approximately constant. The lemmas below take the existence of such a drift parameter $p_{\min}$ as an explicit assumption and convert it into a quantitative cooling-time bound.

\begin{lemma}[\cref{lem:eth-informal} (Formal)]
    \label{lem:warmstart-cooling-eth-perfect}
    Consider a Markov Chain $(J_t)_ {t \geq 0}$ with associated kernel $P$, hitting time $\tau_0$, and initial law $\mu$ that is concentrated in low energy layers for a fixed layer width $\Delta E_\mathrm{layer} >0$. That is, let $E_j$ be an energy assigned to state $j$, with a \textit{unique} ground state $0$ in the sense that $E_j>E_0$ for all $j\neq 0$, and define the layer index as in \cref{eq:app_layer}, so that initial law (warm-start) satisfies the layer-radius condition  $\Pr_{J_0\sim\mu}(L(J_0)>L_\mu)\le \eta$ for some $\eta\in(0,1)$ and $L_\mu\in\mathbb{N}$.
    If $P$ satisfies the local \textit{uniform downward-drift condition}, then, for any $0<\varepsilon<1$, if
    \begin{equation}
        t\ \ge\ \frac{2L_\mu}{p_{\min}}+\frac{8}{p_{\min}}\ln\frac{1}{\varepsilon},
        \label{eq:t-sufficient}
    \end{equation}
    one has $\Pr(\tau_0>t)\ \le\ \eta+\varepsilon.$ More precisely, for all $t\ge 2L_\mu/p_{\min}$,
    \begin{equation}
        \Pr(\tau_0>t)\ \le\ \eta+\exp\Bigl(-\tfrac{p_{\min}}{8}t\Bigr).
    \end{equation}
\end{lemma}
\begin{proof}
    With the layer index defined above, define the one-step success set
    \begin{equation}
        S_j\coloneqq\{k:\ E_k\le E_j-\Delta E_{\mathrm{layer}}\}
    \end{equation}
    for each state $j$, and the corresponding one-step success indicator
    \begin{equation}
        X_s\coloneqq\mathbf{1}\{J_{s+1}\in S_{J_s}\},
        \qquad s\ge 0.
    \end{equation}
    If $J_{s+1}\in S_{J_s}$, then $E_{J_{s+1}}\le E_{J_s}-\Delta E_{\mathrm{layer}}$, hence
    \begin{equation}
        L(J_{s+1})\le L(J_s)-X_s.
        \label{eq:steps}
    \end{equation}
    Under perfectly downhill dynamics, $L(J_s)$ is non-increasing in $s$. Define the one-step success probability
    \begin{equation}
        p(j)\coloneqq\Pr(J_{s+1}\in S_j\mid J_s=j)=\sum_{k\in S_j}P_{j\to k}.
        \label{eq:1step}
    \end{equation}
    By~\cref{eq:uniform-drift}, $p(j)\ge p_{\min}$ for every $j\neq 0$.
    
    Let $(U_s)_{s\ge 0}$ be i.i.d.\ uniform random variables on $[0,1]$.
    By the standard inverse-transform construction, we may realize each transition as a measurable function of $U_s$ as follows: conditional on $J_s=j$, partition $[0,1]$ into intervals of lengths $\{P_{j\to k}\}_k$, placing the outcomes $k\in S_j$ first. Then~\cref{eq:1step} ensures
    \begin{equation*}
        \{U_s\le p(j)\}\ \Longleftrightarrow\ \{J_{s+1}\in S_j\}.
    \end{equation*}
    Define the Bernoulli benchmark
    \begin{equation}
        X'_s\coloneqq\mathbf{1}\{U_s\le p_{\min}\}.
    \end{equation}
    Since $p(J_s)\ge p_{\min}$ for all $J_s\neq 0$, we have the pathwise domination
    \begin{equation}
        X_s\ \ge\ X'_s \qquad \text{on every time } s \text{ for which } J_s\neq 0.
    \end{equation}
    The variables $(X'_s)_{s\ge 0}$ are i.i.d.\ Bernoulli$(p_{\min})$, so $\sum_{s=0}^{t-1}X'_s\sim \mathrm{Binomial}(t,p_{\min})$. In applying this comparison below we always intersect with $\{\tau_0>t\}$; on that event $J_s\neq0$ for all $s<t$, so the domination holds for every summand used in the comparison.
    
    Iterating~\cref{eq:steps} along a sample path gives
    \begin{equation}
        L(J_t)\ \le\ L(J_0)-\sum_{s=0}^{t-1}X_s.
    \end{equation}
    If $\sum_{s=0}^{t-1}X_s\ge L(J_0)$ then $L(J_t)\le 0$. Since $L(\cdot)\ge 0$, it follows that $L(J_t)=0$. By uniqueness of the ground state (i.e.\ $L(j)=0\iff j=0$), this implies $J_t=0$, so the process must have hit the ground state by time $t$, i.e.\ $\tau_0\le t$.
    Equivalently,
    \begin{equation}
        \{\tau_0>t\}\ \subseteq\ \left\{\sum_{s=0}^{t-1}X_s<L(J_0)\right\}.
        \label{eq:tau-subset}
    \end{equation}
    
    By the warm-start layer-radius condition,
    \begin{align}
        \Pr(\tau_0>t)
        &\le
        \Pr\big(L(J_0)>L_\mu\big)
        +
        \Pr\big(\tau_0>t,\ L(J_0)\le L_\mu\big)\nonumber\\
        &\le
        \eta
        +
        \Pr\!\left(
            \tau_0>t,\ 
            \sum_{s=0}^{t-1}X_s<L(J_0),\
            L(J_0)\le L_\mu
        \right)\nonumber\\
        &\le
        \eta
        +
        \Pr\!\left(
            \tau_0>t,\ 
            \sum_{s=0}^{t-1}X'_s<L_\mu
        \right)\nonumber\\
        &\le
        \eta
        +
        \Pr\!\left(\sum_{s=0}^{t-1}X'_s<L_\mu\right)\nonumber\\
        &=
        \eta
        +
        \Pr\!\left(\mathrm{Bin}(t,p_{\min})<L_\mu\right),
        \label{eq:warmstart-split}
    \end{align}
    where the penultimate step uses the preceding domination on the event $\{\tau_0>t\}$. If $t \ge 2 L_\mu / p_{\min}$ then $L_\mu \le \tfrac{1}{2} p_{\min} t$, and by a standard Chernoff bound,
    \begin{equation}
        \Pr\!\left(\mathrm{Bin}(t,p_{\min}) \le \tfrac{1}{2} p_{\min} t\right)
        \ \le\
        \exp\left(-\frac{p_{\min} t}{8}\right).
    \end{equation}
    Thus for all $t\ge 2L_\mu/p_{\min}$,
    \begin{equation}
        \Pr(\tau_0>t)\ \le\ \eta+\exp\Bigl(-\tfrac{p_{\min}}{8}t\Bigr).
    \end{equation}
    Finally, choosing $t \ge \frac{8}{p_{\min}}\ln(1/\varepsilon)$ makes the exponential term at most $\varepsilon$, and together with $t\ge 2L_\mu/p_{\min}$ yields~\cref{eq:t-sufficient} and the claim.
\end{proof}

At step $k$ of the continuation-path discretization, our warm start for $H(s_k)$ is approximately the previous ground state $\ket{E_0(s_{k-1})}$. If the continuation path is Lipschitz continuous in the sense of~\cref{def:regular-path}, then $H(s)$ has $M$ terms
\begin{equation}
    H(s)=\sum_{a=1}^{M} h_a(s),
    \label{eq:local_terms}
\end{equation}
such that each term is uniformly Lipschitz continuous along the path,
\begin{equation}
    \|h_a(s)-h_a(s')\|\le L_h|s-s'|
    \qquad \text{for all } a \text{ and } s,s'\in[0,1].
    \label{eq:lipschitz}
\end{equation}
For chemicals in nature, we observe  $M=O(N_o)$~\cite{Motta_2021}. Therefore, for $\delta_k\coloneqq|s_k-s_{k-1}|$ we have
\[
    \|H(s_k)-H(s_{k-1})\|
    \le\sum_{a=1}^{M}\|h_a(s_k)-h_a(s_{k-1})\|
    \le M L_h\delta_k
    =O(N_o\delta_k).
\]
Consequently, the warm-start excess energy at the new geometry is bounded by
\begin{align}
    0 \le \langle E_0(s_{k-1})|H(s_k)|E_0(s_{k-1})\rangle - E_0(s_k)
    &\le 2\|H(s_k)-H(s_{k-1})\| \nonumber\\
    &\le 2 M L_h\delta_k
    =O(N_o\delta_k).
    \label{eq:warmstart-excess-energy-extensive}
\end{align}
Thus, for an equal-mass (or otherwise bounded-step) mesh with $\delta_k\le \delta_{\max}=O(1)$,
the warm start lies within an extensive energy window:
\[
    \langle H(s_k)\rangle_{E_0(s_{k-1})}-E_0(s_k)\leq c_E N_o,
    \qquad c_E\coloneqq2(\tfrac{M}{N_o})L_h\delta_{\max}.
\]
If $\mu$ is the dephased (energy-eigenbasis) distribution induced by this warm start at $s_k$,
then $\mathbb{E}_\mu[E_{J_0}-E_0]=\langle H(s_k)\rangle_{E_0(s_{k-1})}-E_0(s_k)$ is $O(N_o)$. This motivates \cref{lem:warmstart-cooling-eth-perfect} which shows how we arrive at a mixing time linear in the number of orbitals $N_o$.

\begin{lemma}[\cref{lem:linear_cooling_informal} (Formal)]
    \label{lem:linear-cooling-extensive}
    Adopt the setting of \cref{lem:warmstart-cooling-eth-perfect}. In our regime, the warm start is supported on an extensive energy range: there exists $c_E>0$ such that  $\Pr_{\mu}\!\big(E_{J_0}-E_0 \le c_E N_o\big)\ \ge\ 1-\eta.$
    
    Then the layer radius may be chosen as $L_\mu \ \le\ \left\lceil \frac{c_E N_o}{\Delta E_{\mathrm{layer}}}\right\rceil,$
    and consequently, for any $0<\varepsilon<1$ it suffices to take
    \begin{equation}
        t\ \ge\ \frac{2}{p_{\min}}\left\lceil \frac{c_E N_o}{\Delta E_{\mathrm{layer}}}\right\rceil
        +\frac{8}{p_{\min}}\ln\frac{1}{\varepsilon}.
        \label{eq:t-linear}
    \end{equation}
    In particular, the cooling time
    is linear in $N_o$,
    \begin{equation}
        t = O\!\left(N_o+\ln\frac{1}{\varepsilon}\right).
    \end{equation}
\end{lemma}
\begin{proof}
    On the event $\{E_{J_0}-E_0 \le c_E N_o\}$ one has
    $L(J_0)=\left\lceil\frac{E_{J_0}-E_0}{\Delta E_{\mathrm{layer}}}\right\rceil
    \le \left\lceil\frac{c_E N_o}{\Delta E_{\mathrm{layer}}}\right\rceil$.
    Thus,
    \begin{equation}
        \Pr_{\mu}\!\big(E_{J_0}-E_0 \le c_E N_o\big)\ \ge\ 1-\eta,
        \label{eq:extensive-support}
    \end{equation}
    implies \cref{lem:warmstart-cooling-eth-perfect} with chosen
    \begin{equation}
        L_\mu \ \le\ \left\lceil \frac{c_E N_o}{\Delta E_{\mathrm{layer}}}\right\rceil,
        \label{eq:Lmu-choice}
    \end{equation}
    Substituting this into the sufficient condition
    \cref{eq:t-sufficient} from \cref{lem:warmstart-cooling-eth-perfect} yields~\cref{eq:t-linear}
    and the stated scaling.
\end{proof}

\subsection*{Proof of Theorems 1 \& 3}
\label{sec:proof_of_thm1}

\begin{proof}[Proof of \cref{thm:reaction-path-runtime}]
    We now combine \cref{lem:davis-kahan} (geometric Davis--Kahan control of ground-state drift) and \Cref{lem:warmstart-cooling-eth-perfect,lem:linear-cooling-extensive}  to bound the total number of dissipative blocks required to follow a continuation path, and thereby the ideal dissipative-step complexity of this algorithm. Choose a discretization $0=s_0<s_1<\cdots<s_{N_H}=1$ of the continuation coordinate. Starting from an initial warm start $\rho_0$ at $s_0$, for each $i=0,\dots,N_H-1$ apply the dissipative channel at geometry $s_i$ for $t_i$ steps to obtain $\rho_{i+1}$, which is then used as a warm start for the next geometry $s_{i+1}$. The goal is to achieve a bounded total energy error $|\mathrm{tr}(H(1)\rho_{N_H})-E_0(1)|\le \epsilon_E$ for some $\epsilon_E > 0$. Define the resolved geometric density
    \[
        g(s)\ \coloneqq\ \frac{\|\partial_s H(s)\|}{\Delta(s)},
        \qquad
        C_\mathrm{DK}\ \coloneqq\ \int_0^1 g(s)ds.
    \]
    Choose an (approximate) equal-mass mesh,
    \begin{equation}
        \int_{s_i}^{s_{i+1}} g(s)ds \ \le\ \frac{C_\mathrm{DK}}{N_H},
        \qquad \forall i,
        \label{eq:equal-mass-mesh-res}
    \end{equation}
    then \cref{lem:davis-kahan} directly gives a uniform overlap bound between adjacent ground states:
    \begin{equation}
        1-|\langle E_0(s_i)|E_0(s_{i+1})\rangle|^2
        \le
        \left(\frac{C_{\mathrm{DK}}}{N_H}\right)^2
        =:
        \varepsilon_{dR}.
        \label{eq:epsdR-runtime}
    \end{equation}
    For each geometry $s_i$ define the excited population (energy-basis tail)
    \[
        \delta_i \coloneqq 1-\langle E_0(s_i)|\rho_i|E_0(s_i)\rangle.
    \]
    This quantity is exactly the total population mass outside the absorbing state in the Markov-chain picture of \cref{lem:warmstart-cooling-eth-perfect} (equivalently, if one measures $\rho_i$ in the energy basis and draws an index $J\sim p$, then $\delta_i=\Pr(J\neq 0)$). It upper-bounds the energy error via
    \begin{equation}
        \mathrm{tr}(H(s_i)\rho_i)-E_0(s_i)
        =\sum_{j>0} (E_j(s_i)-E_0(s_i))p_j
        \le
        \Big(\max_{j>0}\{E_j(s_i)-E_0(s_i):p_j>0\}\Big)\delta_i,
        \label{eq:energy-vs-excited-mass}
    \end{equation}
    where $p_j\coloneqq\langle E_j(s_i)|\rho_i|E_j(s_i)\rangle$. Under the warm-start spectral extensivity assumption $\Pr_{\mu_i}(E_{J_0}(s_i)-E_0(s_i)\le c_E N_o)\ge 1-\eta$
    (cf.\ \cref{lem:linear-cooling-extensive}), we may further bound
    \begin{equation}
        \mathrm{tr}(H(s_i)\rho_i)-E_0(s_i)\ \le\ c_E N_o\delta_i
        \qquad\text{except on an event of probability at most }\eta.
        \label{eq:energy-vs-excited-mass-ext}
    \end{equation}
    The geometric bound~\cref{eq:epsdR-runtime} implies that even if $\rho_i$ were perfectly cooled at $s_i$, when reinterpreted at $s_{i+1}$ it acquires at most
    $\varepsilon_{dR}$ additional excited weight (up to universal constants). Thus, suppressing constant factors, we may propagate the schematic recurrence
    \begin{equation}
        \delta_{i+1}
        \le
        (\text{cooling residual at } s_i)
        +
        \varepsilon_{dR}.
        \label{eq:delta-recurrence}
    \end{equation}
    
    Fix a target per-step cooling error parameter $\varepsilon_{\mathrm{mix}}$. At each geometry $s_i$, \cref{lem:warmstart-cooling-eth-perfect} provides a step count $t_i$
    such that after $t_i$ applications of the dissipative update, the remaining excited population mass is at most $\eta+\varepsilon_{\mathrm{mix}}$, where $\eta$ is the
    warm-start layer overflow probability
    \begin{equation}
        \Pr_{\mu_i}(L(J_0)>L_{\mu_i})\le \eta,
    \end{equation}
    with $L_{\mu_i}$ the $\eta$-quantile layer radius at $s_i$.
    Concretely, choosing
    \begin{equation}
        t_i
        =
        \left\lceil
        \frac{2L_{\mu_i}}{p_{\min}}+\frac{8}{p_{\min}}\ln\frac{1}{\varepsilon_{\mathrm{mix}}}
        \right\rceil
        \label{eq:ti-choice}
    \end{equation}
    ensures
    \begin{equation}
        \delta_{i+1}
        \le
        (\eta+\varepsilon_{\mathrm{mix}}) + \varepsilon_{dR}.
        \label{eq:delta-step-bound}
    \end{equation}
    Here $p_{\min}$ is the one-step drift lower bound from \cref{lem:warmstart-cooling-eth-perfect}.
    
    Iterating~\cref{eq:delta-step-bound} over $N_H$ steps gives
    \[
        \delta_{N_H}
        \le
        N_H(\eta+\varepsilon_{\mathrm{mix}}) + N_H\varepsilon_{dR}
        =
        N_H(\eta+\varepsilon_{\mathrm{mix}}) + \frac{C_{\mathrm{DK}}^2}{N_H}.
    \]
    Combining with~\cref{eq:energy-vs-excited-mass-ext} yields, up to the same exceptional
    event of probability at most $\eta$,
    \[
        |\mathrm{tr}(H(1)\rho_{N_H})-E_0(1)|
        \le
        c_E N_o\left(
        N_H(\eta+\varepsilon_{\mathrm{mix}}) + \frac{C_{\mathrm{DK}}^2}{N_H}
        \right).
    \]
    Accordingly, one convenient choice is
    \begin{equation}
        N_H=\Theta\!\left(\frac{c_E N_o C_{\mathrm{DK}}^2}{\epsilon_E}\right),
        \qquad
        \varepsilon_{\mathrm{mix}}=\Theta\!\left(\frac{\epsilon_E}{c_E N_o N_H}\right),
        \qquad
        \eta=\Theta\!\left(\frac{\epsilon_E}{c_E N_o N_H}\right),
        \label{eq:param-choice-runtime}
    \end{equation}
    Equivalently, for an intensive target error $\epsilon_{\rm int}\coloneqq \epsilon_E/N_o$, this is $N_H=\Theta(c_E C_{\mathrm{DK}}^2/\epsilon_{\rm int})$.
    
    Hence, the total number of dissipative steps (i.e., applications of the dissipative update) is bounded by
    \begin{equation}
        \sum_{i=0}^{N_H-1} t_i
        =
        O\!\left(
        \frac{N_H}{p_{\min}}
        \left(
          \overline{L}_\mu + \ln\frac{N_H}{\epsilon_E}
        \right)
        \right),
        \label{eq:total-steps-runtime}
    \end{equation}
    where $\overline{L}_\mu\coloneqq\max_i L_{\mu_i}$.
    Substituting $N_H=\Theta(c_E N_o C_{\mathrm{DK}}^2/\epsilon_E)$ yields the compact end-to-end scaling
    \begin{equation}
        \sum_{i=0}^{N_H-1} t_i
        =
        O\!\left(
        \frac{C_{\mathrm{DK}}^2 N_o\,\overline{L}_\mu}{p_{\min} \epsilon_E}
        \right).
        \label{eq:runtime-compact}
    \end{equation}
    When the warm-start spectral extensivity condition holds uniformly, $\overline{L}_\mu=O(N_o)$ (\cref{lem:linear-cooling-extensive}) and $p_{\min}=\Theta(1)$. Simulating dissipative dynamics for evolution time $t$ on quantum computers requires $\widetilde{O}(t)$ logical quantum resources~\cite{li2023_lineartime,chen2023quantum}. Therefore, the total ideal cooling-step complexity for total energy error is
    \begin{equation}
        \widetilde{O}\!\left(
        \frac{C_{\mathrm{DK}}^2 N_o^2}{\epsilon_E}
        \right)
        =
        \widetilde{O}\!\left(
        \frac{C_{\mathrm{DK}}^2 N_o}{\epsilon_{\rm int}}
        \right).
    \end{equation}
\end{proof}

For a Hamiltonian along the continuation path as a sum of $M$ terms defined in~\cref{eq:local_terms} with each term varying smoothly along the path defined by the Lipschitz constant $L_h$ in~\cref{eq:lipschitz}, if the derivative exists, then $\|\partial_s h_\ell(s)\|\le L_h$. It follows that $\|\partial_s H(s)\|\leq ML_h$. Now, we have that
\begin{equation}
    \label{equ:cdk_bound}
    C_\mathrm{DK} = \int_{0}^1\frac{||\partial_s H(s)||}{\Delta (s)}ds\leq \int_0^1\frac{ML_h}{\Delta_{\mathrm{min}}}ds=O(N_o/\Delta_{\min}).
\end{equation}
Therefore, the coarse ideal cooling-step count is
\begin{equation}
    \widetilde{O}\left(
    \frac{N_o^4}{\Delta_{\min}^2\epsilon_E}
    \right)
    =
    \widetilde{O}\left(
    \frac{N_o^3}{\Delta_{\min}^2\epsilon_{\rm int}}
    \right),
    \qquad \epsilon_{\rm int}=\epsilon_E/N_o.
\end{equation}
The logical gate complexity additionally includes the cost of implementing each filtered Lindbladian step. In the filter constructions analyzed in~\cref{sec:das_vs_diss,sec:thm2}, this per-step cost can include Hamiltonian-time factors such as $\widetilde{O}(\|H\|/\Delta_{\min})$, which should not be folded into the ideal cooling-step count.

\section{NTK-style justification of downward drift}
\label{app:ntk_drift}

This appendix records a standard ``lazy training'' argument, adapted from neural tangent kernel (NTK) analyses of Ref.~\cite{jacot2018neural}, to justify a uniform downward drift condition in our notation.  The result is conditional: it holds in a regime where (i) the relevant kernel concentrates and (ii) optimization remains in a neighborhood where linearization is accurate. Such conditions are expected to hold when the initial law is highly concentrated in low energies.

Fix a Hamiltonian $H$ (or $H(s)$ at a fixed continuation point) and a differentiable ansatz map
\begin{equation}
    \theta \mapsto \rho_\theta,
\end{equation}
where $\rho_\theta$ is a density operator on the system register. Define the energy objective
\begin{equation}
    \mathcal{L}(\theta) \coloneqq  \Tr\!\big[H\,\rho_\theta\big].
\end{equation}
Let $\theta(t)$ evolve under gradient flow
\begin{equation}
    \dot\theta(t) = -\nabla_\theta \mathcal{L}(\theta(t)).
    \label{eq:gradflow}
\end{equation}

For a finite set of observables (or ``features'') $\{\phi_a\}_{a=1}^m$ on states, define the feature vector
\begin{equation}
    f_\theta \in \mathbb{R}^m,\qquad (f_\theta)_a \coloneqq  \Tr\!\big[\phi_a \rho_\theta\big].
\end{equation}
In the simplest case one can take $m=1$ with $\phi_1 = H$, in which case $f_\theta = \mathcal{L}(\theta)$; we keep general $m$ because it matches the NTK proof structure.

Define the empirical kernel (Jacobian Gram matrix)
\begin{equation}
    K_\theta \in \mathbb{R}^{m\times m},\qquad (K_\theta)_{ab} \coloneqq  \langle \nabla_\theta (f_\theta)_a,\ \nabla_\theta (f_\theta)_b\rangle.
    \label{eq:kernel_def}
\end{equation}
By construction $K_\theta \succeq 0$.

    \begin{lemma}[Kernel-governed dynamics]
    \label{lem:kernel_dynamics}
    Let $r_\theta \coloneqq  f_\theta - f_{\theta^\star}$ denote the feature residual relative to any reference $\theta^\star$. Under \eqref{eq:gradflow},
    \begin{equation}
        \frac{d}{dt} r_{\theta(t)} = - K_{\theta(t)} \, r_{\theta(t)}.
        \label{eq:residual_ode}
    \end{equation}
\end{lemma}

\begin{proof}
    Differentiate $(f_{\theta(t)})_a$ and use the chain rule:
    \[
        \frac{d}{dt}(f_{\theta(t)})_a = \langle \nabla_\theta (f_\theta)_a,\ \dot\theta(t)\rangle
        = - \sum_b \langle \nabla_\theta (f_\theta)_a,\ \nabla_\theta (f_\theta)_b\rangle\, (r_{\theta(t)})_b.
    \]
    This is exactly \eqref{eq:residual_ode}.
\end{proof}

Equation \eqref{eq:residual_ode} is the analogue of the NTK evolution equation: if $K_{\theta(t)}$ is approximately constant, the dynamics are linear and contractive. We state two assumptions that are direct analogues of standard NTK conditions.

\begin{assumption}[Kernel concentration at initialization]
    \label{ass:kernel_concentration}
    There exists a deterministic PSD matrix $K_0$ and $\delta_K \ge 0$ such that with high probability over initialization $\theta(0)$,
    \begin{equation}
    \|K_{\theta(0)} - K_0\|_\mathrm{op} \le \delta_K.
    \end{equation}
\end{assumption}

\begin{assumption}[Kernel stability in a neighborhood]
    \label{ass:kernel_stability}
    There exists a radius $R>0$ and a Lipschitz constant $L_K$ such that for all $\theta$ with $\|\theta-\theta(0)\|_2 \le R$,
    \begin{equation}
        \|K_\theta - K_{\theta(0)}\|_\mathrm{op} \le L_K \,\|\theta-\theta(0)\|_2.
    \end{equation}
\end{assumption}

In NTK analyses, overparameterization is used to make $\delta_K$ small and to ensure the trajectory remains inside the radius $R$ for the time needed to reduce the residual.

\begin{theorem}[Contractive dynamics and uniform descent in the NTK regime~\cite{jacot2018neural}]
    \label{thm:ntk_descent}
    Suppose the gradient-flow trajectory satisfies $\|\theta(t)-\theta(0)\|_2 \le R$ for all $t \in [0,T]$, under Assumptions 1 and 2. If
    \begin{equation}
        \lambda_{\min}(K_0) > \delta_K + L_K R,
        \label{eq:pd_margin}
    \end{equation}
    then for all $t \in [0,T]$ the residual decays exponentially:
    \begin{equation}
        \|r_{\theta(t)}\|_2 \le \exp\!\Big(-(\lambda_{\min}(K_0) - \delta_K - L_K R)\,t\Big)\, \|r_{\theta(0)}\|_2.
    \end{equation}
    In particular, the feature error (and hence any linear objective derived from it, including energy) exhibits a monotone ``downward drift'' under gradient flow.
\end{theorem}

\begin{proof}
    By \Cref{lem:kernel_dynamics},
    \[
        \frac{d}{dt}\|r\|_2^2 = -2\, r^\top K_{\theta(t)} r \le -2\,\lambda_{\min}(K_{\theta(t)})\, \|r\|_2^2.
    \]
    It suffices to lower bound $\lambda_{\min}(K_{\theta(t)})$. Under Assumptions 1 and 2 and $\|\theta(t)-\theta(0)\| \le R$,
    \[
        \|K_{\theta(t)} - K_0\|_\mathrm{op}
        \le \|K_{\theta(t)} - K_{\theta(0)}\|_\mathrm{op} + \|K_{\theta(0)} - K_0\|_\mathrm{op}
        \le L_K R + \delta_K.
    \]
    Weyl's inequality then gives
    \[
        \lambda_{\min}(K_{\theta(t)}) \ge \lambda_{\min}(K_0) - (L_K R + \delta_K).
    \]
    Plugging into the differential inequality and applying Gr\"onwall~\cite{Gronwall1919} yields the exponential decay.
\end{proof}

In our setting, the margin condition \eqref{eq:pd_margin} is the precise analogue of ``a good warm-start avoids bad local geometry'': it does not assert the absence of local minima in general, but rather that, within a neighborhood, the linearized dynamics are uniformly contractive. This is exactly the type of condition that implies a uniform downward drift.

\section{Digital Adiabatic Simulation versus Dissipative Continuation}
\label{sec:das_vs_diss}

This section compares (i) large-step digital adiabatic simulation (DAS) along the same nonlinear continuation paths $\{H(s)\}_{s\in[0,1]}$ and (ii) our \textit{dissipative continuation} method that transports a warm start between adjacent geometries and re-cools at each step. The key point is that both approaches exploit the smoothness of the continuation path, but they pay for it in different places: DAS pays heavily for small gaps through adiabatic constraints, whereas dissipative continuation pays through the number of path points $N_H$ (controlled by
$C_{\mathrm{DK}}$) and the per-step mixing time (controlled by the uniform drift parameter $p_{\min}$ and the target failure probability).

\subsection{Setup and shared parameters}
Fix consecutive grid points $s_i<s_{i+1}$ and define the local interval Hamiltonian norm
\[
    \|H\|\ \coloneqq\ \max_{s\in[s_i,s_{i+1}]}\|H(s)\|,
    \qquad
    \Delta_{\min}\ \coloneqq\ \min_{s\in[s_i,s_{i+1}]}\Delta(s),
\]
where $\Delta(s)$ is the ground-state gap at geometry $s$. Throughout, $N_o$ denotes the number of \textit{spatial} orbitals. We compare \textit{per-segment} cost to transport from $s_i$ to $s_{i+1}$ and then multiply by the number of segments $N_H$ required to reach $s=1$.

Independent of the transport mechanism, our path-following regime requires a grid fine enough that successive ground states retain overlap. \cref{lem:davis-kahan} implies that an ``equal-mass'' mesh in the DK density $g(s)\coloneqq\|\partial_s H(s)\|/\Delta(s)$ yields a per-step overlap loss $\varepsilon_{dR} = \Theta\left((C_{\mathrm{DK}}/N_H)^2\right)$, hence choosing
\begin{equation}
    N_H\ =\ \Theta\!\left(\frac{c_E N_o C_{\mathrm{DK}}^2}{\epsilon_E}\right)
    \label{eq:NH_choice_das_diss}
\end{equation}
is sufficient to keep total discretization-induced energy error under control; equivalently, for intensive precision $\epsilon_{\rm int}=\epsilon_E/N_o$, this is $N_H=\Theta(c_E C_{\mathrm{DK}}^2/\epsilon_{\rm int})$. When $\|\partial_s H(s)\|=O(N_o)$ and $\Delta(s)\ge \Delta_{\min}$ along the path, one typically has $C_{\mathrm{DK}}=O(N_o/\Delta_{\min})$, giving
$N_H=\widetilde{O}(c_E N_o^3/(\Delta_{\min}^2\epsilon_E))$ for total energy precision.

\subsection{Per-segment cost of dissipative continuation}

Our algorithm uses a short-time dissipative primitive (a filtered-jump Lindbladian step) at geometry $s_{i+1}$ initialized by a warm start transported from $s_i$ (e.g., via Procrustes-aligned orbital rotations plus state reuse). Let $\varepsilon$ denote the target failure probability for cooling at that segment. Under the conditions of \Cref{lem:warmstart-cooling-eth-perfect,lem:linear-cooling-extensive} provide a per-segment mixing time
\begin{equation}
    t_{\mathrm{mix}}
    \ =\
    O\!\left(N_o+\ln\frac{1}{\varepsilon}\right),
    \label{eq:tmix_segment}
\end{equation}
for fixed layer size $\Delta E_{\mathrm{layer}}$ and drift parameter $p_{\min}$.

Each dissipative step is implemented using filtered jump operators, whose circuit depth is governed by the filter time support needed to achieve leakage $\varepsilon_{\mathrm{leak}}$ (\cref{thm:filter-scaling}). Up to polylogarithmic overheads from Lindbladian simulation, the gate cost per dissipative step scales as the number of distinct real-time evolutions $e^{\pm iHt}$ required by the filter:
\begin{equation}
    \mathrm{Cost}_{\mathrm{step}}
    \ =\ \widetilde{O}\!\Big(S(\varepsilon_{\mathrm{leak}})\cdot \mathrm{Cost}(U_H)\Big),
\end{equation}
where $U_H$ is the chosen block encoding / time-evolution primitive. Using the two filter families:
\begin{align}
    S_{\mathrm{SW}}(\varepsilon_{\mathrm{leak}})
    &=\ O\!\left(\frac{\|H\|}{\Delta_{\min}}\cdot\frac{1}{\varepsilon_{\mathrm{leak}}}\right),
    \\
    S_{\mathrm{G}}(\varepsilon_{\mathrm{leak}})
    &=\ O\!\left(\frac{\|H\|}{\Delta_{\min}}\cdot[\ln(1/\varepsilon_{\mathrm{leak}})]^{\alpha}\right),
    \qquad \alpha>1.
\end{align}
Thus, the \textit{per-segment} runtime complexity of dissipative continuation is
\begin{equation}
    G_{\mathrm{diss}}
    \ =\
    \widetilde{O}\!\left(
    \frac{\|H\|}{\Delta_{\min}}\cdot \mathsf{F}(\varepsilon_{\mathrm{leak}})
    \cdot \Big(N_o+\ln\frac{1}{\varepsilon}\Big)
    \right),
    \label{eq:diss_segment_cost}
\end{equation}
where $\mathsf{F}(\varepsilon_{\mathrm{leak}})=1/\varepsilon_{\mathrm{leak}}$ for the square-wave (Fourier) filter and $\mathsf{F}(\varepsilon_{\mathrm{leak}})=[\ln(1/\varepsilon_{\mathrm{leak}})]^{\alpha}$ for the Gevrey
filter.

Dissipative continuation separates two ``hardness knobs'': (i) spectral hardness enters \textit{linearly} as $\|H\|/\Delta_{\min}$ through filtering, while (ii) system-size hardness enters through the \textit{mixing time} $O(N_o)$ (from extensivity) and through the \textit{grid size} $N_H$ (from $C_{\mathrm{DK}}$). This is qualitatively different from DAS, where $\Delta_{\min}$ controls the required adiabatic runtime more directly.

\subsection{Per-segment cost of large-step digital adiabatic simulation}

Large-step DAS aims to prepare $\ket{E_0(s_{i+1})}$ by simulating an adiabatic interpolation on $s\in[s_i,s_{i+1}]$ with a small number of discrete steps, using a large step size $h=\Theta(1/(\|H(s_i)\|+\|H(s_{i+1})\|))$ and an improved analysis of adiabatic error~\cite{Berry2025Large}. Abstracting away polylog factors and implementation details, the dominant dependence is that adiabatic state-preparation cost grows rapidly as the gap decreases. In the scaling form used in the main text, the \textit{per-segment} DAS runtime complexity is
\begin{equation}
    G_{\mathrm{DAS}}
    \ =\
    \widetilde{O}\!\left(
    \frac{\|H\|^{3}}{\Delta_{\min}^{3}}\cdot\frac{1}{\varepsilon}\cdot \mathrm{poly}(N_o)
    \right),
    \label{eq:das_segment_cost}
\end{equation}
where $\varepsilon$ is the target adiabatic error/failure parameter, and $\mathrm{poly}(N_o)$ accounts for the cost of simulating $e^{-ihH(s)}$ and implementing the schedule. (More refined versions would track the dependence on $\|\partial_s H\|$ and higher derivatives, but~\cref{eq:das_segment_cost} captures the key spectral-gap sensitivity relevant to transition-state regimes.) The estimate in~\cref{eq:das_segment_cost} is used here as a conservative segment-wise comparison. It applies the large-step DAS scaling to each mesh interval and then multiplies by the number of continuation points. This is sufficient for a worst-case comparison of sensitivities to $\Delta_{\min}$, but it should not be read as the sharpest possible end-to-end DAS bound for a globally specified smooth path.

DAS pays a \textit{polynomially higher} price in $1/\Delta_{\min}$ per segment than dissipative continuation. This distinction is most pronounced near transition states, where small gaps are common and where our method is designed to reuse a warm start so that the dissipative dynamics rapidly returns to the instantaneous ground state.

\subsection{End-to-end costs along the full continuation path}

Multiplying by the number of segments $N_H$ in~\cref{eq:NH_choice_das_diss}, we obtain the end-to-end scaling. For DAS:
\begin{equation}
    N_H\cdot G_{\mathrm{DAS}}
    \ =\
    \widetilde{O}\!\left(
    \frac{c_E N_o C_{\mathrm{DK}}^2}{\epsilon_E}\cdot
    \frac{\|H\|^{3}}{\Delta_{\min}^{3}}\cdot\frac{1}{\varepsilon}\cdot \mathrm{poly}(N_o)
    \right).
    \label{eq:das_total_cost}
\end{equation}
For dissipative continuation:
\begin{equation}
    N_H\cdot G_{\mathrm{diss}}
    \ =\
    \widetilde{O}\!\left(
    \frac{c_E N_o C_{\mathrm{DK}}^2}{\epsilon_E}\cdot
    \frac{\|H\|}{\Delta_{\min}}\cdot \mathsf{F}(\varepsilon_{\mathrm{leak}})
    \cdot \Big(N_o+\ln\frac{1}{\varepsilon}\Big)
    \right). 
    \label{eq:diss_total_cost}
\end{equation}

If one further uses the coarse bound $C_{\mathrm{DK}}=O(N_o/\Delta_{\min})$, then
\begin{align}
    N_H\cdot G_{\mathrm{DAS}}
    &=\widetilde{O}\!\left(
    \frac{\|H\|^{3}}{\Delta_{\min}^{5}}\cdot
    \frac{1}{\epsilon_E}\cdot\frac{1}{\varepsilon}\cdot \mathrm{poly}(N_o)\cdot N_o^3
    \right),
    \\
    N_H\cdot G_{\mathrm{diss}}
    &=\widetilde{O}\!\left(
    \frac{\|H\|}{\Delta_{\min}^{3}}\cdot
    \frac{1}{\epsilon_E}\cdot \mathsf{F}(\varepsilon_{\mathrm{leak}})
    \cdot \Big(N_o+\ln\frac{1}{\varepsilon}\Big)\cdot N_o^3
    \right),
\end{align}
which is the total-energy version of the appendix accounting. The intensive-error formulas are recovered by substituting $\epsilon_E=N_o\epsilon_{\rm int}$; this is the convention used when comparing to the older heuristic summaries.

\section{Phase Randomization and Dynamic Cooling versus Dissipative Continuation}
\label{sec:pr_dc_vs_dissip}

Phase randomization (PR)~\cite{Boixo2009Eigenpath} and dynamic cooling (DC)~\cite{motlagh2024_cooling} provide unitary-centric alternatives to using an explicit Lindbladian cooling primitive. Both are useful comparators in the continuation path setting because they exploit the same structural feature as dissipative continuation: a warm start at $s_i$ is already concentrated in the low-energy subspace of $H(s_{i+1})$ when the path is discretized finely enough. The purpose of this section is to justify the gate-complexity scalings used in \Cref{tab:comp_tab}, highlighting how the dependence on $\|H\|$, $\Delta_{\min}$, and $N_o$ arises from (i) Hamiltonian simulation costs, (ii) the time scale required to resolve gaps, and (iii) the number of randomization / cooling iterations required at each path point.

\subsection{Shared cost model and path discretization}

We adopt the standard Hamiltonian-simulation cost model used throughout the paper: implementing time evolution for time $t > 0$ under an interval Hamiltonian $H(s)$ has gate complexity
\begin{equation}
    \label{eq:hsim_cost_model_prdc}
    \mathrm{Cost}\!\left(e^{-iH(s)t}\right)
    =
    \widetilde{O}\!\left(t \|H\|\mathrm{poly}(N_o)\right),
\end{equation}
where $\|H\| \coloneqq \max_{s\in[s_i,s_{i+1}]}\|H(s)\|$ and $\mathrm{poly}(N_o)$ absorbs block-encoding, term-selection, basis-change, and other method-dependent overheads.

To transport the state along the continuation path, we discretize $s\in[0,1]$ into $N_H$ geometries so that the instantaneous ground-state projector does not rotate too quickly between adjacent points. As in \Cref{sec:das_vs_diss}, we use a Davis--Kahan (DK) adapted mesh controlled by the DK length $C_{\mathrm{DK}}$, which yields
\begin{equation}
    \label{eq:NH_choice_prdc_matched}
    N_H
    =
    \Theta\!\left(\frac{c_E N_o C_{\mathrm{DK}}^2}{\epsilon_E}\right)
    =\Theta\!\left(\frac{c_E C_{\mathrm{DK}}^2}{\epsilon_{\rm int}}\right),
    \qquad \epsilon_{\rm int}\coloneqq\epsilon_E/N_o .
\end{equation}
In chemistry regimes where $\|\partial_s H(s)\| = O(N_o)$ and $\Delta(s)\ge \Delta_{\min}$ along the relevant portion of the path, we use the coarse estimate
\begin{equation}
    \label{eq:CDK_coarse_prdc}
    C_{\mathrm{DK}}
    =
    O\!\left(\frac{N_o}{\Delta_{\min}}\right),
\end{equation}
which is the substitution used in the caption of \Cref{tab:comp_tab}. Thus $N_H=\widetilde{O}(c_E N_o^3/(\Delta_{\min}^2\epsilon_E))$ for total energy precision.

\subsection{Phase randomization}

At a fixed geometry $s$, PR applies $H(s)$ for a random time and averages over that randomness, yielding a dephasing channel in the eigenbasis of $H(s)$:
\begin{equation}
    \label{eq:pr_channel_prdc}
    \mathcal{R}_{s}(\rho)
    =
    \mathbb{E}_{t\sim\mu}\!\left[e^{-iH(s)t}\rho e^{+iH(s)t}
    \right].
\end{equation}
In the eigenbasis $\{\ket{E_k(s)}\}$, off-diagonal terms obey
$\rho_{jk}\mapsto \rho_{jk}\mu(E_j-E_k)$, hence choosing a distribution $\mu$ with
$|\mu(\omega)|\ll 1$ for $|\omega|\gtrsim \Delta_{\min}$ suppresses coherences between eigenstates separated by at least the minimum gap.

The PR traversal algorithm~\cite{Boixo2009Eigenpath} stabilizes transport along the discretized path by using this dephasing as an approximate projective step. In~\cite{Boixo2009Eigenpath}, the discretized path is set by a total eigenstate-path length
\begin{equation}
    L = \int_{0}^1 \| \partial_s |E_0(s) \rangle \| \; ds.
\end{equation}
The number of points in the discretized path is then taken to be
\begin{equation}
    q = O \left(\frac{L^2}{1-p} \right),
\end{equation}
where the ground-state fidelity of the final state $|E_0(1)\rangle$ is at least $p > 0$. Observe that $L$ is bounded above by the Davis-Kahan constant,
\begin{equation}
    L\leq \int_{0}^1 \frac{\| \partial_s H(s)\|}{\Delta(s)} \; ds = C_{DK}.
\end{equation}
Furthermore, energy error is bounded by spectral width, and by extension, the operator norm as
\begin{equation}
    \epsilon_E = |\Tr(H \rho)-E_0| \leq (E_{\max} - E_0) (1-p) \leq 2\|H\| (1-p),
    \label{eq:pr-energy-width-bound}
\end{equation}
where $E_0$ is the ground state energy and $E_{\max}$ is the maximum eigenvalue of $H(1)$. Therefore, the number of discretized path points required by PR is
\begin{equation}
    q = O\left(\frac{C_\mathrm{DK}^2 \|H\|}{\epsilon_E} \right),
\end{equation}
resulting in a finer path discretization under this fidelity-to-energy conversion. This is a deliberately conservative comparison: a more refined PR analysis that exploits the same extensive low-energy window used in dissipative continuation could change the explicit norm and system-size factors, so the table should be read as a comparison of commonly stated sufficient bounds rather than as a tight separation theorem.

The total expected Hamiltonian-evolution time is then
\begin{equation}
    T_\mathrm{PR} = \mathbb{E}\left[ \sum_{i=1}^q |t_i| \right],
\end{equation}
where the gap-resolution requirement sets the evolution-time scale. Any distribution $\mu$ that significantly damps frequencies $|\omega|\gtrsim \Delta_{\min}$ must have time spread on the order of
\begin{equation}
    \label{eq:pr_time_scale}
   \mathbb{E}[|t_i|]
    =
    \Theta\!\left(\frac{1}{\Delta_{\min}}\right)
\end{equation}
(up to polylogarithmic factors depending on the chosen $\mu$). Therefore,
\begin{equation}
    T_\mathrm{PR} = O\left(\frac{C_\mathrm{DK}^2 \|H\|}{\Delta_{\min} \epsilon_E} \right) =O\left( \frac{N_o^2 \|H\|}{\Delta_{\min}^3 \epsilon_E} \right).
\end{equation}
where additional factors of $N_o$ and $\Delta_{\min}$ come from the coarse bound $C_{\mathrm{DK}}=O(N_o/\Delta_{\min})$.

Substituting $t=T_{\mathrm{PR}}$ into the Hamiltonian-simulation model in \cref{eq:hsim_cost_model_prdc} gives an overall algorithmic gate complexity of
\begin{equation}
\label{eq:pr_total_cost_with_dk}
    n_G^\mathrm{PR} = \widetilde{O}\left(
        T_\mathrm{PR} \|H\|\cdot \mathrm{poly}(N_o)
    \right)= \widetilde{O}\left(
        \frac{\|H\|^2}{\Delta_{\min}^3}\cdot \frac{1}{\epsilon_E}\cdot \mathrm{poly}(N_o)
    \right).
\end{equation}
which matches the PR row of \Cref{tab:comp_tab}.

\subsection{Dynamic cooling}

Dynamic cooling alternates a dephasing/phase-mixing step with an explicitly energy-selective unitary procedure, implemented using QSP/QSVT transformations interleaved with a weak perturbation that induces transitions biased toward lower energies~\cite{motlagh2024_cooling}. The scaling used in \Cref{tab:comp_tab} can be recovered by combining (i) the same Hamiltonian-simulation cost model of \cref{eq:hsim_cost_model_prdc}, (ii) the same DK-based discretization $N_H=\Theta(c_E C_{\mathrm{DK}}^2/\epsilon_{\rm int})=\Theta(c_E N_o C_{\mathrm{DK}}^2/\epsilon_E)$, and (iii) the DC iteration complexity at each geometry. In the DC analysis, the number of calls to the underlying $e^{-iHt}$ primitive at a fixed geometry is governed by an effective dimension parameter $d$ that controls the schedule or energy-resolution regime of the procedure. For active-space electronic-structure instances along a continuation path, the relevant extensive energy window is proportional to the number of active spatial orbitals, so this parameter is taken to scale as
\begin{equation}
    \label{eq:d_equals_No_prdc}
    d = O(N_o),
\end{equation}
and the number of Hamiltonian-evolution calls per geometry is
\begin{equation}
    \label{eq:dc_num_calls_prdc}
    N_{\mathrm{calls}}^{\mathrm{DC}}
    =
    \widetilde{O}\!\left(d^{3/2}\right)
    =
    \widetilde{O}\!\left(N_o^{3/2}\right),
\end{equation}
which is the convention used when translating the dynamic-cooling iteration count to the chemistry setting considered here. Each call is an evolution for a time on the order of $1/\Delta_{\min}$ to resolve the relevant spectral features (analogous to \cref{eq:pr_time_scale}), so the
per-geometry DC gate complexity is
\begin{equation}
    \label{eq:dc_cost_per_geometry}
    \widetilde{O}\!\left(
        N_{\mathrm{calls}}^{\mathrm{DC}}\cdot
        \frac{\|H\|}{\Delta_{\min}}\cdot
        \mathrm{poly}(N_o)
    \right)
    =
    \widetilde{O}\!\left(
        \frac{\|H\|}{\Delta_{\min}}\cdot
        N_o^{3/2}\cdot
        \mathrm{poly}(N_o)
    \right).
\end{equation}

Transporting along the path multiplies this cost by $N_H$. Using \cref{eq:NH_choice_prdc_matched} and the coarse estimate \cref{eq:CDK_coarse_prdc} gives $N_H=\widetilde{O}(c_EN_o^3/(\Delta_{\min}^2\epsilon_E))$ for total energy precision, hence
\begin{equation}
\label{eq:dc_total_cost_with_dk}
    n_G^{\mathrm{DC}}
    =
    \widetilde{O}\!\left(
        \frac{\|H\|}{\Delta_{\min}}\cdot
        \frac{N_o^3}{\Delta_{\min}^2}\cdot
        \frac{1}{\epsilon_E}\cdot
        N_o^{3/2}\cdot
        \mathrm{poly}(N_o)
    \right)
    =
    \widetilde{O}\!\left(
        \frac{\|H\|}{\Delta_{\min}^{3}}\cdot
        \frac{1}{\epsilon_E}\cdot
        \mathrm{poly}(N_o)
    \right),
\end{equation}
where the $N_o^{3/2}$ factor is absorbed into $\mathrm{poly}(N_o)$ in the table presentation. This recovers the DC row in \Cref{tab:comp_tab} while making explicit the internal accounting: the dominant system-size penalty in DC arises because the primitive $e^{\pm iH}$ is called $\widetilde{O}(N_o^{3/2})$ times per geometry, in addition to the $N_H$ geometries required by the DK discretization.

\section{Ground-state overlap from occupied-space principal angles}
\label{sec:prop1}
Along a nuclear continuation path, one typically solves an independent Hartree-Fock, CASSCF, DMRG-SCF, or related electronic-structure problem at each geometry. This procedure optimizes the electronic energy locally at each point $\mathbf{R}(s_i)$, but it does not enforce continuity of the orbital representation between successive geometries. In particular, Hartree-Fock determines only the occupied subspace at a given geometry, not a unique choice of orbitals spanning that subspace: any unitary rotation among occupied orbitals leaves the Slater determinant (and the Hartree-Fock energy) unchanged, so the numerical orbitals are defined only up to an arbitrary gauge within the occupied space. In order to obtain the maximum overlap between mean-field states between successive geometries $\mathbf{R}(s_k)$ and $\mathbf{R}(s_{k+1})$, we therefore solve for the unitary rotation that maximizes orbital overlap. This motivates our Procrustes alignment procedure, which fixes a consistent gauge along the continuation path. After alignment, the resulting family of Hamiltonians and ground states exhibits the smooth drift behavior required by \cref{thm:reaction-path-runtime_informal}.

We emphasize that Procrustes alignment is not an active-space selection procedure. It does not assert that a fixed reactant-side active space remains chemically valid along the full continuation path. Rather, active-space selection is part of the surrounding classical electronic-structure workflow, and may use automated or corresponding-active-space methods to define compatible active orbital spaces along the path~\cite{SteinReiher2016AutoCAS,BensbergReiher2023CorrespondingActiveSpaces}. Procrustes alignment then fixes only the residual orbital gauge within, or after embedding into, a common active-space representation. The theorem applies to this selected embedded representation. If no such locally compatible representation exists for the chemical channel of interest, the theorem does not claim efficient preparation.

The Procrustes alignment algorithm utilizes \Cref{prop:gs-overlap-principal-angles} to optimize orbital gauges by computing the singular value decomposition (SVD) of $M_{\rm occ}^{(k)} = U_{\rm occ}^{(k)} \Sigma_{\rm occ}^{(k)} V_{\rm occ}^{(k)\dagger}$. From the singular values $\{\sigma_i^{(k)}\}$, we can directly measure the overlaps in~\cref{eq:HF-overlap-product-compact}. The alignment can be initialized from any reference geometry whose active-space frame has been chemically validated. In the TS applications targeted here, it is natural to initialize from the TS-centered active-space frame and propagate the orbital gauge outward along the selected continuation path toward a tractable warm-start geometry. The quantum algorithm then transports the prepared state forward along the same aligned representation. At each neighboring pair of geometries, we solve the Procrustes problem
\begin{equation}
    R^{(k)} \coloneqq \arg\max_R
    \mathrm{Re}\{ \mathrm{Tr}\big(R^\dagger M_{\rm occ}^{(k)}\big)\} = U_{\rm occ}^{(k+1)} (V_{\rm occ}^{(k+1)})^\dagger.
    \label{eq:procrustes}
\end{equation}
Then, the occupied molecular orbital coefficient matrix for geometry $\mathbf{R}(s_k)$ is rotated by $R^{(k)}$, $C^\text{occ}_{k} \gets C^\text{occ}_{k} R^{(k)}$. The matrix $M_{\rm occ}^{(k)}$ is given by
\begin{equation}
    M_{\rm occ}^{(k)}
    =
    \big(C_{k+1}^{\rm occ}\big)^\dagger S_{k,k+1} C_k^{\rm occ}.
\end{equation}
where the \textit{AO cross-overlap matrix} $S_{k,k+1}\in\mathbb{C}^{N_k\times N_{k+1}}$ is defined entry-wise by
\begin{equation}
    \bigl(S_{k,k+1}\bigr)_{\mu\nu}
    \coloneqq 
    \langle \chi^{(k)}_\mu | \chi^{(k+1)}_\nu \rangle
    =
    \int_{\mathbb{R}^3} d\mathbf{r}
    \chi^{(k)}_\mu(\mathbf{r})^{*}\chi^{(k+1)}_\nu(\mathbf{r}),
\end{equation}
where $\{\chi^{(k)}_\mu(\mathbf{r})\}_{\mu=1}^{N_k}$ and $\{\chi^{(k+1)}_\nu(\mathbf{r})\}_{\nu=1}^{N_{k+1}}$ denote the (possibly geometry-dependent) atomic-orbital (AO) basis functions associated with $\mathbf{R}(s_k)$ and $\mathbf{R}(s_{k+1})$, respectively.

\begin{proposition}[Ground-state overlap from occupied-space principal angles]
    \label{prop:gs-overlap-principal-angles}
    Consider $\mathbf{R}(s_k)$ and $\mathbf{R}(s_{k+1})$, two consecutive geometries along the continuation path. Define the atomic orbital (AO) cross-overlap matrix $S_{k,k+1}$, and let $\{\chi^{(k)}_\mu(\mathbf{r})\}_{\mu=1}^{N_k}$ and $\{\chi^{(k+1)}_\nu(\mathbf{r})\}_{\nu=1}^{N_{k+1}}$ denote the (possibly geometry-dependent) atomic-orbital (AO) basis functions. Let the \textit{AO cross-overlap matrix} be denoted $S_{k,k+1}\in\mathbb{C}^{N_k\times N_{k+1}}$. Let $C_k^{\rm occ},C_{k+1}^{\rm occ}\in\mathbb C^{n_{\rm AO}\times N_{\rm occ}}$ collect orthonormal occupied molecular orbitals at $\mathbf{R}(s_k)$ and $\mathbf{R}(s_{k+1})$, and let
    \begin{equation}
        M_{\rm occ}^{(k)}
        =
        \big(C_{k+1}^{\rm occ}\big)^\dagger S_{k,k+1} C_k^{\rm occ}.
    \end{equation}
    Denote the singular values of $M_{\rm occ}^{(k)}$ by $\sigma_i^{(k)} = \cos\theta_i^{(k)}$ with principal angles $\theta_i^{(k)}\in[0,\pi/2]$, $i=1,\dots,N_{\rm occ}$. Let $|\Phi_0^{(k)}\rangle$ and $|\Phi_0^{(k+1)}\rangle$ be the corresponding closed-shell single-determinant Hartree-Fock states. Then the squared overlap between these two single-determinant states is
    \begin{equation}
        \bigl|\langle \Phi_0^{(k+1)} | \Phi_0^{(k)}\rangle\bigr|^2
        =
        \prod_{i=1}^{N_{\rm occ}} \cos^2\!\theta_i^{(k)}.
        \label{eq:HF-overlap-product-compact}
    \end{equation}
\end{proposition}
\begin{proof}
    For brevity, we suppress spin and treat $\{ \phi_p^{(k)} \}_{p=1}^{N_{\rm occ}}$ and $\{ \phi_p^{(k+1)} \}_{p=1}^{N_{\rm occ}}$ as orthonormal one-particle orbitals 
    spanning the occupied spaces at geometries $\mathbf{R}(s_k)$ and $\mathbf{R}(s_{k+1})$, respectively.  The generalization to closed-shell spin-orbital determinants is straightforward and only changes notation.

    Let $|\Phi_0^{(k)}\rangle$ and $|\Phi_0^{(k+1)}\rangle$ be the $N_{\rm occ}$-electron Slater determinants built from these orbitals. In first-quantized notation they are
    \[
        |\Phi_0^{(k)}\rangle
        =
        \frac{1}{\sqrt{N_{\rm occ}!}}
        \phi_1^{(k)} \wedge \cdots \wedge \phi_{N_{\rm occ}}^{(k)},
    \]
    \[
        |\Phi_0^{(k+1)}\rangle
        =
        \frac{1}{\sqrt{N_{\rm occ}!}}
          \phi_1^{(k+1)} \wedge \cdots \wedge \phi_{N_{\rm occ}}^{(k+1)},
    \]
    where $\wedge$ denotes the antisymmetrized wedge product. A standard identity for wedge products (equivalently, the Slater–Condon rule) states that for two sets of one-particle states $\{\psi_i\}_{i=1}^N$ and  $\{\varphi_j\}_{j=1}^N$,
    \begin{equation}
        \big\langle \psi_1 \wedge \cdots \wedge \psi_N
        \big|
        \varphi_1 \wedge \cdots \wedge \varphi_N
        \big\rangle
        =
        \det G, 
        \qquad
        G_{ij} \coloneqq \langle \psi_i | \varphi_j \rangle.
        \label{eq:wedge-inner-product}
    \end{equation}
    Applying~\cref{eq:wedge-inner-product} with
    $\psi_i = \phi_i^{(k+1)}$ and $\varphi_j = \phi_j^{(k)}$ gives
    \begin{equation}
        \langle \Phi_0^{(k+1)} | \Phi_0^{(k)} \rangle
        =
        \det M_{\rm occ}^{(k)},
        \qquad
        M_{{\rm occ},ij}^{(k)}
        \coloneqq
        \langle \phi_i^{(k+1)} | \phi_j^{(k)} \rangle.
        \label{eq:det-overlap}
    \end{equation}
    In the AO basis, the occupied orbitals are $\phi^{(k)}_j = \sum_\mu (C_k^{\rm occ})_{\mu j}\chi^{(k)}_\mu$, $\phi^{(k+1)}_i = \sum_\nu (C_{k+1}^{\rm occ})_{\nu i}\chi^{(k+1)}_\nu$, and the AO cross overlap is  $S_{k,k+1,\nu\mu} = \langle \chi^{(k+1)}_\nu | \chi^{(k)}_\mu \rangle$. Thus
    \[
        M_{\rm occ}^{(k)}
        =
        \big(C_{k+1}^{\rm occ}\big)^\dagger
        S_{k,k+1}
        C_k^{\rm occ},
    \]
    as stated in the proposition.
    
    Let the singular value decomposition of $M_{\rm occ}^{(k)}$ be
    \begin{equation}
        M_{\rm occ}^{(k)}
        =
        U_{\rm occ}^{(k)}
        \Sigma_{\rm occ}^{(k)}
        V_{\rm occ}^{(k)\dagger},
        \qquad
        \Sigma_{\rm occ}^{(k)}
        = \operatorname{diag}(\sigma_1^{(k)},\dots,\sigma_{N_{\rm occ}}^{(k)}),
    \end{equation}
    with singular values $\sigma_i^{(k)} \in [0,1]$. By definition of principal angles between the occupied subspaces
    $\mathcal{W}_k = \operatorname{span}\{\phi_j^{(k)}\}$ and $\mathcal{W}_{k+1} = \operatorname{span}\{\phi_i^{(k+1)}\}$, we can write
    \[
      \sigma_i^{(k)} = \cos \theta_i^{(k)},
      \qquad
      \theta_i^{(k)} \in [0,\pi/2],\quad i=1,\dots,N_{\rm occ}.
    \]
    From~\cref{eq:det-overlap} we have
    \[
        \bigl|\langle \Phi_0^{(k+1)} | \Phi_0^{(k)} \rangle\bigr|^2
        =
        \bigl|\det M_{\rm occ}^{(k)}\bigr|^2
        =
        \det\!\bigl( M_{\rm occ}^{(k)\dagger} M_{\rm occ}^{(k)} \bigr).
    \]
    But $M_{\rm occ}^{(k)\dagger} M_{\rm occ}^{(k)}$ is unitarily similar to $\Sigma_{\rm occ}^{(k)2}$ and thus has eigenvalues $\{(\sigma_i^{(k)})^2\}_{i=1}^{N_{\rm occ}}$.  Therefore
    \[
        \bigl|\langle \Phi_0^{(k+1)} | \Phi_0^{(k)} \rangle\bigr|^2
        =
        \prod_{i=1}^{N_{\rm occ}} (\sigma_i^{(k)})^2
        =
        \prod_{i=1}^{N_{\rm occ}} \cos^2 \theta_i^{(k)},
    \]
    which is exactly~\cref{eq:HF-overlap-product-compact}.
\end{proof}

\section{Theorem 2 (Time support scaling for Gevrey vs.\ periodic square-wave filters)}
\label{sec:thm2}
\begin{theorem}[Time support scaling for Gevrey vs.\ periodic square-wave filters]
    \label{thm:filter-scaling_formal}
    Let $H=\sum_j E_j\ket{E_j}\!\bra{E_j}$ have a nondegenerate ground state $\ket{E_0}$ and gap $\Delta\coloneqq\min_{j>0}(E_j-E_0)>0$. Let $A$ be bounded with $\|A\|\le 1$. Fix a target leakage $0<\varepsilon_{\mathrm{leak}}<1$.
    
    Then there exist two families of filtered jump operators $K^{\mathrm{G}}_{\varepsilon_{\mathrm{leak}}}$ and $K^{\mathrm{SW}}_{\varepsilon_{\mathrm{leak}}}$
    such that
    \[
        \bigl\|K_{\varepsilon_{\mathrm{leak}}}\ket{E_0}\bigr\|\ \le\ \varepsilon_{\mathrm{leak}}.
    \]
    Moreover, writing $\omega_{\max}\coloneqq\max_j(E_j-E_0)\le 2\|H\|$, the required time support (and hence, up to constant discretization factors, the number of distinct real-time evolution calls) can be chosen as:
    \begin{itemize}
        \item \textbf{Gevrey filter.}
        There exists a Gevrey-class filter with essentially compact time support satisfying
        \[
            S_{\mathrm{G}}(\varepsilon_{\mathrm{leak}})
            =O\!\left(\frac{\|H\|}{\Delta}[\ln(1/\varepsilon_{\mathrm{leak}})]^{\alpha}\right),
            \qquad \alpha > 1,
        \]
        with $\alpha$ depending only on the Gevrey class~\cite{ding2024single}.
        
        \item \textbf{Periodic square-wave Fourier filter.}
        There exists a truncated periodic square-wave (Dirichlet) filter with time support
        \[
            S_{\mathrm{SW}}(\varepsilon_{\mathrm{leak}})
            =O\!\left(\frac{\|H\|}{\Delta}\frac{1}{\varepsilon_{\mathrm{leak}}}\right),
        \]
        up to bounded oscillatory factors in the leakage.
    \end{itemize}
    
    Finally, if the filter is implemented by sampling times on a uniform grid with step size $\Delta t=\Theta(1/\|H\|)$, then the number of distinct calls of the form
    $e^{iHt}(\cdot)e^{-iHt}$ is $\Theta(S(\varepsilon_{\mathrm{leak}}))$.
\end{theorem}

\begin{proof}
    Define, for any integrable (or truncated / discretized) kernel $f$,
    \begin{equation}
        K_f\ \coloneqq\ \int_{\mathbb{R}} f(t)e^{itH}Ae^{-itH}dt.
        \label{eq:Kf-def-tight}
    \end{equation}
    In the eigenbasis of $H$,
    \[
        e^{itH}A e^{-itH}
        =\sum_{j,k} e^{it(E_j-E_k)}A_{jk}\ket{E_j}\!\bra{E_k},
        \qquad A_{jk}\coloneqq\bra{E_j}A\ket{E_k},
    \]
    so
    \begin{equation}
        K_f
        =\sum_{j,k}\widehat f(E_j-E_k)A_{jk}\ket{E_j}\!\bra{E_k},
        \label{eq:Kf-matrix-tight}
    \end{equation}
    where $\widehat f(\omega)=\int_{\mathbb{R}} f(t)e^{it\omega}dt$ (same Fourier convention as in the square-wave derivation). Applying to $\ket{E_0}$ gives
    \[
        K_f\ket{E_0}
        =\sum_{j>0}\widehat f(E_j-E_0)A_{j0}\ket{E_j}.
    \]
    Hence
    \begin{equation}
        \|K_f\ket{E_0}\|
        \le \Big(\sup_{\omega\ge \Delta}|\widehat f(\omega)|\Big)\|A\ket{E_0}\|
        \le \sup_{\omega\ge \Delta}|\widehat f(\omega)|.
        \label{eq:leak-reduction-tight}
    \end{equation}
    Thus it suffices to construct $\widehat f$ with $\sup_{\omega\ge\Delta}|\widehat f(\omega)|
    \le \varepsilon_{\mathrm{leak}}$.
    
    Let $\mathrm{Si}(x)=\int_0^x \frac{\sin u}{u}du$. For $x>0$,
    \[
        \frac{\pi}{2}-\mathrm{Si}(x)=\int_x^\infty \frac{\sin u}{u}du
        =\frac{\cos x}{x}-\int_x^\infty \frac{\cos u}{u^2}du,
    \]
    so
    \begin{equation}
        \Big|\frac{\pi}{2}-\mathrm{Si}(x)\Big|
        \le \frac{1}{x}+\int_x^\infty \frac{1}{u^2}du
        =\frac{2}{x}.
        \label{eq:Si-tail-bound}
    \end{equation}
    
    For the one-sided step $\widehat f(\omega)=\mathbf 1_{\{\omega<0\}}$, with convention $\widehat f(\omega)=\int f(t)e^{i\omega t}\,dt$, its inverse transform is the Hilbert kernel $f(t)=\frac{1}{2}\delta(t)+\frac{i}{2\pi}\mathrm{PV}\left(1/t\right)$, where $\mathrm{PV}$ denotes the Cauchy principal value. Rectangular truncation $f_T(t)=\left(\frac{1}{2}\delta(t)+\frac{i}{2\pi}\mathrm{PV}\left(1/t\right)\right)\mathbf{1}_{\{|t|\le T\}}$ yields the closed form
    \begin{equation}
        \widehat f_T(\omega)=\frac12-\frac{1}{\pi}\mathrm{Si}(\omega T).
        \label{eq:fhatT-Si-tight}
    \end{equation}
    For $\omega>0$ (the upward Bohr-frequency side relevant to leakage from the ground state), combine~\cref{eq:Si-tail-bound} with~\cref{eq:fhatT-Si-tight} to get the uniform bound 
    \begin{equation}
        |\widehat f_T(\omega)|
        =\frac{1}{\pi}\Big|\frac{\pi}{2}-\mathrm{Si}(\omega T)\Big|
        \le \frac{2}{\pi\omega T}.
        \label{eq:rect-uniform}
    \end{equation}
    Therefore, for a gap margin $\Delta>0$,
    \begin{equation}
        \sup_{\omega\ge\Delta}|\widehat f_T(\omega)|
        \le \frac{2}{\pi\Delta T}.
        \label{eq:rect-gap-bound}
    \end{equation}
    (The asymptotic refinement
    $\widehat f_T(\omega)=\frac{\cos(\omega T)}{\pi\omega T}+O((\omega T)^{-2})$ explains the bounded oscillatory factor;~\cref{eq:rect-gap-bound} is the clean envelope.) Choosing $T=\Theta\!\big(1/\Delta\varepsilon_{\mathrm{leak}}\big)$ makes the RHS $\le \varepsilon_{\mathrm{leak}}$, and then~\cref{eq:leak-reduction-tight} gives $\|K_f\ket{E_0}\|\le \varepsilon_{\mathrm{leak}}$.
    
    Now implement the filter by periodizing in frequency with period $\Omega$ and truncating to $N_\omega$ Fourier modes (Dirichlet partial sum). Choose
    \begin{equation}
        \Omega\ \ge\ 4\|H\|\ \ \ \Rightarrow\ \ \ \omega_{\max}\le 2\|H\|<\Omega/2,
        \label{eq:alias-guard}
    \end{equation}
    so that all Bohr frequencies $E_j-E_k$ lie strictly inside the principal interval
    $(-\Omega/2,\Omega/2)$ (no wrap-around/aliasing in the relevant band).
    Let $\delta t\coloneqq2\pi/\Omega$ and
    \[
        T_{N_\omega}\coloneqq\Big(N_\omega+\tfrac12\Big)\delta t.
    \]
    As in the derivation, the Dirichlet kernel rescales to the sinc kernel away from $\omega=0$, so the truncated periodic square-wave filter converges (for fixed $\omega$ and, on any band $|\omega|\ge \Delta$, uniformly up to constants) to the rectangularly truncated filter
    with horizon $T_{N_\omega}$. In particular, the same envelope bound applies:
    \begin{equation}
        \sup_{\omega\ge\Delta}|\widehat f^{(\Omega)}_{N_\omega}(\omega)|
        \ \le\ \frac{1}{\Delta T_{N_\omega}}.
        \label{eq:dirichlet-envelope}
    \end{equation}
    Thus it suffices to choose $T_{N_\omega}=\Theta\!\big(\frac{1}{\Delta\varepsilon_{\mathrm{leak}}}\big)$. Finally, the implementation uses time samples $t_n=n\Delta t$ for $|n|\le N_\omega$, so the
    number of distinct real-time evolution calls is $2N_\omega+1=\Theta(N_\omega)$, and the number of modes is bounded above and below by
    \[
        N_\omega\ = \Theta\left(\frac{T_{N_\omega}}{\Delta t}\right)
        =\Theta(\Omega T_{N_\omega})
        =\Theta\!\left(\frac{\|H\|}{\Delta}\frac{1}{\varepsilon_{\mathrm{leak}}}\right),
    \]
    where we used $\Omega=\Theta(\|H\|)$ as in~\cref{eq:alias-guard}. This proves the stated $S_{\mathrm{SW}}$ scaling.
    
    The Gevrey part is an existence statement:~\cite{ding2024single} constructs Gevrey-class filters that approximate the one-sided step outside a spectral margin while keeping the time-domain kernel essentially supported on a window of length $O\!\big(\frac{\|H\|}{\Delta}[\ln(1/\varepsilon_{\mathrm{leak}})]^\alpha\big)$. Applying the same reduction~\cref{eq:leak-reduction-tight} yields the leakage guarantee.
\end{proof}

\section{Reducing the Cost of Simulating Dissipative Dynamics via qDRIFT-inspired Jump Sampling}
\label{sec:qdrift}

It is costly to implement the total Lindblad generator $\mathcal{L} = \sum_a \mathcal{L}_a$ at once as this requires a cost at least as much as implementing the first order product formula $\prod_a \exp(\tau  \mathcal{L}_a)$. Therefore, we simulate dissipative continuation generated by a sum of Lindblad terms by random sampling one term per short time step. Concretely, we write
\[
    \overline{\mathcal L} = \sum_{a\in\mathcal A} \mu_a\mathcal L_a,
    \qquad \sum_{a\in\mathcal A} \mu_a = 1,
\]
and define the single-step randomized channel
\begin{equation}
    \mathcal E_\tau
    \coloneqq
    \mathbb E_{a\sim \mu}\!\left[e^{\tau \mathcal L_a}\right]
    =
    \sum_{a\in\mathcal A} \mu_a e^{\tau \mathcal L_a}.
\end{equation}

This randomized splitting is the open-system analouge of the qDRIFT protocol for Hamiltonian simulation~\cite{Campbell2019}. Note that by a Taylor expansion,
\[
    \mathcal E_\tau
    =
    \mathbb I + \tau \overline{\mathcal L} + O(\tau^2),
\]
so $\mathcal E_\tau$ matches $e^{\tau\overline{\mathcal L}}$ to first order.  For total evolution time $t$ we take $N_T$ steps with $\tau=t/N_T$ and approximate
\begin{equation}
    e^{t\overline{\mathcal L}}
    \approx
    \bigl(\mathcal E_{t/N_T}\bigr)^{N_T}.
    \label{eq:avg_channel}
\end{equation}

Following~\cite{Campbell2019}, the resulting error can be characterized in the diamond norm and scales as
\[
    \bigl\|(\mathcal E_{t/N_T})^{N_T} - e^{t\overline{\mathcal L}}\bigr\|_\diamond
    =
    O\!\left(\frac{\lambda^2 t^2}{N_T}\right),
    \qquad
    \lambda \coloneqq \sum_{a\in\mathcal A}\|\mathcal L_a\|_\diamond,
\]
up to constants that depend on the chosen generator norm and the splitting. We have the freedom to control the value of $\lambda$ by rescaling the simulation time, $\exp(t \mathcal{L}_a)=\exp((t/c) \times c\mathcal{L}_a)$ for some $c > 0$, then $\lambda =\sum_a c\|\mathcal{L}_a\|_\diamond$.

From an implementation standpoint, each randomized step applies only one elementary dissipative update (one chosen jump channel) rather than deterministically applying all $|\mathcal A|$ terms per Trotter step.  In our construction, this corresponds to repeatedly implementing primitives of the form
$e^{-i \frac{\sqrt{\tau}}{2}\sigma_k\otimes A}$.
Therefore, sampling a single jump operator per step can yield a lower per-step cost than deterministic first- or second-order Trotterization, while exhibiting the same asymptotic $O(\lambda^2 t^2/N_T)$ error scaling in the purely dissipative case. Note that, $\lambda$ scales with the aggregate strength of the jump channels included in the decomposition. Consequently, when our implementation uses only a small number of jump operators, the increase in $\lambda$ is modest and does not change the overall asymptotic scaling of the method.

\section{Logical Quantum Circuit Construction for Dissipative Ground State Preparation}
\label{sec:logical_circuit}

The approach of \Cref{sec:qdrift} reduces the cost of implementing the entire dissipative continuation over all jump operators simultaneously to the cost of a single application of an individual randomly sampled jump operator $\mathrm{Cost}(e^{\tau \mathcal{L}_a})$.
Each generator has Lindblad form
\[
    \mathcal L_a(\rho)  = K_a \rho K_a^\dagger - \tfrac12 \{K_a^\dagger K_a,\rho\},
\]
where $\{K_a\}$ are the filtered jump operators defined in \cref{def:filtered_jumps,eq:filtered_jump_prelim} and we implement a circuit that realizes a CPTP map $  e^{\tau\mathcal L_a} $.

For a total time $t=N_T\tau$, one realization of the randomized evolution is the \textit{expected (average) channel} is the mixture over all realizations $\mathcal{E}_{\tau}$.
Define the averaged Lindbladian
\begin{equation}
    \overline{\mathcal{L}} \coloneqq \mathbb{E}_{a\sim\mu}\mathcal{L}_a
    =  \mathbb{E}_{a}\left[K_a (\cdot) K_a^\dagger - \tfrac12\{K_a^\dagger K_a,\cdot\}\right].
\end{equation}
Then the ideal expected evolution at time $t$ is $e^{t\overline{\mathcal{L}}}$ in~\cref{eq:avg_channel}.

In practice, we can reduce the required quantum resources to implement the dissipative continuation over the entire jump set $\mathcal{A}$ by randomly sampling jump operators for each time step $e^{\tau \mathcal{L}_{a}}$ where $A_a$ are drawn i.i.d.  Ding et al.~\cite{ding2024single} provide an efficient circuit for implementing the dissipative continuation of a single filtered jump operator $K_a$ that only requires a single ancilla in addition to the system register (excluding ancilla for block encodings). The operator $W(\sqrt{\tau})$ acting on $n$ system qubits and one ancilla qubit that approximates a time step of dissipative continuation for the target Lindbladian $\mathcal{L}$ to second order,
\begin{equation}
    \mathcal{E}_\tau(\rho)
    \coloneqq
    \operatorname{tr}_{\rm anc}\!\left[
    W(\sqrt{\tau})
    \left(|0\rangle \langle 0| \otimes \rho\right)
    W(\sqrt{\tau})^\dagger
    \right]
    = e^{\tau \mathcal{L}}[\rho] + O(\tau^2).
\end{equation}
The quantum resources required to implement the quantum channel generated by $W(\sqrt{\tau})$ are determined by the choice of filter function $f(t)$ that removes heating transitions from the unfiltered jump operators $A_a$. \cref{thm:filter-scaling} establishes a link between two different classes of filter functions and the cost of implementing $W(\sqrt{\tau})$.

Our circuit construction is inspired by the Monte Carlo–style quantum algorithm for ground-state preparation in the early fault-tolerant regime~\cite{ding2024single}, which features a Lindbladian whose target state is stationary. In the paper, the authors design a single-ancilla circuit that simulates a Lindblad dynamics step by splitting the evolution into a purely Hamiltonian part and a purely dissipative part with a first-order Trotter method with just one extra qubit. Each block is a short interaction between the ancilla and a simple system operator, with the ancilla axis encoding the complex filter coefficient, so the whole channel is built from a sequence of elementary ancilla rotations and controlled interactions. We adapt and extend this in our own circuit design along the continuation path.

A Lindbladian naturally decomposes as $\mathcal{L} = \mathcal{L}_H + \mathcal{L}_a$, where $\mathcal{L}_H$ denotes the coherent part and $\mathcal{L}_a$ the dissipative part. Since the coherent factor does not contribute to energy decrement, we neglect this direct evolution in the cooling analysis and focus on the evolution of the dissipative part for our cooling analysis.  However, it should be noted that the dissipative cooling primitive is still strongly dependent on the Hamiltonian: the jump operators are filtered through the time evolution of the Hamiltonian (\cref{eq:filtered_jump_prelim}). Discretizing the evolution into time steps of duration $\tau$, the dissipative continuation generated by $\mathcal{L}_a$ can be written as
\[
    e^{t\mathcal{L}_a} \approx \left(e^{\tau \mathcal{L}_a}\right)^{N_T}.
\]
Therefore, we seek an efficient quantum circuit implementation of the dissipative update $e^{\tau \mathcal{L}_a}$ at each time step. We can follow the recent work~\cite{Ding2024simulating} in which the authors show that Stinespring dilation can be used to implement dissipative continuation via Hamiltonian simulation and measurement.

Adopting a periodic square-wave filter offers a concrete implementation advantage for simulating $e^{-iK_f\sqrt{\tau}}$ in our scheme. Periodizing the one-sided step in frequency ensures that the corresponding time-domain kernel is a discrete comb rather than a continuous function, so the filtered jump operator $K_f$ naturally decomposes as 
\begin{equation}
    \label{filteredjump}
    \int_{-\infty}^{\infty} f(s)A(s)ds = K^{\text{SW}} \coloneqq \sum_{k\in\mathbb Z} c_k  A  \big(-kT\big)
\end{equation}
where $A  \big(-kT\big) =  e^{-iH(kT) } A  e^{iH(kT) }$ with equally spaced time shift, and $c_k$ are the Fourier coefficients.

Without any discretization error, the integral $K_f$ can be exactly expressed as $K^{\text{SW}}$ and its truncated form at $|k|\le N_{\omega}$ corresponds to square-wave filter defined as $K^{\text{SW}}_{N_{\omega}}$ which approximates the dilated Hermitian jump operator of $K_f$ using one ancilla qubit as follows

\begin{equation}
    \tilde{K}^{\mathrm{SW}}_{N_\omega} =\left(\begin{array}{cc}
    0 & K_{N_\omega}^{{\mathrm{SW}}^{\dagger}} \\
    K_{N_\omega}^{\mathrm{SW}} & 0
    \end{array}\right)=: \sum_{|k|\le N_\omega} \widetilde{H}_k .
    \label{eq:dilated}
\end{equation}
with target leakage $\varepsilon_{\mathrm{leak}}$, i.e., $||K^{\text{SW}}_{N_{\omega}}|E_0\rangle||\le \varepsilon_{\mathrm{leak}}$. The details of this error bound can be found in the proof of \Cref{thm:filter-scaling} in~\Cref{sec:thm2}.

\Cref{eq:dilated} can further be factorized as
\begin{equation}
    H_k=\left(
    \begin{array}{cc}
        0 & c_k^* A(-kT) \\
        c_k A(-kT) & 0
    \end{array}
    \right)=\sigma_k \otimes A(-kT) .
\end{equation}
For the square wave, the Fourier coefficients $c_k$ have phases restricted to $\{0,\pm \pi/2\}$, implying that the ancilla generators $\sigma_k$ lie on the $\{X,Y\}$ axes of the Bloch sphere. This yields an even simpler and more uniform circuit architecture, each step of the evolution consists of a single ancilla–system \textit{`jump'} at a randomly chosen discrete time.

Let $\widetilde{A}_k(\sqrt{\tau})$ be the $k$-th kick. We implement
\begin{align}
    \widetilde{A}_k(\sqrt{\tau})&\coloneqq e^{-i \frac{\sqrt{\tau}}{2} \sigma_k \otimes A}  \\\nonumber
    &  = 
    \bigl(V_{a i}\otimes I_j\bigr)
    i\mathrm{SWAP}_{ij}(\theta_k)
    \bigl(V_{a i}^\dagger\otimes I_j\bigr).
\end{align}
where $V_{a i} = R_z^{(a)}(\phi)H_a\mathrm{CZ}(a,i)$, where $H_a$ is a Hadamard gate on the ancilla qubit (\cref{fig:A_circuit}). Finally, we realize an effective Hamiltonian evolution for two consecutive right-ordered product yields $A\left(-kT\right) A\left(-(k+1)T\right)= e^{i H kT} A e^{i HT} A e^{-iH(k+1)T}$.

We arrive at a final architecture of a single $W$ operator as implemented by a quantum circuit with a single jump $a$ sampled from the set of the jump operators $\mathcal{A}$ as follows
\begin{align}
    \label{eq:2ndorderformula}
    &\Bigg[
    \prod_{|k|\le N_\omega}^{\rightarrow}
    \bigl(I \otimes e^{i H kT}\bigr)
    \widetilde{A}_k(\sqrt{\tau})
    \bigl(I \otimes e^{-i H kT}\bigr)
    \Bigg]\\\nonumber
    &\qquad\times
    \Bigg[
    \prod_{|k|\le N_\omega}^{\leftarrow}
    \bigl(I \otimes e^{i H kT}\bigr)
    \widetilde{A}_k(\sqrt{\tau})
    \bigl(I \otimes e^{-i H kT}\bigr)
    \Bigg]  \\\nonumber
    &\qquad =
    (I \otimes e^{-i H S_T})
    W(\sqrt{\tau})
    (I \otimes e^{i H S_T}),
\end{align}
where $S_T \coloneqq N_\omega T$ and
\begin{align}
    W(\sqrt{\tau})
    \coloneqq
    &\Bigg[
       \prod_{|k|\le N_\omega}^{\rightarrow}
       \widetilde{A}_k(\sqrt{\tau})
       \bigl(I \otimes e^{i H T}\bigr)
    \Bigg]\\\nonumber
    &\quad\times
    \Bigg[
       \prod_{|k|\le N_\omega}^{\leftarrow}
       \bigl(I \otimes e^{-i H T}\bigr)
       \widetilde{A}_k(\sqrt{\tau})
    \Bigg].
    \label{eq:W_ckt}
\end{align}
The quantum channel that we acquire after the exact cancellation of the term $   \bigl(I \otimes e^{i H kT}\bigr)$ with the next $\bigl(I \otimes e^{-i H kT}\bigr)$ in \cref{eq:2ndorderformula}

\begin{align}
    \mathcal{E}_\tau(\rho)
    \coloneqq
    \operatorname{tr}_{\rm anc}\!\left[
        W(\sqrt{\tau})(|0\rangle\langle0|\otimes\rho)W(\sqrt{\tau})^\dagger
    \right]
    \approx e^{\tau \mathcal{L}_a}[\rho] + O(\tau^2).
\end{align}
Hence, the total quantum channel averaged over the sampled jump operator $a\in\mathcal{A}$ is
\begin{equation}
    \overline{\mathcal{E}}_\tau \approx \mathbb{E}_a(\mathcal{E}^a_\tau).
\end{equation}
Along the continuation path, if the initial state at the current geometry $s_i$ is $\rho_{s_i}$, then applying the simulated Lindbladian dynamics for the next geometry gives
\begin{equation}
    \rho_{s_{i+1}} \approx \overline{\mathcal{E}}_\tau(\rho_{s_i}),
\end{equation}
which is then used as the initial state for the following geometry.

\begin{figure}
    \centering
    \includegraphics[width=0.9\linewidth]{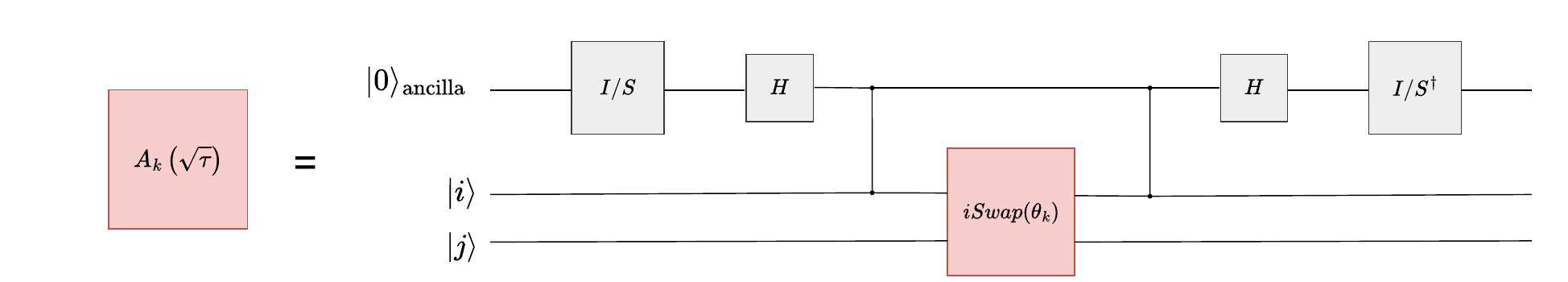}
    \caption{Circuit for realizing $A_k(\sqrt{\tau}) = \exp(- i \sqrt{\tau} \sigma_k \otimes A_a/ 2) = \exp(- i \sqrt{\tau} \mathrm{Im}(c_k)\sigma_y \otimes A_a/ 2)$ for $A_a = (X_iX_j+Y_i Y_j)/2$ using the square wave filter function which only has purely imaginary Fourier coefficients for $\ell > 0$, when $k = 0$, $A_0(\sqrt{\tau})=\exp(- i \sqrt{\tau} \sigma_x \otimes A_a/ 4)$, where $i, j$ are corresponding fermionic modes. The circuit costs enough $T$ gates to synthesize $1$ $R_z$ gates, 2 acting on the ancilla and the gate decomposition of $ i\mathrm{SWAP}(\theta)$ requires only a single $R_z$ leading to a tiny $T$-gate overhead for implementing local jump operators.}
    \label{fig:A_circuit}
\end{figure}

To implement the jump operator which governs the dissipation part of the Lindbladian Master equation, we follow the structure of~\cite{ding2024single} by introducing single ancilla qubit for the environment as below 
\begin{equation}
    \exp \left(-i \frac{\sqrt{\tau}}{2} \sigma_k \otimes A\left(-kT\right)\right)=\left(I \otimes e^{-i H kT}\right) \underbrace{e^{-i \frac{\sqrt{\tau}}{2} \sigma_k \otimes A} }_{=: \widetilde{A}_k(\sqrt{\tau})}\left(I \otimes e^{i H kT}\right) .
    \label{eqH2}
\end{equation}
where,  $\sigma_0=|c_0|\sigma_x$ when $k=0$ and $|c_k|\sigma_y$ otherwise, encoding how strongly the jump operator $A$ is being applied on the system. The operator can also be written as
\begin{equation}
    \sigma_k = |c_k|\left(\cos\phi_k  \sigma_x + \sin\phi_k  \sigma_y \right) 
\end{equation}
where, $\phi_0 = 0$ and $\phi_k = \pi/2$ for $k\neq 0$, with $\phi_k$ being the phase of the bath correlation defined by the Fourier filter $\hat{f}(\omega)=\sum_{k \in \mathbb{Z}} c_k e^{i\left(\frac{2 \pi}{\Omega}\right) k \omega}.
$ Using the Pauli identity, we can rewrite the expression of $\sigma_k$ as a sequence of rotation gates  
\begin{equation}
    \sigma_k = R_z(\phi) X  R_z(-\phi)  
\end{equation}
Then the operation for an arbitrary $\theta$ could be realized as
\begin{align}
    \label{eq:arbitrarytheta}
    e^{-i \theta \sigma_k \otimes A}
    &= (R_z(\phi)\otimes I)
    e^{-i\theta X \otimes A}
    (R_z(-\phi)\otimes I) \\\nonumber
    &= (R_z(\phi)H\otimes I)
    e^{-i\theta Z\otimes A}
    (H R_z(-\phi)\otimes I) .
\end{align}

We now consider the term $e^{-i\theta Z\otimes A}$, where the operator $Z_{\text{anc}}\otimes A$ induces opposite rotations on the system depending on the ancilla’s $Z$-eigenvalue:
\begin{equation}
    e^{-i\theta Z\otimes A}
    = |0\rangle\!\langle 0| \otimes e^{-i\theta A}
    +
    |1\rangle\!\langle 1|\otimes e^{+i\theta A}
    \Rightarrow\quad
    (CZ_{a,i})
    e^{-i\theta A} 
    (CZ_{a,i})
\end{equation}
This holds if and only if every Pauli term in $A$ anticommutes with $Z_i$, i.e.
\begin{equation}
C Z_{a, i}\left(P_i\right) C Z_{a, i}=Z_a^{\operatorname{parity}\left(P_i\right)} P_i
\end{equation}
where
\begin{equation}
    \operatorname{parity}\left(P_i\right)= \begin{cases}1 & \text { if } P_i \in\{X, Y\} \\ 0 & \text { if } P_i \in\{I, Z\}.\end{cases}
\end{equation}
In our case, we take $A_a = \{A_{ax}, A_{ay}\}$ as a set of Hermitian operators $A_{ax} = c_a^{\dagger}c_b+h.c$ and $A_{ay} = i(c_a^{\dagger}c_b-h.c)$. These two types of operators create a complete set of Hermitian operators. The qubit representation of these two types are $\frac{1}{2}(X_aX_b+Y_aY_b)$ and $\frac{1}{2}(Y_aX_b-X_aY_b)$, respectively. Note that, even if $A_{a,b}$ is a non-local operator i.e. $a$ and $b$ are not neighboring qubits in the system, the above condition still holds as  
$A_{i j}=\frac{1}{2}\left(X_i S_{i, j} X_j+Y_i S_{i, j} Y_j\right)$, where,
$S_{i, j}=Z_{i+1} \cdots Z_{j-1}$ after JW mapping.
\begin{align}
\label{eq:CZACZ}
    C Z_{a, i} A C Z_{a, i} & =\frac{1}{2} C Z_{a, i}\left(X_i X_j\right) C Z_{a, i}+\frac{1}{2} C Z_{a, i}\left(Y_i Y_j\right) C Z_{a, i} \\\nonumber
    & =\frac{1}{2}\left(Z_a \otimes X_i X_j\right)+\frac{1}{2}\left(Z_a \otimes Y_i Y_j\right) \\\nonumber
    & =Z_a \otimes \frac{1}{2}\left(X_i X_j+Y_i Y_j\right) \\\nonumber
    & =Z_a \otimes A .
\end{align}

Following the identity $Ue^{A}U^{\dagger} = e^{Ue^{A}U^{\dagger}}$ and \cref{eq:CZACZ}, we can write
\begin{equation}
    C Z_{a, i}e^{-i\theta A_i} C Z_{a, i}=e^{C Z_{a, i}e^{-i\theta A_i} C Z_{a, i}} = e^{-i\theta Z_{a,i}A_i}
    \label{eq:CZeCZ}
\end{equation}
The generalized FSim gate 
\begin{equation}
    \operatorname{FSim}(\theta, \Phi)=\left(
    \begin{array}{cccc}
        1 & 0 & 0 & 0 \\
        0 & \cos \theta & -i \sin \theta & 0 \\
        0 & -i \sin \theta & \cos \theta & 0 \\
        0 & 0 & 0 & e^{-i \Phi}
    \end{array}\right)
\end{equation}
when set to $\Phi=0$ implements the evolution under $H_{\text {fsim }}=X_i X_j+Y_i Y_j$ i.e.
\begin{equation}
    \mathrm{FSWAP}_{ij}(\theta)
    \coloneqq i\mathrm{SWAP}_{ij}(\theta)=
    \exp\Bigl[ -i\tfrac{\theta}{2}(X_iX_j + Y_iY_j) \Bigr].
\end{equation}
From \cref{eq:arbitrarytheta,eq:CZeCZ}, we can implement the circuit as following
\begin{equation}
    R_z^{(a)}(\phi)H_a
    \mathrm{CZ}(a,i)i\mathrm{SWAP}_{ij}(\theta_k)\mathrm{CZ}(a,i)
    H_aR_z^{(a)}(-\phi).
\end{equation}
where, $\theta = \sqrt{\tau} |c_k|/2$ shown in~\Cref{fig:A_circuit}. Note that, to realize $A_y$, we just need to apply local $Z$-rotations conjugating $i\mathrm{SWAP}_{ij}(\theta,0)$ and the above equation looks like

\begin{equation}
    \exp\bigl(-i\kappa\sigma_k \otimes A_a^{(Y)}\bigr)
    =
    \bigl(I \otimes R_z^{(j)}(\tfrac{\pi}{2})\bigr)
    \exp\bigl(-i\kappa\sigma_k \otimes A_a^{(X)}\bigr)
    \bigl(I \otimes R_z^{(j)}(-\tfrac{\pi}{2})\bigr).
\end{equation}
Operationally, this corresponds to inserting $R_z^{(j)}(\pi/2)$ on fermionic
mode $j$ immediately before the $i\mathrm{SWAP}_{ij}(\theta_k)$ block and
$R_z^{(j)}(-\pi/2)$ immediately after, leaving the rest of the circuit unchanged.

Finally, we realize that the Hamiltonian evolution from two consecutive modes $k$ and $k+1$ nearly cancels. To see that, recall the filtered jump operator is 
\begin{equation}
    K_{N_\omega}
    =
    \sum_{|k|\le N_\omega} c_k A(-kT),
    \qquad
    T = \frac{2\pi}{\Omega},
\end{equation}
with Heisenberg evolution
\begin{equation}
    A(-kT)
    =
    e^{i H kT} A e^{-i H kT}.
\end{equation}
From \cref{eqH2}, the $k$-th Fourier mode can be written as
\begin{equation}
    U_k
    \coloneqq
    \exp\bigl(-i\kappa_k \sigma_k \otimes A(-kT)\bigr)
    =
    \bigl(I \otimes e^{i H kT}\bigr)
    \exp\bigl(-i\kappa_k \sigma_k \otimes A\bigr)
    \bigl(I \otimes e^{-i H kT}\bigr),
\end{equation}
for some $\kappa_k \propto \sqrt{\tau}|c_k|$.  Taking the ordered product of two consecutive unitaries associated with modes $k$ and $k+1$, respectively, yields

\begin{align}
    U_k U_{k+1}
    &= \bigl(I \otimes e^{i H kT}\bigr)
    e^{-i\kappa_k \sigma_k \otimes A}
    \bigl(I \otimes e^{-i H kT}\bigr)
    \bigl(I \otimes e^{i H (k+1)T}\bigr)
    e^{-i\kappa_{k+1} \sigma_{k+1} \otimes A}
    \bigl(I \otimes e^{-i H (k+1)T}\bigr)
    \nonumber\\
    &= \bigl(I \otimes e^{i H kT}\bigr)
    e^{-i\kappa_k \sigma_k \otimes A}
    \underbrace{\bigl(I \otimes e^{-i H kT} e^{i H (k+1)T}\bigr)}_{I \otimes e^{i H T}}
    e^{-i\kappa_{k+1} \sigma_{k+1} \otimes A}
    \bigl(I \otimes e^{-i H (k+1)T}\bigr).
\end{align}
Thus, between neighboring Fourier modes there remains an effective evolution step ($I \otimes e^{i H T}$) of duration $T$. In our square-wave construction that becomes $I \otimes e^{i H (\pi/2)}$., in the direct correspondence with the discrete-time structure implied by the periodic filter.

\section{Logical Quantum Resources Estimation  for Dissipative Ground State Preparation}
\label{sec:resource_estimation_method}

We estimate the resources for our algorithm in terms of $T$-gates and logical qubits by explicitly describing the required quantum circuit for a single dissipative time-step operator $W(\sqrt{\tau})\approx e^{\tau\mathcal L}$. In the intended continuation workflow, a warm start is available at a tractable geometry on the aligned path, and the dissipative map is applied only as needed to restore low-energy weight at each subsequent geometry. Thus the total logical cost is obtained by multiplying the per-step cost of $W(\sqrt{\tau})$ by the number of cooling steps $N_T(s_i)$ required at each geometry and summing over path points.

We use a second-order Strang-splitting implementation of $W(\sqrt{\tau})$ with a square-wave filter truncated to $N_\omega$ Fourier modes. The choice of $N_\omega$ depends on the target spectral resolution and on the relevant gap $\Delta(s)$, so the most demanding geometry typically occurs near the TS region.
Thus, for each geometry $s$, the spectral gap of $H(s)$ is different, and therefore we need the most number of Fourier modes $N_{\omega}$ near the transition state. For systems in the milliHartree gap regime $\Delta \sim 0.001$, it suffices to take $N_\omega \geq 1000$ to achieve a sufficiently low leakage bound. Depending on the choice of $N_{\omega}$, we alternatively repeat the implementation of dilated jump operator $A(\sqrt{\tau})$ and the Hamiltonian simulation circuit according to the \cref{eq:W_ckt}. In the $A(\sqrt{\tau})$ block, the $T$-cost arises from the implementation of $i\mathrm{SWAP}(-2\theta/\pi)$. For a precision $\epsilon \approx 10^{-5}$, the rough count of $T$-counts is about 50. To implement each real-time evolution $e^{\pm iHt}$ using quantum singular value transformation (QSVT), which requires $N_{U_H}$ calls  the block-encoding unitary $U_H$ per $W(\sqrt{\tau})$.
\begin{table*}[h]
    \label{fig:resource_estimate_table}
    \centering
    \begin{tabular}{lccc}
        \hline
        System 
        & Qubits 
        & Toffolis for $W(\sqrt{\tau})$
        & Toffolis for QPE \\
        \hline
        Fe$_2$S$_2$ (30e, 20o)           
        & 466  & $2.56 \times 10^8$    & $3.97 \times 10^7$  \\
        Fe$_4$S$_4$ (54e, 36o)           
        & 873  & $1.23 \times 10^9$    & $1.72 \times 10^8$  \\
        FeMoco (54e, 54o)                
        & 1137 & $1.97 \times 10^9$    & $3.41 \times 10^8$  \\
        FeMoco (113e, 76o)               
        & 1459 & $8.41 \times 10^9$    & $9.99 \times 10^8$  \\
        CPD1-P450X (63e, 58o)           
        & 1150 & $3.02 \times 10^9$    & $4.91 \times 10^8$  \\
        CO$_2$[XVIII] (64e, 56o)         
        & 924  & $1.42 \times 10^9$    & $2.05 \times 10^8$  \\
        CO$_2$[XVIII] (100e, 100o)       
        & 1960 & $9.00 \times 10^9$    & $1.06 \times 10^9$  \\
        CO$_2$[XVIII] (150e, 150o)       
        & 2870 & $2.92 \times 10^{10}$ & $2.81 \times 10^9$  \\
        \hline
    \end{tabular}
    \caption{Qubit counts and Toffoli costs per single dissipative time step $e^{\tau \mathcal{L}(s)}$, implemented by a second-order Strang splitting operator $W(\sqrt{\tau})$ with $N_\omega=1000$ Fourier modes for the filter function, together with QPE Toffoli counts (for $\sigma_{\mathrm{PEA}} = 1$~mHa) using DFTHC+BLISS+SA block encodings for systems~\cite{Low_2025}.}
    \label{tab:block-enc-lindblad-qpe}
\end{table*}

\section{Continuation-path optimization}
\label{sec:appendix_path_optimization}

This appendix records a simple but useful fact: the Davis--Kahan integral
\[
    C_{\mathrm{DK}}
    =
    \int_{0}^{1}\frac{\|\partial_s H(s)\|}{\Delta(s)}\,ds
\]
is controlled by a weighted path-length functional of the nuclear trajectory.  This provides a direct route for connecting a continuation optimization objective (such as the one implemented in our N$_3$ example) to reductions in $C_{\mathrm{DK}}$, and therefore to reductions in the discretization cost $N_H$ in \cref{thm:reaction-path-runtime}.

\begin{lemma}[Chain-rule bound for the DK integrand]
    \label{lem:dk_chain_rule}
    Let $\mathbf{R}:[0,1]\to\mathbb{R}^{3N}$ be a $C^1$ continuation path and define $H(s)=H_{\mathrm{el}}(\mathbf{R}(s))$.
    Assume $H_{\mathrm{el}}$ is Fr\'echet differentiable in $\mathbf{R}$ along the image of the path, with derivative
    $D_{\mathbf{R}}H_{\mathrm{el}}(\mathbf{R}) : \mathbb{R}^{3N}\to \mathsf{Herm}$.
    Define the pointwise operator norm of the Fr\'echet derivative
    \[
        L_H(\mathbf{R})
        \coloneqq
        \sup_{\|v\|=1}\bigl\|D_{\mathbf{R}}H_{\mathrm{el}}(\mathbf{R})[v]\bigr\|.
    \]
    Then for all $s\in[0,1]$,
    \begin{equation}
    \label{eq:dk_pointwise_weighted_speed}
        \frac{\|\partial_s H(s)\|}{\Delta(s)}
        \le
        \frac{L_H(\mathbf{R}(s))}{\Delta(s)}\,\|\mathbf{R}'(s)\|.
    \end{equation}
    Consequently,
    \begin{equation}
    \label{eq:cdk_weighted_length2}
        C_{\mathrm{DK}}
        \le
        \int_{0}^{1}\frac{L_H(\mathbf{R}(s))}{\Delta(s)}\,\|\mathbf{R}'(s)\|\,ds.
    \end{equation}
\end{lemma}

\begin{proof}
Let $s\in[0,1]$ be fixed.  By the chain rule in Banach spaces,
\[
    \partial_s H(s)
    =
    D_{\mathbf{R}}H_{\mathrm{el}}(\mathbf{R}(s))[\mathbf{R}'(s)].
\]
Using the definition of $L_H(\mathbf{R}(s))$ and homogeneity in the direction argument,
\[
    \|\partial_s H(s)\|
    =
    \bigl\|D_{\mathbf{R}}H_{\mathrm{el}}(\mathbf{R}(s))[\mathbf{R}'(s)]\bigr\|
    \le
    L_H(\mathbf{R}(s))\,\|\mathbf{R}'(s)\|.
\]
Dividing by $\Delta(s)>0$ gives \cref{eq:dk_pointwise_weighted_speed}.  Integrating both sides from $0$ to $1$ yields \cref{eq:cdk_weighted_length2}.
\end{proof}

\begin{corollary}[Uniform bound and a clean geometry surrogate]
\label{cor:cdk_uniform_surrogate}
    Suppose that along the path image one has
    \[
        L_H(\mathbf{R}(s))\le L_{\max} := \sup_{s\in [0,1]} L_H(\mathbf{R}(s))
        \qquad\text{and}\qquad
        \Delta(s)\ge \Delta_{\min}>0
        \qquad \text{for all } s\in[0,1].
    \]
    Then
    \begin{equation}
    \label{eq:cdk_uniform_length_bound}
        C_{\mathrm{DK}}
        \le
        \frac{L_{\max}}{\Delta_{\min}}
        \int_{0}^{1}\|\mathbf{R}'(s)\|\,ds.
    \end{equation}
    In particular, if $\mathbf{R}$ is parameterized by arc-length (so $\|\mathbf{R}'(s)\|$ is constant), reducing the path length
    $\int_0^1\|\mathbf{R}'(s)\|ds$ reduces this uniform upper bound on $C_{\mathrm{DK}}$ by the same factor. The exact $C_{\mathrm{DK}}$ is controlled by the weighted length in~\cref{eq:cdk_weighted_length2}, so unweighted length minimization should be viewed as a coarse surrogate rather than a guarantee.
\end{corollary}

\begin{proof}
Apply \cref{eq:cdk_weighted_length2} and upper bound the numerator and lower bound the denominator uniformly:
\[
    C_{\mathrm{DK}}
    \le
    \int_0^1 \frac{L_H(\mathbf{R}(s))}{\Delta(s)}\|\mathbf{R}'(s)\|ds
    \le
    \frac{L_{\max}}{\Delta_{\min}}\int_0^1\|\mathbf{R}'(s)\|ds.
\]
\end{proof}

\section{Dissipative Continuation Algorithm Pseudocode}
\label{sec:pseudocode}

\begin{algorithm}
\caption{Dissipative continuation along a classically supplied continuation path}
\label{alg:dissipative_continuation}

\KwIn{
    Classically supplied TS candidate $\mathbf R_{\rm TS}$, tractable warm-start geometry $\mathbf R_0$, continuation path $\mathbf R(s)$ connecting $\mathbf R_0$ to $\mathbf R_{\rm TS}$, target energy error $\epsilon_E$, jump set $\mathcal A$, filter function $\hat f(\omega)$
}
\KwOut{
Approximate transition-state ground state $\rho^\star$ and energy estimate $E^\star$
}

Use the classically supplied TS candidate, warm-start geometry, and continuation path as input to the quantum state-preparation routine\;

Optionally smooth, reparameterize, or locally deform the path by minimizing
\[
J(\mathbf R)
=
\lambda_{\mathrm{len}}\!\int_0^1 \|\mathbf R'(s)\|\,ds
+
\lambda_{\mathrm{curv}}\!\int_0^1 \|\mathbf R''(s)\|^2\,ds
+
\lambda_{\mathrm{CI}}\!\int_0^1 \Phi_{\mathrm{CI}}(\mathbf R(s))\,ds
\]
to obtain a continuation path $\mathbf R(s)$ that remains chemically validated while avoiding sharp projector rotation and known small-gap regions\;

Discretize the continuation path into $\{s_i\}_{i=0}^{N_H-1}$ with $s_0$ denoting the tractable warm-start geometry and $s_{N_H-1}=s^\star$ denoting the target TS geometry\;

For each $s_i$, build the electronic Hamiltonian
\[
H_i = H_{\mathrm{el}}(\mathbf R(s_i))
\]
and embed/align neighboring active orbital spaces to maintain a smooth orbital gauge along the continuation path\;

Prepare an approximate ground state $\rho_0$ for $H_0$ using a tractable classical or hybrid method (e.g., DMRG + dissipative cooling) \;

\For{$i=0$ \KwTo $N_H-2$}{
    Construct filtered jump operators $\{K_a^{(i+1)}\}_{a\in\mathcal A}$ for $H_{i+1}$\;

    Starting from $\rho_i$, apply the dissipative primitive for $t_i$ steps to obtain
    \[
    \rho_{i+1}
    \approx
    \left(e^{\tau \mathcal L_{i+1}}\right)^{t_i}(\rho_i),
    \]
    where
    \[
    \mathcal L_{i+1}[\rho]
    =
    \sum_{a\in\mathcal A}
    \left(
    K_a^{(i+1)}\rho K_a^{(i+1)\dagger}
    -
    \frac12
    \left\{
    K_a^{(i+1)\dagger}K_a^{(i+1)},\rho
    \right\}
    \right).
    \]
}

Set $\rho^\star \leftarrow \rho_{N_H-1}$ and estimate
\[
E^\star = \Tr(H_{N_H-1}\rho^\star).
\]

\Return{$(\rho^\star,E^\star)$}\;

\end{algorithm}

\bibliography{references}
\end{document}